\documentclass[sort&compress,preprint,12pt]{elsarticle}

\usepackage[T1]{fontenc}
\usepackage[utf8]{inputenc}
\usepackage{hyperref}
\hypersetup{
  colorlinks=true,
  linkcolor=blue,
  anchorcolor=blue,
  citecolor=blue,
  urlcolor=blue,
  filecolor=blue,
  breaklinks=true
}

\usepackage{float}
\floatstyle{boxed}
\restylefloat{figure}
\usepackage{multirow}
\usepackage{amssymb}
\usepackage{amsmath}
\usepackage{amsthm}
\usepackage{needspace}
\usepackage{proof}
\usepackage{tikz}
\usetikzlibrary{cd}
\usepackage[only=llbracket,rrbracket]{stmaryrd}
\SetSymbolFont{stmry}{bold}{U}{stmry}{m}{n}
\usepackage{fontawesome}
\usepackage{xspace}
\usepackage{cmll}

\theoremstyle{plain}
\newtheorem{theorem}{Theorem}[section]
\newtheorem{lemma}[theorem]{Lemma}
\newtheorem{corollary}[theorem]{Corollary}

\theoremstyle{definition}
\newtheorem{definition}[theorem]{Definition}

\theoremstyle{remark}
\newtheorem{remark}[theorem]{Remark}
\newtheorem{example}[theorem]{Example}
\newtheorem{notation}[theorem]{Notation}

\newcommand\pelimsupof[7]{\elimsup[#1 #2](#3,#4.#5, #6.#7)}

\newcommand\spdv[1][t]{\mathcal M_{#1}}

\newcommand\irule[3]{\infer[{#3}]{#2}{#1}}
\newcommand\escalar[1]{\mathsf #1}
\newcommand\one{\ensuremath{\mathfrak 1}}
\newcommand\zero{\ensuremath{\mathfrak 0}}
\newcommand\OC[1][\mathcal S]{$\mathcal L\odot^{{#1}\escalar p}$-calculus\xspace}
\newcommand\OCl{$\mathcal L\odot^{\mathcal S\escalar p}$-logic\xspace}

\newcommand\recap[3]{\noindent {\bf #1 \ref{#2}.}\emph{#3}}
\newcommand\xrecap[4]{\noindent {\bf #1 \ref{#3} (#2)}{\bf .} \emph{#4}}
\newcommand\SM{\ensuremath{{\mathbf C}_{\mathcal S}}}
\newcommand\Rel{\ensuremath{\mathbf{Rel}}}

\newcommand\plus{\mathbin{\text{\normalfont\scalebox{0.7}{\faPlus}}}}
\newcommand\pair[2]{\langle #1, #2 \rangle}
\newcommand\super[2]{[#1, #2]}
\newcommand\inl{\mathsf{inl}}
\newcommand\inr{\mathsf{inr}}
\newcommand\elimtens{\delta_{\otimes}}
\newcommand\elimone{\delta_{\one}}
\newcommand\elimzero{\delta_{\zero}}
\newcommand\elimwith{\delta_{\with}}
\newcommand\elimoplus{\delta_{\oplus}}
\newcommand\elimsup[1][{\escalar p}{\escalar q}]{\delta^{#1}_{\odot}}

\newcommand\relimsup[1]{\delta_{\odot}^{#1}}
\newcommand\fstsup{\pi_1^\odot}
\newcommand\sndsup{\pi_2^\odot}
\newcommand\fst{\pi_1}
\newcommand\snd{\pi_2}
\newcommand\home[2]{\left[{#1}\to{#2}\right]} 
\newcommand\Hom[2]{\mathsf{Hom}({#1},{#2})} 
\newcommand\sem[1]{\left\llbracket {#1}\right\rrbracket}
\newcommand\semS[1]{\llparenthesis {#1}\rrparenthesis}
\newcommand\Id{\mathsf{id}}
\newcommand\lra[1][1_{\mathcal S}]{\longrightarrow_{#1}}
\newcommand\xlra[1]{\xrightarrow{#1}}
\newcommand\mas[1][\mathcal S]{\mathbin{\plus_{\!\!{}_{#1}}}}
\newcommand\masS[1][\mathcal S]{\mathbin{+_{\!\!{}_{#1}}}}
\newcommand\produ[1][\mathcal S]{\mathbin{\cdot_{\!{}_{#1}}}}
\newcommand\scal[1][\mathcal S]{\mathbin{\bullet_{\!{}_{#1}}}}
\newcommand\coproducto[2]{{[}#1,#2{]}}
\newcommand\we{\ensuremath{\mathsf{W}}}
\newcommand\sumwe[2]{\ensuremath{\nabla_{\!\!#1#2}}}

\journal{arXiv}

\begin{document}

\begin{frontmatter}
  \title{The Sup Connective in IMALL:\\ A Categorical Semantics}

  \author[Inria,UNQ]{\texorpdfstring{Alejandro Díaz-Caro\fnref{epiq,fundingboth}}{Alejandro Díaz-Caro}}
  \ead{alejandro.diaz-caro@inria.fr}

  \author[IMERL]{\texorpdfstring{Octavio Malherbe\fnref{fundingboth}}{Octavio Malherbe}}
  \ead{malherbe@fing.edu.uy}

  \affiliation[Inria]{
    organization={Université de Lorraine, CNRS, Inria, LORIA},
    city={Nancy},
    country={France}
  }
  \affiliation[UNQ]{
    organization={DCyT, Universidad Nacional de Quilmes},
    city={Bernal},
    state={Buenos Aires},
    country={Argentina}
  }
  \affiliation[IMERL]{
    organization={IMERL, FIng, Universidad de la República},
    city={Montevideo},
    country={Uruguay}
  }
\fntext[epiq]{Supported by the PEPR integrated project EPiQ (ANR-22-PETQ-0007).}
\fntext[fundingboth]{Supported by the European Union through the MSCA-SE project QCOMICAL (Grant Agreement ID: 101182520) and by the Uruguayan project ANII FCE-1-2025-1-186751.}

  \begin{abstract}
    We explore a proof language for intuitionistic multiplicative additive linear
    logic, incorporating the sup connective that introduces additive pairs with a
    probabilistic elimination, and sum and scalar products within the proof-terms.
    We provide an abstract characterisation of the language, revealing that any
    symmetric monoidal closed category with biproducts and a monomorphism from the
    semiring of scalars to the semiring $\Hom II$ is suitable for the job.
    Leveraging the binary biproducts, we define a weighted codiagonal map which is at the
    core of the sup connective.
  \end{abstract}


\end{frontmatter}

\section{Introduction}\label{sec:intro}
\subsection{Historical origins}
In the quest for a logic for quantum computing, the non-cloning principle~\cite{Nocloning} is one of the challenges to tackle. This principle states that it is impossible to create an identical copy of an arbitrary unknown quantum state. This is a consequence of the linearity of the quantum mechanics operators, which is a fundamental principle of quantum computing. 
However, the first step into considering this linearity is to have a language where that linearity can be expressed.
With this aim, calculi with sums and scalar product in the proof-terms has been used for quantum computing and algebraic lambda-calculi on many occasions~\cite{AltenkirchGrattageLICS05,SelingerValironMSCS06,Vaux2009,AssafDiazcaroPerdrixTassonValironLMCS14,ArrighiDiazcaroLMCS12,ZorziMSCS16,Lineal,ArrighiDiazcaroValironIC17,LambdaS,DiazcaroGuillermoMiquelValironLICS19,DiazcaroMalherbeLSFA18,DiazcaroMalherbeACS20,DiazcaroMalherbeLMCS22}.
The idea is that if $t$ and $u$ are proofs of the same proposition $A$, then $t\plus
u$ and $\escalar s\bullet t$ are also proofs of $A$, with $\escalar
s$ in some set of scalars.  Most of these works consider a call-by-value
strategy for the reduction of the terms, which forces a kind of linearity by
considering the reduction rules $t(u\plus v)\longrightarrow tu\plus tv$ and $t(s\bullet u)\longrightarrow s\bullet tu$, when $u$ and $v$ are values.

In~\cite{DiazcaroDowekMSCS24} the approach to have linearity in the
proof-language is different. There is no need to define a reduction strategy. Instead, the logic considered is Intuitionistic
Multiplicative Additive Linear Logic (IMALL), and in the proof language, there is one
proof of the proposition $\one$ (the multiplicative truth) as elements of a semiring of scalars
$\mathcal S$. Then, the proofs
of $\bigwith_{i=1}^n\one$ (for any bracketing) are in
one-to-one correspondence with the elements of $\mathcal
S^n$. In such a calculus any closed
proof $t$ of the proposition $A\multimap B$ is proved to be linear in
the syntactic sense. That is, the proof $t(u\plus v)$ is proof-equivalent to
$tu\plus tv$ and $t(s\bullet u)$ is proof-equivalent to $s\bullet
tu$. Moreover, any $\mathcal
S$-homomorphism $\mathcal S^n\xlra f\mathcal S^m$ has a representation as a
proof-term of the proposition $\bigwith_{i=1}^n\one\multimap\bigwith_{i=1}^m\one$.

The proof language in question is the $\mathcal L^{\mathcal S}$-calculus.
It is a proof language for IMALL, which contains sums and scalar products as proof constructors, but whose provable formulae are nothing more---and nothing less---than the tautologies of IMALL.

A second challenge of a proof-language for quantum computing is the
non-determinism of the measurement. In~\cite{DiazcaroDowekTCS23},
non-determinism has been treated as a new connective in Intuitionistic
Propositional Logic, in a Natural Deduction presentation. This connective,
$\odot$ (read as ``sup'' for superposition), is introduced to express the
superposition of data and, more importantly, the measurement operation. The
connective sup arises from the observation that a superposition behaves as a
conjunction, where both propositions are true (and so, its proof is the pair of
proofs), but also, when measured, it behaves as a disjunction, where only one
proposition will be recovered in a non-deterministic process. The
$\odot^{\mathcal S}$-calculus contains the sup connective, and also sums and
scalar products. While not enforcing linearity (thus allowing cloning), it
allows the encoding of basic quantum lambda calculus.

The sup connective then has the introductions and eliminations of conjunction, but it goes further by also including one extra elimination rule, 
that of the disjunction. This elimination is, in
fact, derivable in Natural Deduction when sup is replaced with a conjunction,
but the derivation is not unique, thus enabling non-determinism.
The $\odot^{\mathcal S}$-calculus shows that superposition and measurement can
be represented by this new connective. 

In~\cite{DiazcaroDowekMSCS24}, alongside the $\mathcal L^{\mathcal S}$-calculus,
the $\mathcal L\odot^{\mathcal S}$-calculus is also considered, which
incorporates the sup connective within the linear setting. This distinguishes it
from other approaches to non-deterministic and probabilistic linear calculus,
such as PCF$^\mathcal R$~\cite{LairdManzonettoMcCuskerPaganiLICS13}, where
the non-deterministic reduction arises from terms like $t_1\,\mathtt{or}\,t_2$ with
$t_1$ and $t_2$ of the same type. In contrast, the $\mathcal L\odot^{\mathcal
S}$-calculus introduces non-deterministic reduction as a pair destructor:
$\fstsup$ and $\sndsup$ serve as deterministic pair destructors, while
$\delta_\odot$ is non-deterministic. That is, $\delta_\odot
(\super{t_1}{t_2},x.{s_1},y.{s_2})$ reduces to either $(t_1/x)s_1$ or
$(t_2/y)s_2$. Consequently, the non-deterministic behaviour is explicit in its
elimination and is not triggered by an introduction term. This approach also
allows for a choice among elements of different types.

In the present paper, our aim is to provide an abstract categorical characterisation for a
proof-language of IMALL with $\odot$.  IMALL with $\odot$ is essentially IMALL,
as the sup connective can be regarded as an additive conjunction, with an extra
rule that is derivable by more than one deduction tree---resulting in
non-determinism. Further technical details are presented in
Remark~\ref{rmk:sup-e-logic}, following the presentation of the deduction rules.

PCF$^{\mathcal R}$~\cite{LairdManzonettoMcCuskerPaganiLICS13} not only
addresses non-determinism, with its $\mathtt{or}$ constructor, but also the
probabilistic choice, with the $\bullet$ constructor.  Hence, $(\escalar
p\bullet t_1)\,\mathtt{or}\,(\escalar q\bullet t_2)$ expresses the
probabilistic choice between $t_1$ and $t_2$, with probabilities $\escalar p$
and $\escalar q$ respectively. In fact, it is slightly more general than a
probabilistic choice since the scalars belong to the continuous semiring
$\mathcal R$. In the case of $\mathcal R = \mathbb R^{\geq 0}$, it is a proper
probabilistic calculus. We will refer to this as ``generalised probabilistic
choice''.

We generalise the $\mathcal L\odot^{\mathcal S}$-calculus to the \OC, where,
instead of considering the non-deterministic destructor $\delta_\odot$, we
employ a (generalised) probabilistic destructor $\elimsup$, with $\escalar p$
and $\escalar q$ scalars in the semiring $\mathcal S$ summing to one. What
PCF$^{\mathcal R}$ expresses as $(\escalar p\bullet t_1)\,\mathtt{or}\,(\escalar
q\bullet t_2)$ can be written in the \OC as
$\elimsup(\super{t_1}{t_2},x.x,y.y)$. Nonetheless, we can also write the term
$(\escalar p\bullet t_1)\plus(\escalar q\bullet t_2)$, which carries the same
denotational interpretation but does not reduce probabilistically. Instead, it
represents a linear combination of terms, enabling us to express linear
functions (matrices) and vectors. In this sense, the \OC uses the sums and
scalar product provided by its model not only to represent probabilistic
reductions but also to denote sums and scalar products within the proof
language. Note, however, that the elimination rule of $\sup$ in the \OC
explicitly carries the probabilities in the operator, and therefore cannot
encode quantum measurement. For modelling measurement one requires a semantics
based on density matrices, which is outside the scope of this paper (see~\cite{DaveDiazcaroZamdzhievAPLAS25} for a recent proposal in this direction).

\subsection{Modelling the sup connective}
Introducing a (generalised) probabilistic operator to a linear language is not
straightforward. We begin the informal analysis of this section with the concrete category
$\mathbf{SM}_{\mathcal S}$ of semimodules over the semiring $\mathcal S$ and linear maps, as a
means to aid intuition. Such a category is one of the concrete construction
examples we will use throughout the paper.

Our interpretation does not rely on the Powerset Monad, which is commonly used to express non-deterministic effects~\cite{MoggiIC91}, because this monad is not compatible with the structure of our category. Specifically, using the Powerset Monad would require forming sets from the Cartesian product of non-deterministic paths. However, the map \( A \times A \xlra{\xi} \mathcal PA \), where \( \xi(a_1, a_2) = \{a_1, a_2\} \), is not linear. Since linearity is required in our categorical setting, such a map cannot be part of the category.

Our approach is instead inspired by the density matrix quantum formalism (see,
for example,~\cite[Section 2.4]{NC}), wherein we consider the linear combination
of results as a representation of a probability distribution.  Let $t$ be a term
reducing with probability $p$ to $t_1$ and a probability of $q$ to $t_2$, with
$p+q=1$.  
We
interpret $t$ as 
\(
 \sumwe pq(t_1,t_2) = p\scal[A]t_1\mas[A]q\scal[A]t_2
\), where, if $\hat p$ is the mapping that multiplies its argument by $p$, then
$\sumwe pq$ is defined as $\coproducto{\hat p}{\hat q}$, that is
\[
  \begin{tikzcd}[labels=description,column sep=2cm,row sep=1cm]
    {A}\ar[r,"i_1"]\ar[dr,"\hat p"] & A+A\ar[d,"\sumwe pq"] & {A} \ar[dl,"\hat q"]\ar[l,"i_2"]\\
    &  A &
  \end{tikzcd}
\]

This approach is close to that used for PCF$^{\mathcal R}$ in~\cite{LairdManzonettoMcCuskerPaganiLICS13}. 

In an abstract categorical framework, we require at least a category with biproducts to interpret \(\sumwe pq\) as \(\nabla \circ (\hat p \oplus\hat q)\), where \(\hat p\) and \(\hat q\) are suitable maps from \(A\) to \(A\). These scalar maps make it necessary for the category to also be monoidal, allowing us to define a semiring of scalars in \(\Hom II\) \cite{KellyLaplaza80}, where \(I\) is the tensor unit.
Next, we define a monomorphism \(\semS{\cdot} : \mathcal{S} \to \Hom II\), which guarantees that if two proof terms are mapped to the same morphism in \(\text{Hom}(I, I)\), then they are considered equivalent in the categorical sense.

\subsection{Related works}\label{sec:related}
The probabilistic choice in linear logic has been studied in many settings. 

\paragraph{Compact closed categories}
In~\cite{AbramskyCoeckeLICS04}, the authors proposed a categorical semantics of
quantum protocols using symmetric monoidal closed categories with biproducts,
which are also compact. The compactness property provides a notion of dagger,
which gives a natural definition of measurements in terms of the {\em Born rule}
in quantum mechanics. Thus, the main difference between our presentation for a
model of IMALL+$\odot$ and their presentation for a model of quantum protocols
is their reliance on a dagger operator and their use of the compactness property
for this purpose.  Remark~\ref{rmk:NoCompactNeeded} illustrates that some
properties would be significantly easier to prove if the category were compact
closed.  However, assuming compactness would limit the generality of the
results.

\paragraph{Probabilistic coherence spaces} 
In~\cite{DanosEhrhardIC11}, based on an idea from Girard~\cite{Girard04}, the
authors proposed a model of linear logic using probabilistic coherence spaces,
interpreting types through continuous domains. Morphisms in the associated
category are Scott-continuous.  Additionally, they provide a probabilistic
interpretation of terms, extending PCF with a probabilistic choice construction
which selects a natural number from a probability distribution.  They show the
denotational semantics of closed terms in their base type as sub-probability
distributions.

\paragraph{Cones}
In~\cite{SelingerQPL04}, the author employed the concept of normed cones to
provide an interpretation for the probabilities inherent in quantum programming.
An abstract cone is analogous to an $\mathbb R$-vector space, except that scalars
are drawn from the set of non-negative real numbers.  This idea has been further
developed in~\cite{EhrhardPaganiTasonPOPL17}, and then proved to be a model of
intuitionistic linear logic in~\cite{EhrhardLICS20}. In addition, it is
proved~\cite{CrubilleLICS18} that this model is a conservative extension of the
probabilistic coherence spaces.

\paragraph{PCF$^{\mathcal R}$}
In~\cite{LairdManzonettoMcCuskerPaganiLICS13}, the authors proposed a model of
PCF$^{\mathcal R}$---that is, PCF with a probabilistic choice operator---based
on the category of weighted relations. The first main difference with our
approach is that PCF$^{\mathcal R}$ introduces a probabilistic choice operator,
whereas our system employs a probabilistic pair destructor, as mentioned in the
previous sections. A second difference is that their model is concrete, given in
the category of matrices over a continuous semiring, while ours is formulated at
an abstract categorical level. They also consider a fixed-point operator, which
is outside the scope of this paper.

A more general categorical semantics of PCF$^{\mathcal R}$ was
later developed by Laird~\cite{LairdLiCS16}, in the setting of
symmetric monoidal closed categories with biproducts. From the
semantic point of view, this framework is essentially the same as
ours: both rely on biproducts and scalars from a semiring to interpret
probabilistic choice. The key difference lies in the set of connectives
considered. PCF$^{\mathcal R}$ incorporates specific constructs for
probabilistic choice and scalar multiplication, whereas our system
deals with the full set of connectives of IMALL, with $\odot$ playing
a central role. In this way, probabilistic reasoning is internalised
within the proof system and uniformly integrated with the other
connectives, instead of being restricted to a dedicated operator at
the term level.

A crucial syntactic difference between our approach and PCF$^{\mathcal R}$ is
the way probabilistic choice is expressed. As mentioned before, in PCF$^{\mathcal R}$ one writes
$(p \cdot t_1) \ \mathtt{or}\ (q \cdot t_2)$, while in our system the corresponding
construction is $\elimsup(\super{t_1}{t_2},x.x,y.y)$. Both terms have the same
denotation (a distribution over $t_1$ and $t_2$), but their logical roles differ.
In PCF$^{\mathcal R}$, the operator ``or'' is a primitive construct of the language,
whereas in our calculus it arises from the elimination of the connective $\odot$.
This shows that $\odot$ offers a uniform logical account of probabilistic choice,
integrated with the other connectives of IMALL.

\subsection{Contents of the paper}
In Section~\ref{sec:calculus}, we introduce the \OC, detailing its grammars,
deduction and reduction rules, 
and its correctness properties.

In Section~\ref{sec:examples}, we show how to use it to encode matrices and vectors, and give some concrete examples of how to encode the probabilistic choice.

In Section~\ref{sec:seminodules}, we introduce the categorical construction
together with some specific maps, such as $\sumwe pq$ and $\hat p$, which are
fundamental to interpreting the language.

Section~\ref{sec:semantics} is dedicated to providing the denotational semantics
of the \OC within the category just defined, and establishing its soundness and
adequacy proofs.

Finally, in Section~\ref{sec:conclusion}, we offer some concluding remarks.

\section{The \texorpdfstring{\OC}{L-sup-calculus}}\label{sec:calculus}
\subsection{Grammars}
\begin{definition}
  [Propositions of the \OCl]
  \label{def:props}
  The propositions of the \OCl are those of IMALL with $\odot$.
  \begin{align*}
    A = &\ \one\mid A\otimes A\mid A\multimap A & \text{multiplicative}\\
    \mid&\ \top\mid\zero\mid A\with A\mid A\oplus A\mid A\odot A & \text{additive}
  \end{align*}
\end{definition}

\begin{remark}
	In intuitionistic linear logic there is no multiplicative
  falsehood ($\bot$), multiplicative disjunction ($\parr$), nor additive
  implication ($\Rightarrow$). 
\end{remark}

\begin{definition}[Proof-terms of the \OC]
		\label{def:proof-terms}
  The proof-terms of the \OC are those produced by the following grammar, where
    $x\in\mathsf{Vars}$, an infinite set of variables, 
    $\mathcal S$ is a fixed semiring,
    $\escalar s, \escalar p,\escalar q\in\mathcal S$, and
    $\escalar p\masS \escalar q=1_{\mathcal S}$.
  \[
    \begin{array}{r@{}ll@{}lc}
      &\textrm{introductions} & \textrm{eliminations} && \textrm{connective}\\
      t = x \mid t\plus t\mid \escalar s\bullet t  \\
      &\mid \escalar s.\star & \mid \elimone(t,t)  && \one\\
      &\mid\lambda x.t & \mid tt && \multimap\\
      &\mid t\otimes t & \mid\elimtens(t,xy.t) &&\otimes\\
      &\mid\langle\rangle & && \top\\
      & &\mid\elimzero(t) && \zero\\
      &\mid\pair tt & \mid\fst(t)\mid\snd(t) && \with\\
      &\mid \inl(t)\mid \inr(t) & \mid\elimoplus(t,x.t,y.t) && \oplus\\
      &\mid \super tt & \mid\fstsup(t)\mid\sndsup(t)&\,\mid\elimsup(t,x.t,y.t) & \odot
    \end{array}
  \]
\end{definition}

The substitution of $x$ by $u$ in $t$ is written $(u/x)t$.

\begin{definition}
  [Proof-term context]
  \label{def:term-context}
  We let $K$ be a proof-term with a distinguished variable $[\cdot]$.
  We write $K[t]$ for $(t/[\cdot])K$, that is, the substitution of $[\cdot]$ by $t$ in $K$.
\end{definition}

\subsection{Deduction rules}
\label{sec:deductionRules}
The deduction rules are given in Figure~\ref{fig:typingrules}.  They include the
standard rules of IMALL, plus the extra rules for $\plus$, $\bullet$, and $\odot$.

\begin{remark}
  \label{rmk:sup-e-logic}
  Rules $\odot_i$, $\odot_{e1}$, and $\odot_{e_2}$ coincide with $\with_i$, $\with_{e1}$, and $\with_{e2}$. If we use those rules instead, the extra rule $\odot_e$ could be derivable in IMALL as follows:
  \[
    \infer[\oplus_e,]{\Gamma,\Theta\vdash C}
    {
      \infer[\oplus_i]{\Gamma\vdash A\oplus B}{\infer[\with_{e1}]{\Gamma\vdash A}{\Gamma\vdash A\with B}}
      &
      A,\Theta\vdash C & B,\Theta\vdash C
    }
  \]
  or, similarly
  \[
    \infer[\oplus_e.]{\Gamma,\Theta\vdash C}
    {
      \infer[\oplus_i]{\Gamma\vdash A\oplus B}{\infer[\with_{e2}]{\Gamma\vdash B}{\Gamma\vdash A\with B}}
      &
      A,\Theta\vdash C & B,\Theta\vdash C
		}
  \]
  The goal of having $\odot$ instead of just these two derivations is that these two have a
  deterministic cut-elimination, while $\odot_e$ makes a non-deterministic choice
  between the two.
\end{remark}

\begin{figure}[t]
  \[
    \infer[\mbox{ax}]{x:A\vdash x:A}{}
    \qquad
    \qquad
    \infer[\plus]{\Gamma\vdash t\plus u:A}{\Gamma\vdash t:A & \Gamma\vdash u:A}
    \qquad
    \infer[\bullet(\escalar s)]{\Gamma\vdash\escalar s\bullet t:A}{\Gamma\vdash t:A}
  \]
  \[
    \infer[\one_i]{\vdash \escalar s.\star:\one}{}
    \qquad
    \infer[\one_e]{\Gamma,\Theta\vdash\elimone(t,u):A}{\Gamma\vdash t:\one & \Theta\vdash u:A}
  \]
  \[
    \infer[\otimes_i]{\Gamma,\Theta\vdash t\otimes u:A\otimes B}{\Gamma\vdash t:A & \Theta\vdash u:B}
    \qquad
    \infer[\otimes_e]{\Gamma,\Theta\vdash\elimtens(t,xy.u):C}
    {\Gamma\vdash t:A\otimes B & \Theta,x:A,y:B\vdash u:C}
  \]
  \[
    \infer[\multimap_i]{\Gamma\vdash\lambda x.t:A\multimap B}{\Gamma,x:A\vdash t:B}
    \qquad
    \infer[\multimap_e]{\Gamma,\Theta\vdash tu:B}{\Gamma\vdash t:A\multimap B & \Theta\vdash u:A}
  \]
  \[
    \infer[\top_i]{\Gamma\vdash\langle\rangle:\top}{}
    \qquad
    \infer[\zero_e]{\Gamma,\Theta\vdash\elimzero(t):C}{\Gamma\vdash t:\zero}
  \]
  \[
    \infer[\with_i]{\Gamma\vdash\pair tu:A\with B}{\Gamma\vdash t:A & \Gamma\vdash u:B}
    \qquad
    \infer[\with_{e1}]{\Gamma\vdash\fst(t):A}{\Gamma\vdash t:A\with B}
    \qquad
    \infer[\with_{e2}]{\Gamma\vdash\snd(t):B}{\Gamma\vdash t:A\with B}
  \]
  \[
    \infer[\oplus_{i1}]{\Gamma\vdash\inl(t):A\oplus B}{\Gamma\vdash t:A}
    \qquad
    \infer[\oplus_{i2}]{\Gamma\vdash\inl(t):A\oplus B}{\Gamma\vdash t:B}
  \]
  \[
    \infer[\oplus_e]{\Gamma,\Theta\vdash\elimoplus(t,x.u,y.v):C}
    {\Gamma\vdash t:A\oplus B & x:A,\Theta\vdash u:C & y:B,\Theta\vdash v:C}
  \]
  \[
    \infer[\odot_i]{\Gamma\vdash\super tu:A\odot B}{\Gamma\vdash t:A & \Gamma\vdash u:B}
    \qquad
    \infer[\odot_{e1}]{\Gamma\vdash\fstsup(t):A}{\Gamma\vdash t:A\odot B}
    \qquad
    \infer[\odot_{e2}]{\Gamma\vdash\sndsup(t):B}{\Gamma\vdash t:A\odot B}
  \]
  \[
    \infer[\odot_e]{\Gamma,\Theta\vdash\elimsup(t,x.u,y.v):C}
    {\Gamma\vdash t:A\odot B & x:A,\Theta\vdash u:C & y:B,\Theta\vdash v:C}
  \]
  \caption{The deduction rules of the \OC.\label{fig:typingrules}}
\end{figure}

\subsection{Reduction rules}
The reduction rules define a relation between two proof-terms and a scalar in $\mathcal S$ (in the particular case of $\mathcal S=\mathbb R^{\geq 0}$, it can be seen as a probabilistic reduction relation). 
The first group of rules, that we call ``beta group'' and are presented in Figure~\ref{fig:betarules}, are standard, except for those corresponding to the term $\elimsup$.
\begin{figure}[t]
  \begin{align*}
    \elimone(\escalar s.\star,t) &\lra \escalar s\bullet t & (\elimone)\\
    \elimtens(t\otimes u,xy.r) &\lra (t/x,u/y)r & (\elimtens)\\
    (\lambda x.t)u & \lra  (u/x)t & (\beta)\\
    \pi_1\pair{t}{u} & \lra  t & (\pi_1)\\
    \pi_2\pair{t}{u} & \lra  u & (\pi_2)\\
    \elimoplus(\inl(t), x.v, y.w) & \lra  (t/x)v &({\elimoplus}_1)\\
    \elimoplus(\inr(u), x.v, y.w) & \lra  (u/y)w & ({\elimoplus}_2)\\
    \pi_1^\odot\super{t}{u} & \lra  t & (\pi^\odot_1)\\
    \pi_2^\odot\super{t}{u} & \lra  u & (\pi^\odot_2)\\
    \elimsup(\super{t_1}{t_2}, x.r, y.s) & \lra[\escalar p]  ({t_1}/x)r &({\relimsup \ell})\\
    \elimsup(\super{t_1}{t_2}, x.r, y.s) & \lra[\escalar q]  ({t_2}/x)r &({\relimsup r}) 
  \end{align*}
	\[
  \infer[(C)]{K[t]\lra[\escalar p] K[r]}{t\lra[\escalar p] r}
	\]
  \caption{The beta group of reduction rules of the \OC.\label{fig:betarules}}
\end{figure}
\begin{remark}
  \label{rmk:sup-e-reduction}
  Continuing with Remark~\ref{rmk:sup-e-logic}, if we consider instead of $\super{t_1}{t_2}$, the term $\pair{t_1}{t_2}$, the rule $({\relimsup \ell})$ would be equivalent to
  \[
    \elimoplus(\inl(\fst\pair{t_1}{t_2}),{x}.u,{y}.v) \lra[\escalar p]  ({t_1}/x)u,
  \]
  and the rule $({\relimsup r})$ to 
  \[
    \elimoplus(\inr(\snd\pair{t_1}{t_2}),{x}.u,{y}.v) \lra[\escalar q]  ({t_2}/y)v.
  \]
\end{remark}

The deduction rules $\plus$ and $\bullet(\escalar s)$ allows the building of proofs that cannot be reduced because the introduction rule of some connective and its elimination rule are separated by an interstitial rule. For example,
\[
  \irule{\irule{\irule{\irule{\pi_1}{\Gamma \vdash A}{}
      }
      {\Gamma \vdash A \oplus B}
      {\oplus_{i1}}
      &
      \irule{\irule{\pi_2}{\Gamma \vdash A}{}
      }
      {\Gamma \vdash A \oplus B}
      {\oplus_{i1}}
    }
    {\Gamma \vdash A \oplus B}
    {\plus}
    &
    \irule{\pi_3}{\Gamma, A \vdash C}{}
    &
    \irule{\pi_4}{\Gamma, B \vdash C}{}
  }
  {\Gamma \vdash C}
  {\oplus_e.}
\]
Reducing such a proof, sometimes called a commuting cut, requires reduction rules to commute the rule sum either with the elimination rule below or with the introduction rules above.

As commutation with the introduction rules above is not always possible, for example in the proof
\[
  \irule{\irule{\irule{\pi_1}{\Gamma \vdash A}{}}
    {\Gamma \vdash A \oplus B}
    {\oplus_{i1}}
    &
    \irule{\irule{\pi_2}{\Gamma \vdash B}{}}
    {\Gamma \vdash A \oplus B}
    {\oplus_{i2}}
  }
  {\Gamma \vdash A \oplus B}
  {\plus,}
\]
the commutation with the elimination rule below is often preferred.  However, in the \OC, the commutation of the interstitial rules with the introduction rules is chosen, rather than with the elimination rules, whenever it is possible, that is for all connectives except the disjunction and the tensor. For example, the proof
\[
  \irule{\irule{\irule{\pi_1}{\Gamma \vdash A}{}
      &
    \irule{\pi_2}{\Gamma \vdash B}{}}
    {\Gamma \vdash A \with B}
    {\with_i}
    &
    \irule{\irule{\pi_3}{\Gamma \vdash A}{}
      &
      \irule{\pi_4}{\Gamma \vdash B}{}
    }
    {\Gamma \vdash A \with B}
    {\with_i}
  }
  {\Gamma \vdash A \with B}
  {\plus}
\]
reduces to 
\[
  \irule{\irule{\irule{\pi_1}{\Gamma \vdash A}{}
      &
    \irule{\pi_3}{\Gamma \vdash A}{}}
    {\Gamma \vdash A}
    {\plus}
    &
    \irule{\irule{\pi_2}{\Gamma \vdash B}{}
      &
      \irule{\pi_4}{\Gamma \vdash B}{}
    }
    {\Gamma \vdash B}
    {\plus}
  }
  {\Gamma \vdash A \with B}
  {\with_i.}
\]

Such a choice of commutation yields a stronger introduction property for the
considered connective (Theorem~\ref{thm:intros}): Most connectives have as closed normal forms, introductions, rather than linear combinations of those.  The reduction rules corresponding to these commutations are presented in
Figure~\ref{fig:commutationrules}.
\begin{figure}[t]
  \begin{align*}
    \escalar s_1.\star\plus\,\escalar s_2.\star &\lra (\escalar s_1\masS[\mathcal S]\escalar s_2).\star & (\plus_\one)\\
    \elimtens(t\plus u,xy.v) &\lra \elimtens(t,xy.v)\plus\elimtens(u,xy.v) & (\plus_\otimes)\\
    (\lambda x.t)\plus(\lambda x.u) &\lra \lambda x.(t\plus u) & (\plus_\multimap)\\
    \langle\rangle\plus\langle\rangle &\lra\langle\rangle & (\plus_\top)\\
    \pair tu\plus\pair vw&\lra \pair{t\plus v}{u\plus w} & (\plus_\with)\\
    \elimoplus(t\plus u,x.v,y.w) &\lra \elimoplus(t,x.v,y.w)\plus\elimoplus(u,x.v,y.w) & (\plus_\oplus)\\
    \super tu\plus\super vw&\lra \super{t\plus v}{u\plus w} & (\plus_\odot)\\
    \\
    \escalar s_1\bullet\escalar s_2.\star &\lra (\escalar s_1\produ[\mathcal S]\escalar s_2).\star & (\bullet_\one)\\
    \elimtens(\escalar s\bullet t,xy.v) &\lra \escalar s\bullet\elimtens(t,xy.v) & (\bullet_\otimes)\\
    \escalar s\bullet (\lambda x.t) &\lra\escalar s\bullet  \lambda x.t & (\bullet_\multimap)\\
    \escalar s\bullet  \langle\rangle &\lra\langle\rangle & (\bullet_\top)\\
    \escalar s\bullet  \pair tu&\lra \pair{\escalar s\bullet t}{\escalar s\bullet u} & (\bullet_\with)\\
    \elimoplus(\escalar s\bullet t,x.v,y.w) &\lra \escalar s\bullet \elimoplus(t,x.v,y.w) & (\bullet_\oplus)\\
    \escalar s\bullet \super tu&\lra \super{\escalar s\bullet t}{\escalar s\bullet u} & (\bullet_\odot)
  \end{align*}
  \caption{The commutation group of reduction rules of the \OC.\label{fig:commutationrules}}
\end{figure}

\subsection{Correctness}\label{sec:correctness}
The safety properties (subject reduction, confluence, strong normalisation, and introduction) have been established for the $\mathcal{L}\odot^{\mathcal{S}}$-calculus in~\cite{DiazcaroDowekMSCS24} (with the exception of confluence, which has been proved for the fragment of the calculus without $\odot$). These results extend trivially to the \OC. We state the theorems next.

\begin{theorem}[Subject reduction~{\cite[Theorem 2.2]{DiazcaroDowekMSCS24}}]
  \label{thm:SR}
  If $\Gamma\vdash t:A$ and $t\lra[\escalar p] u$, then $\Gamma\vdash u:A$.
  \qed
\end{theorem}

\begin{theorem}[Confluence~{\cite[Theorem 2.3]{DiazcaroDowekMSCS24}}]
  \label{thm:confluence}
  The \OC is confluent if we exclude the rules ${\relimsup \ell}$ and ${\relimsup r}$.
      \qed
\end{theorem}
\begin{theorem}[Strong normalisation~{\cite[Corollary 2.29]{DiazcaroDowekMSCS24}}]
  \label{thm:SN}
  The $\mathcal L\odot^{\mathcal S\escalar p}$-calcu-lus is strongly normalizing.
      \qed
\end{theorem}
\begin{theorem}[Introduction~{\cite[Theorem 2.30]{DiazcaroDowekMSCS24}}]
  \label{thm:intros}
  Let $\vdash t:A$ and $t$ be irreducible.
  \begin{itemize}
    \item If $A=\one$, then $t=\star$.
    \item If $A=B\otimes C$, then $t=u\otimes v$, $u\plus v$, or $\escalar s\bullet u$.
    \item If $A=B\multimap C$, then $t=\lambda x.u$.
    \item If $A=\top$, then $t=\langle\rangle$.
    \item $A$ cannot be equal to $\zero$.
    \item If $A=B\with C$, then $t=\pair uv$.
    \item If $A=B\oplus C$, then $t=\inl(l)$, $t=\inr(r)$, $u\plus v$, or $\escalar s\bullet u$.
    \item If $A=B\odot C$, then $t=\super uv$.
      \qed
  \end{itemize}
\end{theorem}

\section{Examples and applications}\label{sec:examples}
\subsection{Vectors and matrices}
In this section we replicate some results of~\cite{DiazcaroDowekMSCS24} for the $\mathcal L^{\mathcal S}$-calculus, that is, the fragment of \OC without $\odot$\footnote{In fact, we can just remove $\elimsup$, since $\odot$, without rule $\odot_e$ from Figure~\ref{fig:typingrules} becomes a second additive conjunction where all the results are still valid.}.
These results show that the \OC can be used to encode vectors and matrices, and, moreover, that the sum and scalar product in the syntax represent the sum and scalar product of the elements of a semimodule, and that all the abstractions that we can construct with these symbols are homomorphisms.  Please, refer to that paper for a comprehensive treatment of the subject.

The set of semimodule propositions ${\mathcal V}$ is inductively defined as follows: $\one \in {\mathcal V}$, and if $A$ and $B$ are in ${\mathcal V}$, then so is $A \with  B$.
To each proposition $A \in {\mathcal V}$, we associate a positive natural number $d(A)$, which is the number of occurrences of the symbol $\one$ in $A$: $d(\one) = 1$ and $d(B \with  C) = d(B) + d(C)$.

If $A \in {\mathcal V}$ and $d(A) = n$, then the closed irreducible proofs of $A$ and the elements of the semimodule ${\mathcal S}^n$ are in one-to-one correspondence: to each closed irreducible proof $t$ of $A$, we associate an element $\underline{t}$ of ${\mathcal S}^n$ and to each element ${\bf u}$ of ${\mathcal S}^n$, we associate a closed irreducible proof $\overline{\bf u}^A$ of $A$.

\begin{definition}[One-to-one correspondence~{\cite[Definition 3.6]{DiazcaroDowekMSCS24}}]
  \label{def:onetoone}
  Let $A \in {\mathcal V}$ with $d(A) = n$.  To each closed irreducible
  proof $t$ of $A$, we associate an element $\underline{t}$ of ${\mathcal
  S}^n$ as follows.
  \begin{itemize}
    \item
      If $A = \one$, then $t = \escalar s.\star$. We let $\underline{t} =
      \left(\begin{smallmatrix} \escalar s \end{smallmatrix}\right)$.

    \item
      If $A = A_1 \with  A_2$, then $t = \pair{u}{v}$.  We let
      $\underline{t}$ be $\underline{t} = \left(\begin{smallmatrix}
        \underline{u}\\\underline{v} \end{smallmatrix}\right)$, where we use the block notation with the convention that if ${\bf u} =
      \left(\begin{smallmatrix} 1\\2 \end{smallmatrix}\right)$
      and ${\bf v} =
      \left(\begin{smallmatrix} 3 \end{smallmatrix}\right)$,
      then 
      $\left(\begin{smallmatrix}
        {\bf u}\\{\bf v} \end{smallmatrix}\right) =
        \left(\begin{smallmatrix}
        1\\2\\3 \end{smallmatrix}\right)$ and not
       $\left(\begin{smallmatrix}
          \left(\begin{smallmatrix} 1\\2 \end{smallmatrix}\right)
          \\
          \left(\begin{smallmatrix} 3 \end{smallmatrix}\right)
        \end{smallmatrix}\right)$.
  \end{itemize}

  To each element ${\bf u}$ of ${\mathcal S}^n$, we associate a closed
  irreducible proof $\overline{\bf u}^A$ of $A$.

  \begin{itemize}
    \item If $n = 1$, then ${\bf u} = \left(\begin{smallmatrix}
      \escalar s \end{smallmatrix}\right)$. We let $\overline{\bf u}^A = \escalar s.\star$.

    \item If $n > 1$, then $A = A_1 \with  A_2$, let $n_1$ and $n_2$ be
      the dimensions of $A_1$ and $A_2$.  Let ${\bf u}_1$ and ${\bf u}_2$
      be the two blocks of ${\bf u}$ of $n_1$ and $n_2$ rows, so ${\bf u}
      = \left(\begin{smallmatrix} {\bf u}_1\\ {\bf
      u}_2\end{smallmatrix}\right)$.  We let $\overline{\bf u}^A =
      \pair{\overline{{\bf u}_1}^{A_1}}{\overline{{\bf u}_2}^{A_2}}$.
  \end{itemize}
\end{definition}

We extend the definition of $\underline{t}$ to any closed proof of
$A$, $\underline{t}$ is by definition $\underline{t'}$ where $t'$ is
the irreducible form of $t$.

The following two lemmas show that the sum and scalar product in the syntax express the sum and scalar product of the vectors just defined.

\begin{lemma}[Sum of two vectors~{\cite[Lemma 3.7]{DiazcaroDowekMSCS24}}]
  \label{parallelsum}
  Let $A \in {\mathcal V}$, and $u$ and $v$ be two closed proofs of $A$.
  Then, $\underline{u \plus v} = \underline{u} + \underline{v}$.
\end{lemma}

\begin{lemma}[Product of a vector by a scalar~{\cite[Lemma 3.8]{DiazcaroDowekMSCS24}}]
  \label{parallelprod}
  Let $A \in {\mathcal V}$ and $u$ be a closed proof of $A$.  Then,
  $\underline{\escalar s \bullet u} = \escalar s \underline{u}$.
  \qed
\end{lemma}

\begin{theorem}[Matrices~{\cite[Theorem 3.10]{DiazcaroDowekMSCS24}}]
  \label{thm:matrices}
  Let $A, B \in {\mathcal V}$ with $d(A) = m$ and $d(B) = n$ and let $M$
  be a matrix with $m$ columns and $n$ rows, then there exists a closed
  proof $t$ of $A \multimap B$ such that, for all the elements ${\bf u}$ of
  ${\mathcal S}^m$, we have $\underline{t\overline{\bf u}^A} = M {\bf u}$.
\end{theorem}
\begin{proof}
  By induction on $A$.  We reproduce here the proof of~\cite[Theorem
  3.10]{DiazcaroDowekMSCS24} in full details since it gives the explicit
  proof-terms representing the matrices. 
  \begin{itemize}

    \item If $A = \one$, then $M$ is a matrix of one column and
      $n$ lines. Hence, it is also a vector of $n$ lines.
      We take
      $$t = \lambda x. \elimone(x,\overline{M}^B)$$
      Let ${\bf u} \in {\mathcal S}^1$, ${\bf u}$ has the form
      $\left(\begin{smallmatrix}  a \end{smallmatrix}\right)$ and 
      $\overline{\bf u}^A = \escalar s.\star$.
      Then, using
      Lemma~\ref{parallelprod}, we have
      \begin{align*}
      \underline{t~\overline{\bf u}^A} 
      &= \underline{\elimone(\overline{\bf u}^A,\overline{M}^B)}
      = \underline{\elimone(\escalar s.\star,\overline{M}^B)}\\
      &= \underline{\escalar s \bullet \overline{M}^B}
      = \escalar s \underline{\overline{M}^B} = \escalar s M = M
      \left(\begin{smallmatrix}  a\end{smallmatrix}\right) =
      M {\bf u}
      \end{align*}

    \item If $A = A_1 \with  A_2$, then let $d(A_1) = m_1$
      and $d(A_2) = m_2$. Let $M_1$ and $M_2$ be the two
      blocks of $M$ of $m_1$ and $m_2$ columns, so $M =
      \left(\begin{smallmatrix} M_1 & M_2\end{smallmatrix}\right)$.

      By induction hypothesis, there exist closed proofs $t_1$ and $t_2$ of
      the propositions $A_1 \multimap B$ and $A_2 \multimap B$ such
      that, for all vectors ${\bf u}_1 \in {\mathcal S}^{m_1}$ and ${\bf
      u}_2 \in {\mathcal S}^{m_2}$, we have $\underline{t_1~\overline{{\bf
      u}_1}^{A_1}} = M_1 {\bf u}_1$ and $\underline{t_2~\overline{{\bf
      u}_2}^{A_ 2}} = M_2 {\bf u}_2$.  We take
      $$t = \lambda x. (\elimwith^1(x,y.(t_1~y)) \plus \elimwith^2(x, {z}. (t_2~z)))$$
      Let ${\bf u} \in {\mathcal S}^m$, and ${\bf u}_1$ and ${\bf u}_2$ be
      the two blocks of $m_1$ and $m_2$ lines of ${\bf u}$, so ${\bf u} =
      \left(\begin{smallmatrix} {\bf u}_1  {\bf
      u}_2 \end{smallmatrix}\right)$, and $\overline{\bf u}^A =
      \pair{\overline{{\bf u}_1}^{A_1}}{\overline{{\bf u}_2}^{A_ 2}}$.
      Then, using Lemma~\ref{parallelsum}, we have
      \begin{align*}
	\underline{t~\overline{\bf u}^A} 
	&= \underline{\elimwith^1(\pair{\overline{{\bf
	    u}_1}^{A_1}}{\overline{{\bf u}_2}^{A_ 2}}, {y}.
	  (t_1~y)) \plus \elimwith^2(\pair{\overline{{\bf
	    u}_1}^{A_1}}{\overline{{\bf u}_2}^{A_ 2}}, {z}.
	(t_2~z))}
	\\
	&= \underline{(t_1~\overline{{\bf u}_1}^{A_1}) \plus (t_2~\overline{{\bf u}_2}^{A_ 2})}
	= \underline{t_1~\overline{{\bf u}_1}^{A_1}} + \underline{t_2~\overline{{\bf u}_2}^{A_ 2}}\\
	&= M_1 {\bf u}_1 + M_2 {\bf u}_2
	= \left(\begin{smallmatrix} M_1 & M_2 \end{smallmatrix}\right)
	\left(\begin{smallmatrix} {\bf u}_1 \\ {\bf u}_2  \end{smallmatrix}\right)
	= M {\bf u}
	\qedhere
      \end{align*}
  \end{itemize}
\end{proof}

\begin{definition}[Computational equivalence~{\cite[Definition 4.1]{DiazcaroDowekMSCS24}}]\label{def:compeq}
  Two proofs of a proposition $A$ are computationally equivalent, denoted $t_1 \cong t_2$, if for all propositions $B\in {\mathcal V}$ and all proofs $u$ such that $x:A\vdash u:B$, we have $(u[t_1/x])_\downarrow = (u[t_2/x])_\downarrow$, where $(t)_\downarrow$ is the normal form of $t$.
\end{definition}

\begin{theorem}[Linearity~{\cite[Corollary 4.12]{DiazcaroDowekMSCS24}}]
  \label{thm:linearity}
  Let $A$ and $B$ be propositions, $t$ a closed proof of $A \multimap B$ and
  $u_1$ and $u_2$ be closed proofs of $A$.

	\begin{itemize}
		\item If $B \in {\mathcal V}$, we have
		\[
		 (t(u_1 \plus u_2))_\downarrow = (tu_1 \plus tu_2)_\downarrow
		 \qquad\textrm{and}\qquad
		 (t(\escalar s\bullet u_1))_\downarrow = (\escalar s\bullet tu_1)_\downarrow
		 \]
		 where $(t)_\downarrow$ is the normal form of $t$.
		\item In the general case, we have 
		\[
		t(u_1 \plus u_2)\cong tu_1 \plus tu_2
		\qquad\textrm{and}\qquad
		t(\escalar s\bullet u_1)\cong \escalar s\bullet tu_1
		\tag*{\qed}
		\]
	\end{itemize}
\end{theorem}

The next corollary is the converse of Theorem \ref{thm:matrices}.
\begin{corollary}[Linearity~{\cite[Corollary 4.13]{DiazcaroDowekMSCS24}}]
  \label{cor:linearity}
  Let $A, B \in {\mathcal V}$, such that $d(A) = m$ and $d(B) = n$, and  $t$ be
  a closed proof of $A \multimap B$.  Then the function ${\mathcal S}^m\xlra
  f{\mathcal S}^n$, defined as $f({\bf u}) = \underline{t\overline{\bf u}^A}$
  is linear.
	\qed
\end{corollary}

Finally, we can prove that the sum and scalar product in the syntax represent the sum and scalar product of the elements of a semimodule.
\begin{theorem}[Syntactic sum and scalar multiplication~{\cite[Lemmas 3.7 and 3.8]{DiazcaroDowekMSCS24}}]
  Let $A \in {\mathcal V}$, and $u$ and $v$ be two closed proofs of $A$.
	Then,
		$\underline{u \plus v} = \underline{u} + \underline{v}$ and 
		$\underline{\escalar s \bullet u} =  \escalar s\underline{u}$.
		\qed
\end{theorem}

\subsection{Concrete examples: probabilistic choice}
In this section we present some examples of the use of the \OC[\mathbb R^{\geq 0}], that is, the \OC where the semiring $\mathcal S$ is the semiring of non-negative real numbers. In this case, the reduction relation can be seen as a probabilistic reduction relation, where the probability of a reduction $t \lra[\escalar p] u$ is $p$. We sometimes use $\one \with \one$ to represent $\mathbb R^2$, as in Definition~\ref{def:onetoone}, and other times $\one \odot \one$, depending on the encoding we want to illustrate.

\begin{example}[Biased coin toss]\label{ex:coin}
  The first example is a simple biased coin toss in the \OC[\mathbb R^{\geq 0}]. We represent the two possible outcomes, heads and tails, by the proofs $\inl(1.\star)$ and $\inr(1.\star)$ of the proposition $\one \oplus \one$. The biased coin toss itself, which returns heads with probability $\tfrac{3}{4}$ and tails with probability $\tfrac{1}{4}$, is represented by the proof
  \[
    \elimsup[\tfrac{3}{4},\tfrac{1}{4}](\super{\inl(1.\star)}{\inr(1.\star)}, x.x, y.y)
  \]
  which reduces with probability $\tfrac{3}{4}$ to $\inl(1.\star)$ and with probability $\tfrac{1}{4}$ to $\inr(1.\star)$.
\end{example}

\begin{example}[Stochastic matrix]\label{ex:stochastic}
  The second example is a simple stochastic matrix in the \OC[\mathbb R^{\geq 0}], using the encoding from Theorem~\ref{thm:matrices}. Let $M$ be the stochastic matrix
  \[
    M =
    \begin{pmatrix}
      \tfrac{3}{4} & \tfrac{1}{2}\\
      \tfrac{1}{4} & \tfrac{1}{2}
    \end{pmatrix}
  \]
  which can be encoded as follows.  
  The two columns of the matrix are represented by the two proofs
  \begin{align*}
    t_1 &= \lambda x.\elimone(x.\pair{\tfrac{3}{4}.\star}{\tfrac{1}{4}.\star})\\
    t_2 &= \lambda x.\elimone(x.\pair{\tfrac{1}{2}.\star}{\tfrac{1}{2}.\star})
  \end{align*}
  Thus, the matrix itself is represented by the proof
  \[
    t = \lambda x.(\elimwith^1(x,y.t_1~y) \plus \elimwith^2(x,z.t_2~z)).
  \]

  The action of $M$ on the vector $\bigl(\begin{smallmatrix} 1 \\ 0 \end{smallmatrix}\bigr)$ corresponds to the same biased coin toss as in Example~\ref{ex:coin}, and is represented by the proof $t~\pair{1.\star}{0.\star}$, which reduces, as expected, to $\pair{\tfrac{3}{4}.\star}{\tfrac{1}{4}.\star}$.

  More generally, the action of $M$ on the vector $\left(\begin{smallmatrix} \tfrac{1}{2} \\ \tfrac{1}{2} \end{smallmatrix}\right)$ is represented by the proof $t~\pair{\tfrac{1}{2}.\star}{\tfrac{1}{2}.\star}$, which reduces, as expected, to $\pair{\tfrac{5}{8}.\star}{\tfrac{3}{8}.\star}$.
\end{example}

Thus, Example~\ref{ex:coin} shows how to encode the probabilistic behaviour of a coin toss using the $\odot$ connective (via $\elimsup$), while Example~\ref{ex:stochastic} shows how to encode the probabilistic behaviour of a stochastic matrix using the matrix encoding. The next example shows how to express a probabilistic projection of a vector over the canonical basis in $\mathbb R^2$, which illustrates the expressiveness of the \OC[\mathbb R^{\geq 0}].

\begin{example}[Probabilistic projection]\label{ex:proj}
  Consider the projectors $\pi_1$ and $\pi_2$ over the canonical basis of $\mathbb R^2$, defined as 
  \begin{align*}
    \pi_1\begin{pmatrix} a \\ b \end{pmatrix} &= \begin{pmatrix} a \\ 0 \end{pmatrix},
    &
    \pi_2\begin{pmatrix} a \\ b \end{pmatrix} &= \begin{pmatrix} 0 \\ b \end{pmatrix}.
  \end{align*}
  We can encode a probabilistic projector, applying $\pi_1$ with probability $p$ and $\pi_2$ with probability $1-p$, as follows:
  \[
    \pi = \lambda x.\elimsup[p,1-p](x,y.\super y{0.\star},z.\super{0.\star}z).
  \]
  Indeed, $\pi~\super{a.\star}{b.\star}$ reduces with probability $p$ to $\super{a.\star}{0.\star}$ and with probability $1-p$ to $\super{0.\star}{b.\star}$.
\end{example}

\section{The categorical construction}\label{sec:seminodules}
\subsection{Some properties of categories with biproducts}

\begin{definition}\label{def:semiadditive}
A semiadditive category is a category enriched over commutative monoids, such
that composition is bilinear and the monoid unit acts as an absorbing element.
More precisely:
\begin{enumerate}
    \item \textit{Enrichment in commutative monoids:}  
    For any two objects $A, B$ in the category, the hom-set
    $\Hom AB$
    is equipped with the structure of a commutative monoid
    $
        \big( \Hom AB, +, 0_{AB} \big).
    $

    \item \textit{Bilinearity of composition:}  
    Composition respects the monoid structure, i.e.\ it is bilinear:
    \begin{align*}
      f \circ (g + h) &= f \circ g + f \circ h, & &\text{for all } B \xlra f C, \ A \xlra g B,\text{ and } A\xlra h B, \\
      (g + h) \circ f &= g \circ f + h \circ f, & &\text{for all } A \xlra f B, \ B \xlra g C,\text{ and }B\xlra h C.
    \end{align*}

    \item \textit{Units are absorbing:}  
    The monoid unit $0_{AB}$ acts as a categorical zero morphism, i.e.\ for every $A \xlra f B$,
    \[
        0_{BB} \circ f = 0_{AB}, \qquad f \circ 0_{AA} = 0_{AB}.
    \]
\end{enumerate}

\end{definition}

\begin{definition}\label{defmapaflecha}
  In a category with biproduct, we can define the following operation between maps.
\begin{center}
  \begin{tikzcd}[labels=description,column sep=2cm,row sep=1cm]
  A\ar[d,"\Delta"]\ar[r,"f+g"] & B\\
  A\oplus A\ar[r,"f\oplus g"] & B\oplus B\ar[u,"\nabla"],
\end{tikzcd}
\end{center}
for $A\xlra f B$ and $A\xlra g B$, where $\Delta=\pair{\Id}{\Id}$ and
$\nabla=\coproducto{\Id}{\Id}$.
\end{definition}

\begin{theorem}[Semiadditive structure~{\cite[Proposition 18.4]{Mitchell}}]\label{lem:monoid}
	A category with a biproduct has a unique semiadditive structure in the sense
	of Definition~\ref{def:semiadditive}, where the sum of maps is given by
	Definition~\ref{defmapaflecha}, and the unit of each monoid is given by the
	map $0_{AB}$ defined as $A\xrightarrow{!} 0\xrightarrow{!} B$ (where the zero
	object $0$ is due to the biproducts).
\end{theorem}

\begin{corollary}
	[Semiring]
In a category with biproduct, each $\Hom AA$ of morphisms is a semiring with $+$ given by Definition~\ref{defmapaflecha}
as additive operation, $\circ$ as product operator, where $0_{AA}$ and $\Id_{A}$ are the units of the addition and product respectively.
\end{corollary}
\begin{proof}
	Straightforward.
\end{proof}




\subsection{The category \texorpdfstring{$\SM$}{C-S}}
\begin{definition}[The category $\SM$]
  \label{def:SM}
	Let $\mathcal S$ be a fixed semiring.
  Let $\SM$ be a symmetric monoidal closed category with biproduct
  where there exists a monomorphism from the semiring
  $\mathcal S$ to the semiring $\Hom II$, being $I$ the
  unit object.
\end{definition}

\begin{notation}
  We write 
  \begin{align*}
    \home AB &\textrm{ for the internal hom between $A$ and $B$, }\\
    \otimes &\textrm{ for the tensor product, }\\
    \oplus &\textrm{ for the biproduct, } \\
    I &\textrm{ for the unit object.}
  \end{align*}
  The usual coherence maps are denoted as follows.
  \begin{align*}
    A\otimes B& \xlra{\sigma_{A,B}}B\otimes A, &
    A\otimes (B\otimes C)&\xlra{\alpha_{A,B,C}} (A\otimes B)\otimes C,\\
    I\otimes A&\xlra{\lambda_A} A, &
    A\otimes I&\xlra{\rho_A} A.
  \end{align*}

  The usual maps for the biproduct are denoted as follows.
	\[
    A\oplus B \xlra{\pi_1} A,
		\qquad
    A\oplus B \xlra{\pi_1} B,
		\qquad
    A \xlra{i_1} A\oplus B,\\
		\qquad
    A \xlra{i_2} A\oplus B.
	\]

  Finally, we note $\semS{\cdot}:\mathcal S\longrightarrow\Hom II$ the monomorphism.
\end{notation}

\begin{example}
	\label{ex:injective}
  The following are examples of categories with the properties asked by Definition~\ref{def:SM}.
  \begin{enumerate}
		\item\label{ex:noRel} 
		 The category $(\Rel,\times,\{\star\},\uplus)$, 
		where 
		objects are sets,
		arrows are relations, 
		the tensor is the Cartesian product,
		and the biproduct is the disjoint union, under the condition that $\mathcal S=\{\star\}$, otherwise the map from $\mathcal S$ to $\Hom{\{\star\}}{\{\star\}}$ would not be injective.

    \item The category $(\mathbf{SM}_{\mathcal S},\otimes,\mathcal S,\oplus)$,
		where
		objects are semimodules over the semiring $\mathcal S$,
		arrows are semimodule homomorphisms, 
		the tensor is the semimodules tensor,
		and the biproduct is Cartesian product.
		The map $\semS{\escalar s}$ is $\escalar s'\mapsto \escalar s\produ\escalar s'$. 

    Our first model for the $\mathcal L^{\mathcal S}$-calculus has been given in this category in a previous draft~\cite{DiazcaroMalherbe23arXiv}.

		\item The category $(\mathbf{CPM},\otimes,1,\times)$,
		where
		objects are the lists of natural numbers,
		arrows $\vec n\xlra f\vec m$ are matrices $(f_{ij})$ of completely positive maps from $\mathbb C^{n_i\times n_i}$ to $\mathbb C^{m_j\times m_j}$,
		the tensor is the tensor of vector spaces,
		and the biproduct is the Cartesian product.
		In this category $I=1$ and $\Hom II\simeq\mathbb R^{\geq 0}$, so any monomorphism from $\mathcal S$ to $\mathbb R^{\geq 0}$ is enough. For the $\mathcal L\odot^{\mathbb R^{\geq 0}\escalar p}$-calculus, we can take the identity.

		This category has been defined in~\cite{SelingerMSCS04} and used to model quantum computing in~\cite{PaganiSelingerValironPOPL14}.

		\item The category $(\mathcal R^\Pi,\times,\{\star\},\uplus)$, 
		where
		objects are sets,
		arrows are matrices over the continuous semiring $\mathcal R$,
		the composition is the matrix product,
		the tensor is the Cartesian product,
		and the biproduct is the disjoint union.
		In this category we have $I=\{\star\}$ and $\Hom II\simeq\mathcal R$, so any monomorphism from $\mathcal S$ to $\mathcal R$ is enough. For the $\mathcal L\odot^{\mathcal R\escalar p}$-calculus, we can take the identity.

		This category has been defined and used to model PCF$^{\mathcal R}$, a probabilistic extension of PCF, in~\cite{LairdManzonettoMcCuskerPaganiLICS13}.
  \end{enumerate}

	Notice that 
	the category $(\mathbf{Pcoh},\otimes,(\{\star\},[0,1]))$
	where objects are probabilistic coherence spaces
	and arrows are given by matrices,
	used in~\cite{DanosEhrhardIC11} to model a probabilistic extension of PCF is not an example of our construction since it does not have a biproduct.
	Indeed, the interpretation of $\one$ is $(\{\star\},[0,1])$ and so both $\one\with\one$ and $\one\oplus\one$ have the same web $\{0,1\}$ but $P(\one\with\one) = [0,1]\times[0,1]$ whereas $P(\one\oplus\one) =\{(\alpha,\beta)\in[0,1]\times[0,1] : \alpha+\beta\leq 1\}$.
	
\end{example}


\begin{definition}
  A semiadditive functor is a functor preserving the monoid structure on each hom.
\end{definition}

\begin{lemma}
  \label{lem:distrib}
  Let $F:\SM\to\SM$ be a semiadditive functor. Then there is a natural
  isomorphism
  \[
    F(A)\oplus F(B)\;\cong\;F(A\oplus B),
  \]
  where the arrows are given by
  \begin{align*}
    F(A\oplus B)\xlra{f} F(A)\oplus F(B) & \quad\text{with } f = \pair{F(\pi_1)}{F(\pi_2)},\\
    F(A)\oplus F(B)\xlra{f^{-1}} F(A\oplus B) & \quad\text{with } f^{-1} = \coproducto{F(i_1)}{F(i_2)}.
  \end{align*}
\end{lemma}
\begin{proof}
  Given in~\ref{app:distrib}.
\end{proof}

\begin{corollary}[Distributions]\label{cor:maps}
  In monoidal closed categories with biproducts, there exist the following natural transformations. 
  \begin{enumerate}
    \item $(A\oplus B)\otimes C\xlra{d} (A\otimes C)\oplus (B\otimes C)$
	with $d = \pair{\pi_1\otimes\Id_C}{\pi_2\otimes\Id_C}$.
    \item $
      (A\otimes C)\oplus (B\otimes C)
      \xlra{d^{-1}}
      (A\oplus B)\otimes C
      $
	 with $d^{-1} = \coproducto{i_1\otimes\Id}{i_2\otimes\Id}$.
    \item $\home A{B\oplus C}\xlra\gamma\home AB\oplus\home AC$
	with $\gamma = \pair{\home A{\pi_1}}{\home A{\pi_2}}$.
    \item $\home AB\oplus\home AC\xlra{\gamma^{-1}}\home A{B\oplus C}$ 
	with $\gamma^{-1}=\coproducto{\home A{i_1}}{\home A{i_2}}$. 
  \end{enumerate}
\end{corollary}
\begin{proof}
   Direct consequence of Lemma~\ref{lem:distrib}.
\end{proof}




\subsection{The map \texorpdfstring{$\hat s$}{scalar}}
\begin{definition}[Scalar map]
  For any map $I\xlra{s} I$, we define its corresponding map
  $A\xlra{\hat{s}_A}A$ by $\hat s_A= \rho_A\circ(\Id\otimes{s})\circ\rho_A^{-1}$.
\end{definition}
\begin{lemma}
  \label{lem:s-nat}
  For any map $I\xlra{s} I$, the map $\hat s_A$ is a natural transformation.
\end{lemma}
\begin{proof}
  Given in~\ref{app:s-nat}.
\end{proof}

\begin{lemma}[Some properties of the scalar map]
  \label{lem:PropsS}
 Let $s$ be any map $I\xlra{s} I$.
 Then,
  \begin{enumerate}
    \item $\hat s_{I} = s$.
    \item\label{ite:PropsSoTimes} $\hat s_{A\otimes B} = \hat s_A\otimes\Id_B$.
    \item $\hat s_{A\oplus B} = \hat s_A\oplus\hat s_B$.
  \end{enumerate}
\end{lemma}
\begin{proof}
  Given in~\ref{app:PropsS}.
\end{proof}

Property~\ref{ite:PropsSoTimes} of Lemma~\ref{lem:PropsS} can be rephrased to $F(\hat s)=\hat s$, in the particular case of $F$ being the functor $-\otimes B$.
If we change the functor to be $\home A{-}$, the property, which would be stated as $\home A{\hat s}=\hat s$ is more subtle to prove. We do this in Lemma~\ref{lem:HomeScalEqScal}, but in its proof we need to use the map $\tau$ associated with
the adjunction between the tensor product and the hom, and its naturality with respect to $I$ (Lemma~\ref{lem:tauNatural}).

\begin{lemma}[The map $\tau$]
  \label{lem:tauNatural}
	The following map in the arrows of $\SM$ is a natural transformation with respect to $I$.
  \[
    \tau=\home AB\otimes I\xlra{\varphi_{A,\home AB\otimes I,B\otimes I}(\varepsilon\otimes\Id)}\home A{B\otimes I},
  \]
  where 
  $\varphi_{A,\home AB\otimes I,B\otimes I}$ is the map given by the adjunction 
  \[
    \Hom{X\otimes Y}{Z}\overset{\varphi_{X,Y,Z}}{\underset{\varphi^{-1}_{X,Y,Z}}{\leftrightarrows}}\Hom{Y}{\home XZ},
  \]
  by taking $X = A$, $Y=\home AB\otimes I$, and $Z=B\otimes I$, and where $\varepsilon:\home AB\otimes A\xlra{\varepsilon} B$ is the counit of the adjunction.
\end{lemma}
\begin{proof}
	Given in~\ref{app:tauNatural}.
\end{proof}

\begin{lemma}
  \label{lem:HomeScalEqScal}
   Let $s$ be any map $I\xlra{s} I$.
   Then, for any $A$ and $B$, we have
  $\home A{\hat s_B}=\hat s_{\home AB}$.
\end{lemma}
\begin{proof}
  Consequence of the commutation of the following diagram.
  \[
    \begin{tikzcd}[labels=description,column sep=2cm,row sep=2cm,
	execute at end picture={
	  \path 
	  (\tikzcdmatrixname-1-1) -- (\tikzcdmatrixname-2-1) coordinate[pos=0.5](aux1)
	  (\tikzcdmatrixname-2-1) -- (\tikzcdmatrixname-3-1) coordinate[pos=0.5](aux2)
	  (\tikzcdmatrixname-1-2) -- (\tikzcdmatrixname-2-2) coordinate[pos=0.5](aux3)
	  (\tikzcdmatrixname-2-2) -- (\tikzcdmatrixname-3-2) coordinate[pos=0.5](aux4)
	  (aux1) -- (\tikzcdmatrixname-1-2) node[pos=0.5,red,sloped]{\small (*)}
	  (aux2) -- (aux3) node[pos=0.5,red,sloped]{\small (Lemma~\ref{lem:tauNatural})}
	  (\tikzcdmatrixname-3-1) -- (aux4) node[pos=0.5,red,sloped]{\small (**)};
	}
    ]
      {\home AB} & {\home A{B\otimes I}} \\
      {\home AB\otimes I} & {\home A{B\otimes I}} \\
      {\home AB\otimes I} & {\home AB}
      \arrow["{\home A\rho}", from=1-1, to=1-2]
      \arrow["{\home A{\Id\otimes s}}", from=1-2, to=2-2]
      \arrow["{\home A{\rho^{-1}}}", from=2-2, to=3-2]
      \arrow["\rho"', from=1-1, to=2-1]
      \arrow["{\Id\otimes s}"', from=2-1, to=3-1]
      \arrow["{\rho^{-1}}"', from=3-1, to=3-2]
      \arrow["\tau", from=2-1, to=1-2,sloped,dashed]
      \arrow["\tau", from=3-1, to=2-2,sloped,dashed]
    \end{tikzcd}
  \]
The commutation of the diagram {\color{red}(*)} is proved by an equivalent diagram, obtained through the adjunction of Lemma~\ref{lem:tauNatural}, taking $X=A$, $Y=\home AB$, and $Z=B\otimes I$. Beware, these are not the same variables taken in the definition of $\tau$.
The resulting diagram is as follows.
  
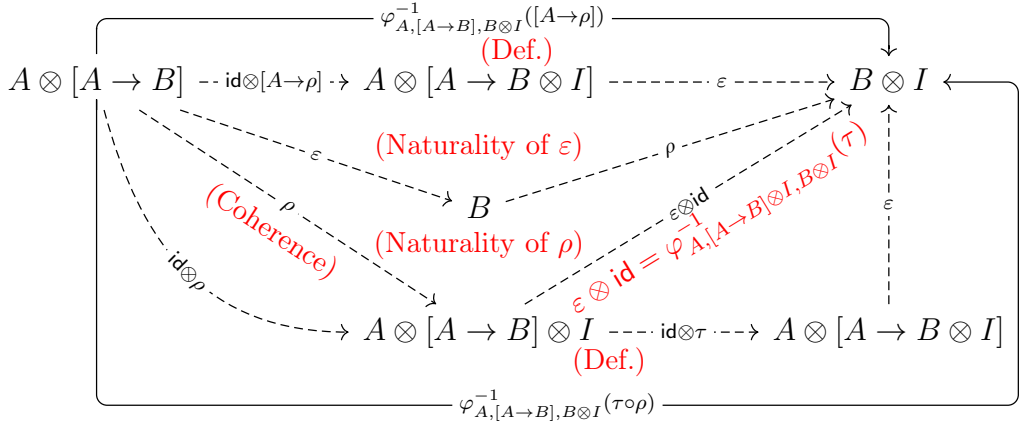
\[
    \begin{tikzcd}[labels=description,column sep=2cm,row sep=1cm,
	execute at end picture={
	  \path 
	  (\tikzcdmatrixname-3-3) -- (\tikzcdmatrixname-1-3) coordinate[pos=0.5](aux)
	  (\tikzcdmatrixname-1-2) -- (\tikzcdmatrixname-2-2) node[pos=0.5,red]{\small (Naturality of $\varepsilon$)}
	  (\tikzcdmatrixname-2-2) -- (\tikzcdmatrixname-3-2) node[pos=0.2,red]{\small (Naturality of $\rho$)}
	  (\tikzcdmatrixname-1-1) -- (\tikzcdmatrixname-3-2) node[pos=0.5,red,sloped,yshift=-.4cm]{\small (Coherence)}
	  (\tikzcdmatrixname-1-1) -- (\tikzcdmatrixname-1-3) node[pos=0.5,red,yshift=4mm]{\small (Def.)}
	  (\tikzcdmatrixname-3-2) -- (\tikzcdmatrixname-3-3) node[pos=0,red,yshift=-4mm]{\small (Def.)}
	  (\tikzcdmatrixname-3-2) -- (\tikzcdmatrixname-1-3) node[pos=0.55,red,sloped,yshift=-4mm]{\small $\varepsilon\otimes\Id = \varphi^{-1}_{A,\home AB\otimes I,B\otimes I}(\tau)$};} ]
      A\otimes\home AB
      \arrow[rr, "{\varphi^{-1}_{A,\home AB,B\otimes I}(\tau\circ\rho)}" description,sloped,
	rounded corners,
	to path={[pos=0.75]
	  -- (\tikztostart.south)
	  |- ([xshift=1cm,yshift=-4.3cm]\tikztotarget.east)\tikztonodes
	|- (\tikztotarget.east)}
      ]
      \arrow[rr, "{\varphi^{-1}_{A,\home AB,B\otimes I}(\home A\rho)}" description,sloped,
	rounded corners,
	to path={[pos=0.75]
	  -- ([yshift=2mm]\tikztostart.north)
	  |- ([yshift=5mm]\tikztotarget.north)\tikztonodes
	|- (\tikztotarget.north)}
      ]
      \ar[ddr,"\Id\otimes\rho",dashed,out=285,in=180,looseness=1.1,sloped]\ar[ddr,"\rho",dashed,sloped]\ar[dr,"\varepsilon",dashed,sloped]\ar[r,dashed,"\Id\otimes\home A\rho"] & A\otimes \home A{B\otimes I}\ar[r,"\varepsilon",dashed] & B\otimes I\\
      & B\ar[ur,"\rho",dashed,sloped] \\
      & A\otimes\home AB\otimes I\ar[r,"\Id\otimes\tau",dashed]\ar[uur,"\varepsilon\otimes\Id",dashed,sloped] & A\otimes\home A{B\otimes I}\ar[uu,"\varepsilon",dashed]
    \end{tikzcd}
  \]

  The commutation of the diagram {\color{red}(**)} is also proved by an
  equivalent diagram, obtained through the adjunction of
  Lemma~\ref{lem:tauNatural}, taking this time the same variables as in
  the definition of $\tau$: $X=A$, $Y=\home AB\otimes I$, and $Z=B\otimes I$.
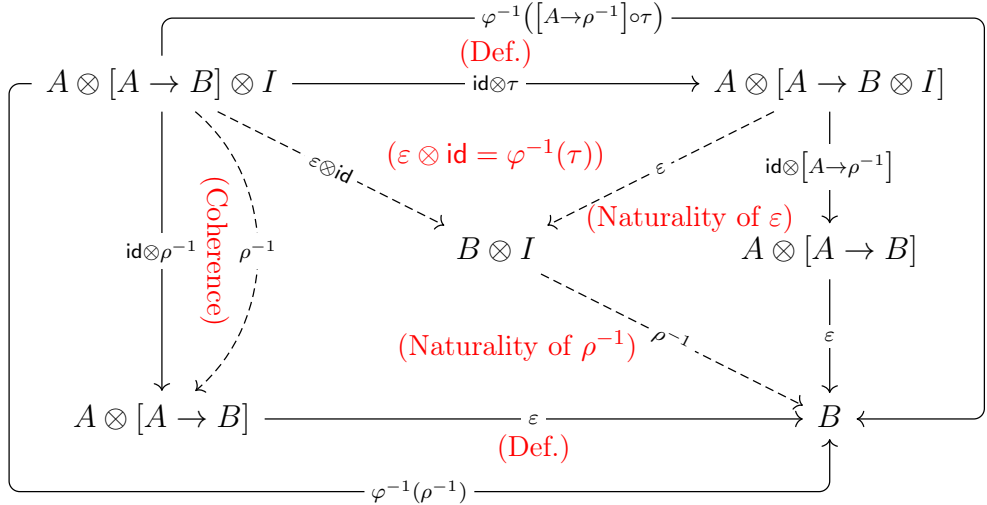
  \[
    \begin{tikzcd}[labels=description,column sep=2cm,row sep=1.5cm,
	execute at end picture={
	  \path 
	  (\tikzcdmatrixname-1-1) -- (\tikzcdmatrixname-1-3) coordinate[pos=0.5](aux1)
	  (\tikzcdmatrixname-1-3) -- (\tikzcdmatrixname-3-3) coordinate[pos=0.3](aux2)
	  (\tikzcdmatrixname-3-1) -- (\tikzcdmatrixname-3-3) coordinate[pos=0.5](aux3)
	  (\tikzcdmatrixname-1-1) -- (\tikzcdmatrixname-3-1) node[pos=0.5,red,sloped,yshift=.7cm]{\small (Coherence)}
	  (aux1) -- (\tikzcdmatrixname-2-2) node[pos=0.5,red]{\small ($\varepsilon\otimes\Id=\varphi^{-1}(\tau)$)}
	  (aux2) -- (\tikzcdmatrixname-2-2) node[pos=0.5,red]{\small (Naturality of $\varepsilon$)}
	  (aux3) -- (\tikzcdmatrixname-2-2) node[pos=0.5,red]{\small (Naturality of $\rho^{-1}$)}
	  (\tikzcdmatrixname-1-1) -- (\tikzcdmatrixname-1-3) node[pos=0.5,red,yshift=4mm]{\small (Def.)}
	  (\tikzcdmatrixname-3-1) -- (\tikzcdmatrixname-3-3) node[pos=0.5,red,yshift=-4mm]{\small (Def.)};
	}
      ]
      {A\otimes\home AB\otimes I}
      \arrow[rrdd, "{\varphi^{-1}\left(\home A{\rho^{-1}}\circ\tau\right)}" description,sloped,
	    rounded corners,
	    to path={[pos=0.25]
	      -- ([yshift=5mm]\tikztostart.north)
	      -| ([xshift=1.7cm]\tikztotarget.east)\tikztonodes
	    -- (\tikztotarget)}
	  ]
	  \arrow[rrdd, "{\varphi^{-1}(\rho^{-1})}" description,sloped,
	    rounded corners,
	    to path={[pos=0.75]
	      -- ([xshift=-3mm]\tikztostart.west)
	      |- ([yshift=-7mm]\tikztotarget.south)\tikztonodes
	    -- (\tikztotarget)}
	  ]
      && {A\otimes\home A{B\otimes I}} \\
      & {B\otimes I} & {A\otimes\home AB} \\
      {A\otimes\home AB} && B
      \arrow["\Id\otimes\tau", from=1-1, to=1-3]
      \arrow["{\Id\otimes\home A{\rho^{-1}}}", from=1-3, to=2-3]
      \arrow["\varepsilon", from=2-3, to=3-3]
      \arrow["{\Id\otimes\rho^{-1}}"', from=1-1, to=3-1]
      \arrow["\varepsilon"', from=3-1, to=3-3]
      \arrow["{\varepsilon\otimes \Id}", dashed, from=1-1, to=2-2,sloped]
      \arrow["{\rho^{-1}}", dashed, from=2-2, to=3-3,sloped]
      \arrow["{\rho^{-1}}", bend left=50, dashed, from=1-1, to=3-1]
      \arrow["\varepsilon", dashed, from=1-3, to=2-2,sloped]
    \end{tikzcd}
  \]
\end{proof}

\begin{remark}
  \label{rmk:NoCompactNeeded}
  In order to prove the Lemma~\ref{lem:HomeScalEqScal} we needed the natural transformation $\tau$ coming from the adjunction given by the fact that the category is assumed to be closed. Notice, however, that the property is almost trivial in the case of compact categories.
  
  Also, if $F=\home A{-}$ were a monoidal functor, the property could have been easily proven by the following diagram.
      \[
	\begin{tikzcd}[labels=description,column sep=1.8cm,row sep=1cm,
	    execute at end picture={
	      \path 
	      (\tikzcdmatrixname-1-2) -- (\tikzcdmatrixname-2-2) coordinate[pos=0.5](aux2)
	      (\tikzcdmatrixname-1-3) -- (\tikzcdmatrixname-2-3) coordinate[pos=0.5](aux3)
	      (\tikzcdmatrixname-1-1) -- (\tikzcdmatrixname-2-2) node[midway,red]{\small (Monoidality axiom)}
	      (aux2) -- (aux3) node[midway,red]{\small (Naturality of $m$)}
	      (\tikzcdmatrixname-2-2) -- (\tikzcdmatrixname-3-2) coordinate[pos=0.5](aux4)
	      (\tikzcdmatrixname-2-3) -- (\tikzcdmatrixname-3-3) coordinate[pos=0.5](aux5)
	      (aux4) -- (aux5) node[midway,red]{\small (Naturality of $m_I$)}
	      (\tikzcdmatrixname-1-4) -- (\tikzcdmatrixname-2-3) node[midway,red]{\small (Monoidality axiom)};
	    }
	  ]
	  F(A)\ar[r,"F(\rho)"]\ar[rdd,"\rho",out=270,in=180,sloped] & F(A\otimes I)\ar[r,"F(\Id\otimes s)"] & F(A\otimes I)\ar[r,"F(\rho^{-1})"] & F(A) \\
	  & F(A)\otimes F(I)\ar[u,"m",dashed]\ar[r,"\Id\otimes F(s)"] & F(A)\otimes F(I)\ar[u,"m",dashed]\\
	  & F(A)\otimes I\ar[r,"\Id\otimes s"]\ar[u,"\Id\otimes m_I",dashed] & F(A)\otimes I\ar[ruu,"\rho^{-1}\otimes\Id",sloped,out=0,in=270]\ar[u,"\Id\otimes m_I",dashed]
	\end{tikzcd}
      \]
\end{remark}

\subsection{The map \texorpdfstring{$\sumwe pq$}{weighted codiagonal}}\label{sec:nablapq}
The map $\sumwe pq$ is the key map mentioned in the introduction. Its related map $\Delta_{pq}$ is not needed for the interpretation, instead, we need the usual diagonal map $\Delta$, since $\sumwe pq$, for some particular $(p,q)$, are left inverses of $\Delta$, as shown by Lemma~\ref{lem:nablaPQInvDelta}.

\begin{lemma}[Weighted codiagonal]
  \label{lem:sumwe-nat}
  Let $I\xlra p I$ and $I\xlra q I$ be two maps.
  The map $A\oplus A\xlra{\sumwe{p}{q}} A$ defined by $\sumwe{p}{q} = \coproducto{\hat p}{\hat q}$
  is a natural transformation.
\end{lemma}
\begin{proof}
  Given in~\ref{app:sumwe-nat}.
\end{proof}

\begin{lemma}
	\label{lem:PropsF}
	Let $I\xlra p I$ and $I\xlra q I$ be two maps
	and let $F$ be a semiadditive functor such that $F(\hat p)=\hat p$ and $F(\hat q)=\hat q$. Then,
	\[
	\sumwe pq\circ\pair{F(\pi_1)}{F(\pi_2)} = F(\sumwe pq).
	\]
\end{lemma}
\begin{proof}
      We must show that $F(\sumwe pq)=\sumwe pq\circ \pair{F(\pi_1)}{F(\pi_2)} = F(\sumwe pq)$. We show equivalently (cf.~\ref{app:distrib}), that
      \(
		F(\sumwe pq)\circ \coproducto{F(i_1)}{F(i_2)} = \sumwe pq
		\).
      \begin{align*}
	F(\sumwe pq)\circ \coproducto{F(i_1)}{F(i_2)}
	& = F(\coproducto{\hat p}{\hat q})\circ \coproducto{F(i_1)}{F(i_2)}\\
	& = \coproducto{F(\coproducto{\hat p}{\hat q})\circ F(i_1)}{F(\coproducto{\hat p}{\hat q})\circ F(i_2)}\\
	& = \coproducto{F(\coproducto{\hat p}{\hat q}\circ i_1)}{F(\coproducto{\hat p}{\hat q}\circ i_2)}\\
	& = \coproducto{F(\hat p)}{F(\hat q)}\\
	& = \coproducto{\hat p}{\hat q}\\
	& = \sumwe pq
	\qedhere
      \end{align*}
\end{proof}

\begin{corollary}
	\label{cor:PropsF}
	For any $I\xlra p I$ and $I\xlra q I$, we have
	\begin{enumerate}
    \item\label{ite:d} $\sumwe pq\circ d = \sumwe pq\otimes B$. 
	\item\label{ite:g} $\sumwe pq\circ\gamma = \home A{\sumwe pq}$.
    \item\label{ite:cord} $\nabla\circ d = \nabla\otimes B$.
	\item\label{ite:corg} $\nabla\circ\gamma = \home A{\nabla}$.
	\end{enumerate}
	Where $d$ and $\gamma$ are the distribution maps of Corollary~\ref{cor:maps}.
\end{corollary}
\begin{proof}
	~
	\begin{itemize}
		\item 
	Items~\ref{ite:d} and~\ref{ite:g}:
	It is straightforward to check that both $-\otimes B$ and  $\home A{-}$ are semiadditive functors.
	Thus, by Lemmas~\ref{lem:PropsS}.\ref{ite:PropsSoTimes} and~\ref{lem:HomeScalEqScal}, these functors meet the conditions of Lemma~\ref{lem:PropsF}, which concludes the proof.

	\item
	Items~\ref{ite:cord} and~\ref{ite:corg}:
	These are particular cases of Items~\ref{ite:d} and~\ref{ite:g}, respectively, since $\nabla = \coproducto{\Id}{\Id} = \coproducto{\hat{\Id}}{\hat{\Id}} = \sumwe{\Id}{\Id}$.
	\qedhere
	\end{itemize}
\end{proof}

Analogously to Lemma~\ref{lem:PropsF}, we can state and prove the following lemma for $\Delta$, with a similar corollary to Corollary~\ref{cor:PropsF}. Remark that in Lemma~\ref{lem:PropsFinv} we do not need the hypothesis $F(\hat s)=\hat s$, since we only need to make use of the trivial property $F(\Id)=\Id$.
\begin{lemma}
	\label{lem:PropsFinv}
	Let $F$ be a semiadditive functor. Then,
	\[
	\coproducto{F(i_1)}{F(i_2)}\circ\Delta = F(\Delta).
	\]
\end{lemma}
\begin{proof}
      We must show that $
	 F(\Delta) = \coproducto{F(i_1)}{F(i_2)}\circ\Delta	 
	  $. We show equivalently (cf.~\ref{app:distrib}), that
      \(
	 \pair{F(\pi_1)}{F(\pi_2)}\circ F(\Delta) = \Delta	 
		\).
      \begin{align*}
			\pair{F(\pi_1)}{F(\pi_2)}\circ F(\Delta)
			&=\pair{F(\pi_1)}{F(\pi_2)}\circ F(\pair{\Id}{\Id})\\
			&=\pair{F(\pi_1)\circ F(\pair{\Id}{\Id})}{F(\pi_2)\circ F(\pair{\Id}{\Id})}\\
			&=\pair{F(\pi_1\pair{\Id}{\Id})}{F(\pi_2\pair{\Id}{\Id})}\\
			&=\pair{F({\Id})}{F({\Id})}\\
			&=\pair{\Id}{\Id}\\
			&=\Delta
			\qedhere
		\end{align*}
\end{proof}

\begin{corollary}
  \label{cor:PropsFinv}
  ~
  \begin{enumerate}
	\item $d^{-1}\circ\Delta = \Delta\otimes\Id$.
	\item $\gamma^{-1}\circ\Delta = \home A{\Delta}$.
  \end{enumerate}
	Where $d$ and $\gamma$ are the distribution maps of Corollary~\ref{cor:maps}.
\end{corollary}
\begin{proof} 
	It is straightforward to check that both $-\otimes B$ and  $\home A{-}$ are semiadditive functors.
	Thus, we conclude by Lemma~\ref{lem:PropsFinv}.
\end{proof}

The usual extension of $\Delta$ and $\nabla$ to more general objects is also valid for $\sumwe pq$. The next lemma shows this, for $\Delta$ and $\sumwe pq$, which are the only cases we need.

\begin{lemma}
  \label{lem:PropsSumwe}
  \label{lem:DeltaSigma}
	For any $I\xlra p I$ and $I\xlra q I$, we have
  \begin{enumerate}
	\item $(\sumwe pq\oplus\sumwe pq)\circ(\Id\oplus\sigma\oplus\Id) = \sumwe pq$.
	\item $(\Id\oplus\sigma\oplus\Id)\circ(\Delta\oplus\Delta)=\Delta$.
  \end{enumerate}
\end{lemma}
\begin{proof}
	Given in~\ref{app:PropsSumwe}.
\end{proof}

\subsection{The set \texorpdfstring{$\we$}{of weights}}
\begin{definition}
  $\we = \{(p,q)\in\Hom II\times\Hom II : p + q=\Id_I\}$.
\end{definition}

\needspace{4em}
\begin{example}
  ~
  \begin{enumerate}
    \item In the category $\Rel$, $\we = \{(\emptyset,\Id),(\Id,\emptyset),(\Id,\Id)\}$, where $\emptyset$ is the empty relation.

          Indeed, there are only two elements in $\Hom II$, which are $\emptyset$ and $\Id$, and we can check that for $s_1,s_2\in\{\emptyset,\Id\}$, the equation $\nabla\circ(s_1\oplus s_2)\circ\Delta=\Id$ is non-valid in the case $(\emptyset,\emptyset)$, and it is valid in the other cases.

          First, notice that $I\oplus I=\{T,F\}$ with $T=(\star,0)$ and $F=(\star,1)$. Thus, $I\xlra\Delta I\oplus I$ is the relation $\{(\star,T),(\star,F)\}$. In the same way, $I\oplus I\xlra\nabla I$ is the relation $\{(T,\star),(F,\star)\}$.

          Now, we can analyse the four cases:
          \begin{itemize}
            \item Let $s_1=s_2=\emptyset$. In this case $\nabla\circ(\emptyset\oplus\emptyset)\circ\Delta = \nabla\circ\emptyset\circ\Delta=\emptyset\neq\Id$.

            \item Let $s_1=\emptyset$, $s_2=\Id$. In this case,
                  \begin{align*}
                    &\nabla\circ(\emptyset\oplus\Id)\circ\Delta\\
                     & =\{(T,\star),(F,\star)\}\circ\{(F,F)\}\circ\{(\star,T),(\star,F)\} \\
                     & =\{(T,\star),(F,\star)\}\circ\{(\star,F)\}                         \\
                     & =\{(\star,\star)\}                                                 \\
                     & =\Id.
                  \end{align*}

            \item Let $s_1=\Id$, $s_2=\emptyset$
                  Analogous to the previous case.
            \item Let $s_1=s_2=\Id$. In this case,
                  \begin{align*}
                    &\nabla\circ(\Id\oplus\Id)\circ\Delta\\
                     & =\{(T,\star),(F,\star)\}\circ\{(T,T),(F,F)\}\circ\{(\star,T),(\star,F)\} \\
                     & =\{(T,\star),(F,\star)\}\circ\{(\star,T),(\star,F)\}                     \\
                     & =\{(\star,\star)\}                                                       \\
                     & =\Id.
                  \end{align*}
          \end{itemize}
    \item In the category $\mathbf{SM}_{\mathcal S}$, $\we = \{(p,q)\in\mathcal S^2: p+_{\mathcal S}q = 1_{\mathcal S}\}$.
	  \item In the category $\mathbf{CPM}$, $\we = \{(f,g)\in\mathbb C^2: f,g\in\mathbb R^{\geq 0}$ and $f+g = 1\}$, that is, the stochastic vectors in $\mathbb C^2$.
		\item In the category $\mathcal R^{\Pi}$, $\we = \{(p,q)\in\mathcal R^2:p +_{\mathcal R}q =1_{\mathcal R}\}$.
  \end{enumerate}
\end{example}

The relation between $\Delta$ and $\sumwe pq$ has become evident in Section~\ref{sec:nablapq}. However, the next lemma goes a bit further, showing that in the particular cases of $(p,q)\in\we$, $\sumwe pq$ are the left inverses of $\Delta$.

\begin{lemma}
  \label{lem:nablaPQInvDelta}
  If $(p,q)\in\we$, then 
  $\sumwe pq\circ\Delta=\Id_A$.
\end{lemma}
\begin{proof}
  Consequence of the commutation of the following diagram.
  \[
    \begin{tikzcd}[labels=description,column sep=2cm,
	execute at end picture={
	  \path 
	  (\tikzcdmatrixname-2-3) -- (\tikzcdmatrixname-1-3) coordinate[pos=0.5](aux1)
	  (\tikzcdmatrixname-2-3) -- (\tikzcdmatrixname-4-3) coordinate[pos=0.5](aux2)
	  (\tikzcdmatrixname-5-3) -- (\tikzcdmatrixname-4-3) coordinate[pos=0.5](aux3)
	  (\tikzcdmatrixname-1-2) -- (\tikzcdmatrixname-3-2) node[pos=0.35,red,xshift=3mm]{\small (Naturality of $\Delta$)}
	  (\tikzcdmatrixname-3-2) -- (\tikzcdmatrixname-5-2) node[midway,red,xshift=-3mm]{\small (Naturality of $\nabla$)}
	  (\tikzcdmatrixname-1-1) -- (\tikzcdmatrixname-3-2) node[pos=0.5,red,sloped]{\small ($\Id$)}
	  (\tikzcdmatrixname-1-3) -- (\tikzcdmatrixname-5-3) node[pos=0.5,red,sloped,yshift=1.5cm]{\small (Lemma~\ref{lem:distrib})}
	  (\tikzcdmatrixname-1-2) -- (\tikzcdmatrixname-1-3) node[pos=0.5,red,yshift=.5cm]{\small ($\sumwe pq = \nabla\circ({\hat p}\oplus{\hat q})$)}
	  (\tikzcdmatrixname-3-2) -- (aux1) node[pos=0.6,red,sloped]{\small (Corollary~\ref{cor:PropsFinv})}
	  (\tikzcdmatrixname-3-2) -- (aux2) node[pos=0.5,red]{\small ($(p,q)\in\we$)}
	  (\tikzcdmatrixname-3-2) -- (aux3) node[pos=0.6,red,sloped]{\small (Corollary~\ref{cor:PropsF})};
	}
      ]
      A & {A\oplus A}
	  \arrow[ddddl, "{\sumwe pq}",
	    rounded corners,
	    to path={[pos=0.75]
	      -- ([yshift=5mm]\tikztostart.north)
	      -| ([yshift=-5mm,xshift=10.5cm]\tikztotarget.south)\tikztonodes
	    -| (\tikztotarget.south)}
	  ]
      & {(A\otimes I)\oplus(A\otimes I)} \\
      && {A\otimes (I\oplus I)} \\
      & {A\otimes I} \\
      && {A\otimes (I\oplus I)} \\
      A & {A\oplus A} & {(A\otimes I)\oplus(A\otimes I)}
      \arrow["\Id", from=1-1, to=5-1]
      \arrow["\Delta"', from=1-1, to=1-2]
      \arrow["\rho\oplus\rho"', dashed, from=1-2, to=1-3]
      \arrow["{(\Id\otimes p)\oplus(\Id\otimes q)}", out=0,in=0, looseness=0.3, sloped, dashed, from=1-3, to=5-3]
      \arrow["{\rho^{-1}\oplus\rho^{-1}}"', dashed, from=5-3, to=5-2]
      \arrow["\nabla"', dashed, from=5-2, to=5-1]
      \arrow["\rho"', dashed, from=1-1, to=3-2,sloped,bend left=15]
      \arrow["\rho^{-1}"', dashed, to=1-1, from=3-2,sloped,bend left=15]
      \arrow["\Delta"', dashed, from=3-2, to=1-3,sloped]
      \arrow["\Id\otimes\Delta", dashed, from=3-2, to=2-3,sloped]
      \arrow["d"', dashed, from=2-3, to=1-3]
      \arrow["{\Id\otimes(p\oplus q)}"', dashed, from=2-3, to=4-3]
      \arrow["d"', dashed, from=4-3, to=5-3]
      \arrow["\nabla"', dashed, from=5-3, to=3-2,sloped]
      \arrow["\Id\otimes\nabla", dashed, from=4-3, to=3-2,sloped]
    \end{tikzcd}
  \]
\end{proof}

\begin{remark}
  In the categories $\mathbf{SM}_{\mathbb R^+}$ and $(\mathbb R^+)^\Pi$, where
  $\we=\{(p,q):p+q=1\}$, we have $\sumwe pq(a_1,a_2) = p.a_1+q.a_2$, that is, a probability distribution. In this particular case, Lemma~\ref{lem:nablaPQInvDelta} simply states that $p.a+(1-p).a = a$. 
\end{remark}

\subsection{The map \texorpdfstring{$\delta$}{delta}}
In Lemma~\ref{lem:delta-nat} we introduce a natural transformation $\delta$, which shares some similarity with the map $d$ in its interaction with the map $\sumwe pq$. The property from Corollary~\ref{cor:PropsF}.\ref{ite:d} has an analogy with $\delta$ when $(p,q)\in\we$, as shown by Lemma~\ref{lem:prop-times-sumwe}. However, its proof does not use Lemma~\ref{lem:PropsF}, since the functor $F=-\oplus B$ does not satisfy the hypothesis $F(\hat s) = \hat s$ needed by Lemma~\ref{lem:PropsF}.
The same analogy applies to $\Delta$ with Corollary~\ref{cor:PropsFinv}, as shown by Lemma~\ref{lem:deltaDelta}.

\begin{lemma}
  \label{lem:delta-nat}
  The map $(A\oplus B)\oplus C\xlra\delta (A\oplus C)\oplus(B\oplus C)$ defined by
  \[
    (A\oplus B)\oplus C\xlra{\Id\oplus\Delta}(A\oplus B)\oplus (C\oplus C)\xlra{\Id\oplus\sigma\oplus\Id} (A\oplus C)\oplus (B\oplus C),
  \]
  is a natural transformation
\end{lemma}
\begin{proof}
  The maps $\sigma$ and $\Delta$ are natural, thus, $\delta$ is natural.
\end{proof}

\begin{lemma}
  \label{lem:prop-times-sumwe}
  If $(p,q)\in\we$, then
  $\sumwe pq\circ\delta = \sumwe pq\oplus\Id$.
\end{lemma}
\begin{proof}
  Given in~\ref{app:prop-times-sumwe}.
\end{proof}

\begin{lemma}
  \label{lem:deltaDelta}
  $\Delta = \delta\circ(\Delta\oplus\Id)$.
\end{lemma}
\begin{proof}
	Given in~\ref{app:deltaDelta}.
\end{proof}

The next property is about the interaction of two weighted codiagonal maps.
\begin{lemma}
  \label{lem:iteration-sumwe}
  If $(p,q)\in\we$, then
  $\sumwe{p'}{q'}\circ(\sumwe pq\oplus\Id)=\sumwe pq\circ(\sumwe{p'}{q'}\oplus\sumwe{p'}{q'})\circ\delta$.
\end{lemma}
\begin{proof}
  Consequence of the commutation of the following diagram.
  \[
    \begin{tikzcd}[labels=description,column sep=1.5cm,row sep=2cm,
	execute at end picture={
	  \path 
	  (\tikzcdmatrixname-1-2) -- (\tikzcdmatrixname-2-1) coordinate[pos=0.5](aux1)
	  (\tikzcdmatrixname-1-2) -- (\tikzcdmatrixname-1-3) coordinate[pos=0.3](aux2)
	  (\tikzcdmatrixname-1-3) -- (\tikzcdmatrixname-2-3) coordinate[pos=0.5](aux3)
	  (\tikzcdmatrixname-1-1) -- (aux1) node[pos=0.5,red]{\small (Lemma~\ref{lem:nablaPQInvDelta})}
	  (\tikzcdmatrixname-2-1) -- (aux2) node[pos=0.7,red]{\small (Lemma~\ref{lem:PropsSumwe})}
	  (\tikzcdmatrixname-2-1) -- (aux3) node[pos=0.7,red]{\small (Lemma~\ref{lem:sumwe-nat})};
	}
      ]
      {(A\oplus A)\oplus A}
      \arrow[rr, "{\delta}" description,sloped,
	rounded corners,
	to path={[pos=0.75]
	  -- ([yshift=2mm]\tikztostart.north)
	  |- ([yshift=5mm]\tikztotarget.north)\tikztonodes
	|- (\tikztotarget.north)}
      ]
      & {(A\oplus A)\oplus (A\oplus A)} &[4.5mm] {(A\oplus A)\oplus (A\oplus A)} \\
      {A\oplus A} & A & {A\oplus A}
      \arrow["\Id\oplus\Delta", dashed, from=1-1, to=1-2]
      \arrow["\Id\oplus\sigma\oplus\Id", dashed, from=1-2, to=1-3]
      \arrow["{\sumwe pq\oplus\Id}", from=1-1, to=2-1]
      \arrow["{\sumwe{p'}{q'}}", from=2-1, to=2-2]
      \arrow["{\sumwe{p'}{q'}\oplus \sumwe{p'}{q'}}", from=1-3, to=2-3]
      \arrow["{\sumwe pq}", from=2-3, to=2-2]
      \arrow["{\sumwe pq\oplus\sumwe pq}", dashed, from=1-2, to=2-1,sloped]
      \arrow["{\sumwe pq}", dashed, from=1-3, to=2-1,sloped]
    \end{tikzcd}
  \]
\end{proof}

\section{Denotational semantics}\label{sec:semantics}
\subsection{Definitions and properties}
In this section, we give an interpretation of the \OC in the category $\SM$.
The interpretation of types and contexts are standard, interpreting the $\odot$ as the biproduct.

\begin{definition}[Interpretation of propositions]\label{def:semTypes}
  We consider the following interpretation of propositions in the objects of $\SM$.
  \begin{align*}
    \sem{\one} &= I &
    \sem{\zero} &= 0\\
    \sem{A\otimes B} &=\sem{A}\otimes \sem{B} &
    \sem{A\with B} &= \sem{A}\oplus \sem{B}\\
    \sem{A\multimap B} &= \home{\sem{A}}{\sem{B}}  &
    \sem{A\oplus B} &={\sem{A}}\oplus{\sem{B}}\\
    \sem{\top} &= 0 &
    \sem{A\odot B} &=\sem{A}\oplus \sem{B}
  \end{align*}
\end{definition}

\begin{remark}[Conjunction, disjunction, sup, and biproducts]
\label{rmk:biproduct}
It is worth noting how the biproduct can serve as an interpretation for three
different connectives: $\with$, $\oplus$, and $\odot$. 

For $\with$ and $\oplus$, this is a classical fact in categorical logic
\cite{seely1989linear}. Intuitively, the biproduct is simultaneously a product
and a coproduct, as it satisfies both universal properties. Concretely, by
taking the upper part of the diagram we obtain the product (interpreting
$\with$), while the lower part yields the coproduct (interpreting $\oplus$).
\[
  \begin{tikzcd}[labels=description,column sep=2cm,row sep=1cm]
    &  C\ar[d,"\pair{f_1}{f_2}"]\ar[dr,"f_2"]\ar[dl,"f_1"] & \\
    {A}\ar[r,"i_1",yshift=-1.5mm]\ar[dr,"g_1"] & A\oplus B\ar[d,"{[g_1,g_2]}"]\ar[l,"\pi_1",yshift=1.5mm]\ar[r,"\pi_2",yshift=1.5mm] & {B}\ar[dl,"g_2"]\ar[l,"i_2",yshift=-1.5mm] \\
    &  D &
  \end{tikzcd}
\]
Here, $\pair{f_1}{f_2}$ interprets the introduction rule of $\with$, while
$\pi_1$ and $\pi_2$ correspond to its elimination rules. Similarly, $i_1$ and
$i_2$ interpret the introduction rules of $\oplus$, and $[g_1,g_2]$ its
elimination rule. Observe that $[g_1,g_2]$ can be factorised as
$[g_1,g_2]=\nabla\circ(g_1\oplus g_2)$, yielding the following extended
diagram:
\[
  \begin{tikzcd}[labels=description,column sep=2cm,row sep=1cm]
    &  C\ar[d,"\pair{f_1}{f_2}"]\ar[dr,"f_2"]\ar[dl,"f_1"] & \\
    {A}\ar[r,"i_1",yshift=-1.5mm]\ar[ddr,"g_1"] & A\oplus B\ar[d,"{g_1\oplus g_2}"]\ar[l,"\pi_1",yshift=1.5mm]\ar[r,"\pi_2",yshift=1.5mm] & {B}\ar[ddl,"g_2"]\ar[l,"i_2",yshift=-1.5mm] \\
    & D\oplus D\ar[d,"\nabla"] & \\
    &  D &
  \end{tikzcd}
\]

The case of $\sup$ is subtler, as it combines both aspects of the diagram.
Indeed, the introduction rule is again interpreted by $\pair{f_1}{f_2}$, but
there are now two forms of elimination: a deterministic one, using $\pi_1$ and
$\pi_2$, and a non-deterministic one, inherited from the coproduct but relying
on $\sumwe pq$ instead:
\[
  \begin{tikzcd}[labels=description,column sep=2cm,row sep=1cm]
    &  C\ar[d,"\pair{f_1}{f_2}"]\ar[dr,"f_2"]\ar[dl,"f_1"] & \\
    {A} & A\oplus B\ar[d,"{g_1\oplus g_2}"]\ar[l,"\pi_1"]\ar[r,"\pi_2"] & {B} \\
    & D\oplus D\ar[d,"\sumwe pq"] & \\
    &  D &
  \end{tikzcd}
\]

We will return to this distinction after presenting the interpretation of
deduction rules in Remark~\ref{rmk:semTerms}.
\end{remark}

\begin{definition}[Interpretation of contexts]
  \label{def:semGamma}
  \begin{align*}
    \sem{\emptyset} &= I &
    \sem{\Gamma,x:A}&=\sem{\Gamma}\otimes\sem{A}
  \end{align*}
\end{definition}

\begin{definition}[Interpretation of deduction rules]\label{def:semTerms}
  We consider the following interpretation of proof-terms in the arrows of $\SM$,
	 where $\semS{\escalar s}$ is the interpretation of the scalar $\escalar s$ in $\Hom II$ (see Definition~\ref{def:SM}).

  Since the deduction system is syntax directed
  (cf.~Figure~\ref{fig:typingrules}), we give instead an interpretation for each
  deduction rule.

    \item $\sem{\vcenter{\infer[ax]{x:A\vdash x:A}{}}} 
      =\sem{A}\xlra{\Id}\sem A$

    \item $\sem{\vcenter{\infer[\one_i(\escalar s)]{\vdash \escalar s.\star:\one}{}}}
      =I\xlra{\semS{\escalar s}}I$

    \item $\sem{\vcenter{\infer[\plus]{\Gamma\vdash t\plus u:A}{\Gamma\vdash t:A & \Gamma\vdash u:A}}}
      =\sem\Gamma\xlra{t + u}\sem A$

    \item $\sem{\vcenter{\infer[\bullet(\escalar s)]{\Gamma\vdash\escalar s\bullet t:A}{\Gamma\vdash t:A}}}
      =\sem\Gamma\xlra{t}\sem A\xlra{\widehat{\semS{\escalar s}}}\sem A$

    \item
      $\sem{\vcenter{\infer[\one_e]{\Gamma,\Theta\vdash\elimone(t,u):A}{\Gamma\vdash t:\one & \Theta\vdash u:A}}}
      =
      \sem{\Gamma}\otimes\sem{\Theta}\xlra{t\otimes u} I\otimes\sem A\xlra{\lambda}\sem A
      $

    \item
      \(
	\sem{\vcenter{\infer[\otimes_i]{\Gamma, \Theta\vdash t \otimes u:A \otimes B}{\Gamma \vdash t:A & \Theta\vdash u:B}}}
	=
	\sem\Gamma\otimes\sem\Theta\xlra{t\otimes u} \sem A\otimes\sem B
      \)

    \item
      \(
	\sem{\vcenter{\infer[\otimes_e]{\Gamma, \Theta\vdash \elimtens(t,{x y}.u):C}{\Gamma, x:A, y:B \vdash u:C & \Theta\vdash t:A \otimes B}}}
      \)

      \hfill\(
	=
	\sem\Gamma\otimes\sem\Theta
	\xlra{\Id\otimes t} \sem\Gamma\otimes \sem{A}\otimes \sem{B} \xlra{u}\sem C
      \)

    \item\(
	\sem{
	  \vcenter{\infer[\multimap_i]{\Gamma\vdash\lambda x.t:A\multimap B}{\Gamma,x:A\vdash t:B}}
	}
	=\sem\Gamma\xlra{\eta^{\sem{A}}}\home{\sem{A}}{\sem\Gamma\otimes \sem{A}}
	\xlra{\home {\sem{A}}t} \home{\sem{A}}{\sem B}
      \)

    \item\(
	\sem{
	  \vcenter{\infer[\multimap_e]{\Gamma,\Theta\vdash tu:B}{\Gamma\vdash t:A\multimap B & \Theta\vdash u:A}}
	}
      \)

      \hfill\(
	=
	\sem\Gamma\otimes\sem\Theta
	\xlra{t\otimes u}
	\home{\sem{A}}{\sem{B}}\otimes \sem{A} \xlra{\varepsilon} \sem B
      \)

    \item \(
	\sem{\vcenter{\infer[\top_i]{\Gamma \vdash \langle \rangle:\top}{}}}
	=\sem\Gamma\xlra{!}0
      \)

    \item $\sem{\vcenter{\infer[\zero_e]{\Gamma,\Theta\vdash\elimzero(t):C}{\Gamma\vdash t:\zero}}}
      =\sem{\Gamma}\otimes\sem{\Theta}\xlra{t\otimes\Id} 0 \otimes\sem\Theta\xlra{0}\sem C$

  \item
    \(
      \sem{\vcenter{\infer[\with_i]{\Gamma \vdash \pair tu:A \with B}{\Gamma \vdash t:A & \Gamma \vdash u:B}}}
      =
      \sem\Gamma\xlra{\Delta}\sem\Gamma\oplus\sem\Gamma\xlra{t\oplus u} 
      \sem A\oplus\sem B
    \)

  \item
    $\sem{\vcenter{\infer[\with_{e1}]{\Gamma\vdash\pi_1(t):A}{\Gamma\vdash t:A\with B}}}
    =\sem\Gamma\xlra t \sem A\oplus\sem B\xlra{\pi_1}\sem A$

  \item
    $\sem{\vcenter{ \infer[\with_{e2}]{\Gamma\vdash\pi_2(t):B}{\Gamma\vdash t:A\with B} }}
    =\sem\Gamma\xlra t \sem A\oplus\sem B\xlra{\pi_2}\sem B$

  \item
    $\sem{\vcenter{ \infer[\oplus_{i1}]{\Gamma\vdash\inl\ t:A\oplus B}{\Gamma\vdash t:A} }}
    =\sem\Gamma\xlra{t}\sem A\xlra{i_1}\sem A\oplus\sem B$

  \item
    $\sem{\vcenter{ \infer[\oplus_{i2}]{\Gamma\vdash\inr\ t:A\oplus B}{\Gamma\vdash t:B}
    }}
    =\sem\Gamma\xlra{t}\sem B\xlra{i_2}\sem A\oplus\sem B$

  \item
    $\sem{\vcenter{\infer[\oplus_e]{\Gamma,\Theta\vdash \elimoplus(t,{x}.u,{y}.v):C}{\Gamma \vdash t:A\oplus B & x:A,\Theta\vdash u:C & y:B,\Theta\vdash v:C}}}$

    \hfill$
      =
      \sem\Gamma\otimes\sem\Theta
      \xlra{t\otimes\Id}
      (\sem A\oplus \sem B)\otimes \sem\Theta
      \xlra{d}
      (\sem{A}\otimes\sem\Theta)\oplus (\sem{B}\otimes\sem\Theta)
      \xlra{\coproducto uv}
      \sem C$

    \item\(
    \sem{\vcenter{
	\infer[\odot_i]{\Gamma\vdash\super tu:A\odot B}{\Gamma\vdash t:A & \Gamma\vdash u:B}
    }}
    =
    \sem\Gamma\xlra{\Delta}\sem\Gamma\oplus\sem\Gamma\xlra{t\oplus u} 
    \sem A\oplus\sem  B
  \)

\item\(
    \sem{\vcenter{
	\infer[\odot_{e1}]{\Gamma\vdash\fstsup(t):A}{\Gamma\vdash t:A\odot B}
    }}
    =\sem\Gamma\xlra t \sem A\oplus\sem B\xlra{\pi_1}\sem A
  \)

\item\(
    \sem{\vcenter{
	\infer[\odot_{e2}]{\Gamma\vdash\sndsup(t):B}{\Gamma\vdash t:A\odot B}
    }}
    =\sem\Gamma\xlra t \sem A\oplus\sem B\xlra{\pi_2}\sem B
  \)

\item\(
    \sem{\vcenter{
	\infer[\odot_e]{\Gamma,\Theta\vdash\elimsup(t,x.u,y.v):C}
	{\Gamma\vdash t:A\odot B & x:A,\Theta\vdash u:C & y:B,\Theta\vdash v:C}
  }}\)

  \hfill\(\begin{aligned}[t]
      = & \sem\Gamma\otimes\sem\Theta
      \xlra{t\otimes\Id}(\sem{A}\oplus \sem{B})\otimes\sem\Theta
      \xlra{d}(\sem{A}\otimes\sem\Theta)\oplus (\sem{B}\otimes\sem\Theta)\\
      & \xlra{u\oplus v} \sem C\oplus\sem C
    \xlra{\sumwe{\semS{\escalar p}}{\semS{\escalar q}}} \sem C
  \end{aligned}\)
\end{definition}

\begin{remark}\label{rmk:semTerms}
	The interpretation of deduction rules is mostly standard, having $\odot_i$,
	$\odot_{e1}$, and $\odot_{e2}$ interpreted as if they were the rules for the additive conjunction $\with_i$,
	$\with_{e1}$, and $\with_{e2}$ respectively. However, even if the rule $\odot_e$ looks quite
	similar to the elimination of the additive disjunction, $\oplus_e$, its interpretation has a slight but important
	difference: instead of applying the mediating arrow of the coproduct $\sem A\oplus\sem B\xlra{\coproducto uv}\sem C$, which is equivalent to 
	\[
	  \sem A\oplus\sem B\xlra{u\oplus v}\sem C\oplus\sem C\xlra{\nabla}\sem C,
	\]
	we use 
	\[
	  \sem A\oplus\sem B\xlra{u\oplus v}\sem C\oplus\sem C\xlra{\sumwe pq}\sem C.
	\]
\end{remark}

\subsection{Soundness}
Our interpretation is sound (Theorem~\ref{thm:soundness}) with respect to
reduction.

\begin{lemma}
  [Substitution]
  \label{lem:subst}
  If $\Gamma,x:A\vdash t:B$ and $\Theta\vdash v:A$, then
  $\sem{\Gamma\vdash (v/x)t:B}=\sem{\Gamma,x:A\vdash t:B}\circ (\Id\otimes\sem{\Gamma\vdash v:A})$.
  That is, the following diagram commutes.
  \begin{center}
    \begin{tikzcd}[labels=description,row sep=5mm,column sep=2cm]
      \sem\Gamma\otimes\sem\Theta\ar[dr,"\Id\otimes v",sloped]\ar[rr,"(v/x)t"] && \sem B\\
      &\sem\Gamma\otimes \sem A\ar[ru,"t",sloped]
    \end{tikzcd}
  \end{center}
\end{lemma}
\begin{proof}
  By induction on $t$. Given in~\ref{app:subst}.
\end{proof}

\needspace{4em}
\begin{theorem}[Soundness]
  \label{thm:soundness}
  Let $\Gamma\vdash t:A$.
  \begin{itemize}
    \item If $t\lra r$, by any rule but $({\relimsup \ell})$ and $({\relimsup r})$,  then
      \[
	\sem{\Gamma\vdash t:A}=\sem{\Gamma\vdash r:A}.
      \]
    \item If $t\lra[\escalar p] r_1$ by rule $({\relimsup \ell})$ and $t\lra[\escalar q] r_2$ by rule $({\relimsup r})$, 
       then
      \[
	\sem{\Gamma\vdash t:A} = \sumwe{\semS{\escalar p}}{\semS{\escalar q}}\circ(\sem{\Gamma\vdash r_1:A}\oplus\sem{\Gamma\vdash r_2:A})\circ\Delta,
      \]
      that is,
      \begin{center}
	\begin{tikzcd}[labels=description,row sep=8mm,column sep=1.5cm]
	  \sem\Gamma\ar[r,"t"]\ar[d,"\Delta"] & \sem A\\
	  \sem\Gamma\oplus\sem\Gamma\ar[r,"r_1\oplus r_2"] & \sem A\oplus\sem A.\ar[u,"{\sumwe{\semS{\escalar p}}{\semS{\escalar q}}}"]
	\end{tikzcd}
      \end{center}
  \end{itemize}
\end{theorem}
\begin{proof}
  By induction on the relation $\lra[\escalar p]$, using the properties proven in the previous section.
	The full details are given in~\ref{app:soundness}.
\end{proof}

\subsection{Adequacy and lack of full abstraction}
As usual, our interpretation is not complete with respect to the reduction
relation because we do not consider eta-rules. For example, $\vdash\lambda
x.\pair{\fst x}{\snd x}:A\with B\multimap A\with B$ has the same interpretation
as $\vdash\lambda x.x:A\with B\multimap A\with B$, but one does not reduce to
the other.
Thus, we can only expect it to be {\em adequate} with respect to an observational equivalence\footnote{This equivalence is not the same as that from Definition~\ref{def:compeq}, which was defined only for the fragment without $\odot$.}.

Note also that soundness with respect to observational equivalence---and hence
full abstraction---fails in our setting. For instance, the proof-terms
$1\bullet\lambda x.\pi_1(x)$ and $2\bullet\lambda x.\pi_1(x)$ of the
proposition $(1\with 0)\multimap 1$ are observationally equivalent, since there
is proof of $1\with 0$ to distinguish them. However, their denotations differ,
because $\sem{\zero}$ is not empty.

In addition, we have chosen to not distinguish in the semantics certain
situations. For example, let $S=\mathbb{R}^{\geq 0}$,
$t = \pelimsupof{\frac{1}{2}}{\frac{1}{2}}{\super{\frac{1}{2}.\star}
{\frac{1}{2}.\star}}{x}{x}{y}{y}$
and
$u = \pelimsupof{\frac{1}{2}}{\frac{1}{2}}{\super{\frac{3}{4}.\star}
{\frac{1}{4}.\star}}{x}{x}{y}{y}$.
Then $t$ reduces with probability $1$ to $\frac{1}{2}.\star$, while $u$
reduces with probability $1/2$ to $\frac{3}{4}.\star$ and with probability
$1/2$ to $\frac{1}{4}.\star$.
However, both terms have the same semantics:
\[
  \sem{\vdash t:\one} = \frac{1}{2}\bullet\hat{\frac{1}{2}} \plus
  \frac{1}{2}\bullet\hat{\frac{1}{2}} = \hat{\frac{1}{2}}
  = \frac{1}{2}\bullet\hat{\frac{3}{4}} \plus
  \frac{1}{2}\bullet\hat{\frac{1}{4}} = \sem{\vdash u:\one}.
\]

This design choice clearly draws the line between probabilistic and
quantum behaviour. While the syntax and the idea of indistinguishability
are conceptually inspired by the density matrix formalism in quantum
physics---where different ensembles of pure states can yield the same mixed
state---our calculus remains strictly probabilistic. There are no quantum
phases, amplitudes, or interference phenomena. Instead, we adopt the idea of
representing probabilistic combinations of outcomes as weighted sums of
terms. We consider both the scalars within the terms and the probabilities
arising from reductions uniformly as classical probabilities.

Consequently, we must establish adequacy (Theorem~\ref{thm:adequacy}) with
respect to a ``mixed computational equivalence'' (Definition~\ref{def:mce}),
which equates terms such as $t$ and $u$. Through the Curry-Howard
correspondence, terms represent proofs, and this operational equivalence
becomes the natural notion of equivalence for proofs in our setting. It
identifies two proofs if they observationally yield the same probabilistic
distribution of canonical proofs, effectively internalising probabilistic
reasoning within the proof system itself.

\begin{definition}[Elimination context]
  \label{def:eliminationContextIMALL}
  An elimination context is a typed proof-term context (cf.~Definition~\ref{def:term-context})
  produced by the following grammar.
  \begin{align*}
    K & := [\cdot]\mid\elimtens(K,xy.t)\mid Kt\mid\fst(K)\mid\snd(K)       
     \mid\elimoplus(K,x.t,y.u)\mid\fstsup(K)\mid\sndsup(K),
  \end{align*}
  where
  in $\elimtens(K,xy.t)$, and
  $\elimoplus(K,x.t,y.u)$, 
  the proposition proved by $K$ is strictly bigger than that proved by $t$ and $u$.
\end{definition}

To distinguish between two programs, we can require that these programs, when
placed in any elimination context of a certain type, produce the same outputs. In
particular, that type must admit more than one closed value, which is the case
with $\one$, when there are more than one element in $\mathcal S$, as all the
proof-terms $\escalar s.\star$ are proofs of $\one$. This renders our adequacy
result unsuitable for the category $\Rel$, which serves as a model of
the \OC only in the case of $\mathcal S=\{\star\}$ (see
Example~\ref{ex:injective}.\ref{ex:noRel}). However, the case $\mathcal
S=\{\star\}$ is a degenerate case, as all the rules from
Figure~\ref{fig:commutationrules} become trivial, and we may be able to find a
simpler model for that particular case than the one presented in this paper.

\begin{definition}
  \label{def:basicTypes}
  A basic proposition $\tau$ is either
  $\one$ or $\top$.
\end{definition}



\begin{definition}
	Let $P = [t_0,\ldots,t_n]$ be a list of terms. 

	We write $t\twoheadrightarrow_{\escalar p}^P v$ if 
	\[
	t = t_0\lra[\escalar p_1] t_1\lra[\escalar p_2] \cdots \lra[\escalar p_n] t_n = v,
	\]
	where $n\geq 0$ and $\prod_{i=1}^n \escalar p_i=\escalar p$.

 	That is, the product of the probabilities of all the reductions along the path, give us the probability of the path.
\end{definition}
\begin{notation}
	We write $P_v$ for the list $P$ if $v$ is the last element of the list.
\end{notation}




\begin{definition}
  [Probabilistic computational equivalence]
  \label{def:ndequivIMALL}
  Let $\vdash t:A$ and $\vdash u:A$. We say that $t$ and $u$ are
  probabilistically equivalent, notation $t\sim u$, if for every elimination
  context $[\cdot]:A\vdash K:\tau$ such that $\tau$ is a basic proposition, we have that
	$\forall P$, $K[t]\twoheadrightarrow_{\escalar p}^P v$ iff $K[u]\twoheadrightarrow_{\escalar p}^P v$.
\end{definition}

\begin{definition}
  [Multiset of probability distribution of values of a term]
  The multiset of probability distributions of values of a term $t$ is the following multiset of terms,
  \[
    \spdv[t]
    =\{\escalar p\bullet v : t\twoheadrightarrow_{\escalar p}^P v\}.
  \]
\end{definition}

\begin{notation}
We write $\sum_{t\in T} t$ for the term produced by the constructor $\plus$ with the terms of the set $T$ taken in a lexicographical order, and associating the parenthesis to the left.
For example,  let $T=\{ t_1,t_2,t_3 \}$ with $t_1<t_2<t_3$, then,
\[
  \sum_{t\in T} t = (t_1\plus t_2)\plus t_3.
\]
\end{notation}

\begin{definition}
  [Mixed computational equivalence]
  \label{def:mce}
  Let $\vdash t:A$ and $\vdash u:A$. We say that $t$ and $u$ are
  mixed computational equivalent, notation $t\approx u$, if 
  \[
    \sum_{t'\in \spdv[t]} t'
    \sim
    \sum_{u'\in \spdv[u]} u'.
  \]
\end{definition}

In order to prove adequacy (Theorem~\ref{thm:adequacy}), we need the following alternative formulation of soundness.
\begin{corollary}
  [Soundness]
  \label{cor:soundness}
  $\sem{\vdash t:A}=\sem{\vdash \sum_{t'\in \spdv[t]} t':A}$.
\end{corollary}
\begin{proof}
  Without lost of generality, we make the following assumptions.
  \begin{itemize}
    \item In order to represent the reductions, we make the following modification: for each reduction rule of the form $t\lra[1_{\mathcal S}] r$ we add a new reduction new rule of the form $t\lra[0_{\mathcal S}] r$. It is easy to see that Theorem~\ref{thm:soundness} still stands, since 
      $f 
      = \nabla\circ({\Id_I}\times{0_I})\circ(f\times f)\circ\Delta
      = \sumwe{\semS{1_{\mathcal S}}}{\semS{0_{\mathcal S}}}\circ(f\times f)\circ\Delta
      $.
    \item In addition, to make the analysis easier, we also consider that the reduction tree for a term has all its leaves at the same level, by simply continue reducing the values to themselves until all the branches had reached its values. This does not alter the analysis, because the interpretation of a term is the same as its reduct when it reduces with ``probability'' $1_{\mathcal S}$.
    \item Finally, we use this notation: The first reducts of $t$ are the terms $r_0$ and $r_1$. The next level is as follows.
      The reducts of $r_0$ are $r_{00}$ and $r_{01}$ and the reducts of $r_1$ are $r_{10}$ and $r_{11}$. 
      The next level will add one more bit, and so on. The scalars associated to each reduction follows the same pattern.

      Thus, the term $r_{b_1\cdots b_n}$ is the one reached by the following path 
      \[
	t\lra[\escalar p_{b_1}] r_{b_1}\lra[\escalar p_{b_1b_2}]r_{b_1b_2}\lra[\escalar p_{b_1b_2b_3}]\cdots\lra[\escalar p_{b_1\cdots b_n}]r_{b_1\cdots b_n}.
      \]
  \end{itemize}
  By Theorem~\ref{thm:soundness}, we know that
  \begin{equation}
    \label{eq:known}
    \sem{\vdash t:A} = \sumwe{\semS{\escalar p_0}}{\semS{\escalar p_1}}\circ(\sem{\vdash r_0:A}\times\sem{\vdash r_1:A})\circ\Delta.
  \end{equation}
  Using the same Theorem~\ref{thm:soundness}, we also have
  \begin{align*}
    \sem{\vdash r_0:A} &= \sumwe{\semS{\escalar p_{00}}}{\semS{\escalar p_{01}}}\circ(\sem{\vdash r_{00}:A}\times\sem{\vdash r_{01}:A})\circ\Delta\textrm{, and}\\
    \sem{\vdash r_1:A} &= \sumwe{\semS{\escalar p_{10}}}{\semS{\escalar p_{11}}}\circ(\sem{\vdash r_{10}:A}\times\sem{\vdash r_{11}:A})\circ\Delta.
  \end{align*}

  Let $\overline{\escalar p_{b_1\cdots b_n}} = \escalar p_{b_1}\escalar p_{b_1b_2}\dots \escalar p_{b_1\cdots b_n}$.
  Then, we have
  \begin{align}\nonumber
    &\sem{\vdash \sum_{t'\in \spdv[t]} t':A}
    =
    \sem{\vdash
      \overline{\escalar p_{0\cdots 0}}\bullet r_{0\cdots 0}
      \plus\cdots\plus
    \overline{\escalar p_{1\cdots 1}}\bullet r_{1\cdots 1}:A}
    \\
    \label{eq:toshow}
    &=
    \overline{\nabla}_n\circ(\widehat{\semS{\overline{{\escalar p_{0\cdots 0}}}}}\times\dots\times\widehat{\semS{\overline{\escalar p_{1\cdots 1}}}})\circ
    (\sem{\vdash r_{0\cdots 0}:A}
      \times\cdots\times
    \sem{\vdash r_{1\cdots 1}:A})\circ\overline{\Delta}_n,
  \end{align}
  where
  \begin{align*}
  \overline{\Delta}_1 &= \Delta  &
  \overline{\nabla}_1 &= \nabla \\
  \overline{\Delta}_{n+1} &=\Delta^{2^n}\circ\overline{\Delta}_n &
  \overline{\nabla}_{n+1} &=\nabla^{2^n}\circ\overline{\nabla}_n.
  \end{align*}

  We must prove that \eqref{eq:known} = \eqref{eq:toshow}.

  From equation \eqref{eq:known}, we have:
  \begin{align*}
    &\sem{\vdash t:A}\\
    &= \sumwe{\semS{\escalar p_0}}{\semS{\escalar p_1}}\circ(\sem{\vdash r_0:A}\times\sem{\vdash r_1:A})\circ\Delta\\
    &=
    \sumwe{\semS{\escalar p_0}}{\semS{\escalar p_1}}
    \circ
    \big(
      \begin{aligned}[t]
	&(\sumwe{\semS{\escalar p_{00}}}{\semS{\escalar p_{01}}}\circ(\sem{\vdash r_{00}:A}\times\sem{\vdash r_{01}:A})\circ\Delta)
	\times\\
	&(\sumwe{\semS{\escalar p_{10}}}{\semS{\escalar p_{11}}}\circ(\sem{\vdash r_{10}:A}\times\sem{\vdash r_{11}:A})\circ\Delta)
      \big)
      \circ\Delta
    \end{aligned}
    \\
    &=
    \sumwe{\semS{\escalar p_0}}{\semS{\escalar p_1}}
    \circ
    \big(
      \begin{aligned}[t]
	&(\sumwe{\semS{\escalar p_{00}}}{\semS{\escalar p_{01}}}\circ(\sem{\vdash r_{00}:A}\times\sem{\vdash r_{01}:A}))
	\times\\
	&(\sumwe{\semS{\escalar p_{10}}}{\semS{\escalar p_{11}}}\circ(\sem{\vdash r_{10}:A}\times\sem{\vdash r_{11}:A}))
      \big)
      \circ\overline{\Delta}_1
    \end{aligned}
    \\
    &=
      \begin{aligned}[t]
    &\nabla\circ(\hat{\semS{\escalar p_0}}\times\hat{\semS{\escalar p_1}})
    \circ
	\\
    \big(
	&
	(\nabla\circ(\widehat{\semS{\escalar p_{00}}}\times\widehat{\semS{\escalar p_{01}}})\circ(\sem{\vdash r_{00}:A}\times\sem{\vdash r_{01}:A}))
	\times\\
	&(\nabla\circ(\widehat{\semS{\escalar p_{10}}}\times\widehat{\semS{\escalar p_{11}}})\circ(\sem{\vdash r_{10}:A}\times\sem{\vdash r_{11}:A}))
      \big)
      \circ\overline{\Delta}_1
    \end{aligned}
    \\
    {\color{red}(*)}&=
      \begin{aligned}[t]
   & \overline{\nabla}_1\circ(\widehat{\semS{\overline{\escalar p_{00}}}}\times\widehat{\semS{\overline{\escalar p_{01}}}}\times\widehat{\semS{\overline{\escalar p_{10}}}}\times\widehat{\semS{\overline{\escalar p_{11}}}})
    \circ
	\\
	&
    \big(
	(\sem{\vdash r_{00}:A}\times\sem{\vdash r_{01}:A})
	\times
	(\sem{\vdash r_{10}:A}\times\sem{\vdash r_{11}:A})
      \big)
      \circ\overline{\Delta}_1
    \end{aligned}
  \end{align*}
  Where the equality ${\color{red}(*)}$ is justified by the following commuting diagram.
  \[
    \begin{tikzcd}[labels=description,row sep=1cm,column sep=1cm,
	execute at end picture={
	  \path
	  (\tikzcdmatrixname-2-1) -- (\tikzcdmatrixname-2-3) node[yshift=5mm,pos=0.8,red]{\small (Compositionality)}
	  (\tikzcdmatrixname-4-1) -- (\tikzcdmatrixname-2-3) node[yshift=5mm,midway,red,sloped]{\small (Naturality of $\nabla\times\nabla$)}
	  (\tikzcdmatrixname-4-1) -- (\tikzcdmatrixname-2-3) node[yshift=-5mm,midway,red,sloped]{\small (Same map)}
	  ;
	}
      ]
      {\mathcal S} &[-1.4cm] {(\mathcal S\times \mathcal S)\times (\mathcal S\times \mathcal S)} 
      &[2.4cm] {(A\times A)\times (A\times A)} \\
      {(A\times A)\times (A\times A)} && {(A\times A)\times (A\times A)} \\
      {A\times A} & & {A\times A} \\
      {A\times A} && A 
      \arrow["{\overline{\Delta}_1}", from=1-1, to=1-2]
      \arrow["{(r_{00}\times r_{01})\times (r_{10}\times r_{11})}", from=1-2, to=1-3]
      \arrow["{f}",sloped, from=1-3, to=2-1]
      \arrow["\nabla\times\nabla", from=2-1, to=3-1]
      \arrow["{\hat{p_0}\times\hat{p_1}}", from=3-1, to=4-1]
      \arrow["\nabla", from=4-1, to=4-3]
      \arrow["{g}", from=1-3, to=2-3]
      \arrow["\nabla\times\nabla", from=2-3, to=3-3]
      \arrow["\nabla"{description}, from=3-3, to=4-3]
      \arrow["{(\widehat{\semS{\escalar p_0}}\times \widehat{\semS{\escalar p_0}})\times(\widehat{\semS{\escalar p_1}}\times \widehat{\semS{\escalar p_1}})}"', dashed, from=2-1, to=2-3]
      \arrow["\nabla\times\nabla", dashed, from=2-3, to=4-1,sloped]
    \end{tikzcd}
	\]
	\begin{align*}
		\textrm{where}\qquad
	f&=(\widehat{\semS{\escalar p_{00}}}\times\widehat{\semS{\escalar p_{01}}})\times(\widehat{\semS{\escalar p_{10}}}\times\widehat{\semS{\escalar p_{11}}})\\
	g&=(\widehat{\semS{\escalar p_0\escalar p_{00}}}\times \widehat{\semS{\escalar p_0\escalar p_{01}}})\times(\widehat{\semS{\escalar p_1\escalar p_{10}}}\times \widehat{\semS{\escalar p_1\escalar p_{11}}})
	\qedhere
	\end{align*}
\end{proof}

\begin{theorem}[Adequacy of $\approx$]
  \label{thm:adequacy}
  If $\sem{\vdash t:A}=\sem{\vdash u:A}$ then $t\approx u$.
\end{theorem}
\begin{proof}
  To prove $t\approx u$ we need to prove that
  \[
    \sum_{t'\in \spdv[t]} t'
    \sim
    \sum_{u'\in \spdv[u]} u'.
  \]
  That is, for every elimination context $[\cdot]:A\vdash K:\tau$, we have
  \[
    \forall P,\quad
    {K[\sum_{t'\in \spdv[t]} t']}\twoheadrightarrow_{\escalar p}^P v
    \quad\textrm{iff}\quad
    {K[\sum_{u'\in \spdv[u]} u']}\twoheadrightarrow_{\escalar p}^P v.
  \]
  Or, since the only elimination context $[\cdot]:\tau\vdash K':\tau$ is $[\cdot]$,
  \begin{equation}
    \label{eq:goal}
    K[\sum_{t'\in \spdv[t]} t']
    \sim
    K[\sum_{u'\in \spdv[u]} u'].
  \end{equation}

  We proceed by induction on $K$.
  \begin{itemize}
    \item Let $K = [\cdot]$, we have two cases.
          \begin{itemize}
            \item If $A=\one$, then
                  by Theorem~\ref{thm:intros},
                  $\spdv[t]=\{\escalar p_i\bullet\escalar s_i.\star\}_{i\in I}$ and
                  $\spdv[u]=\{\escalar p'_j\bullet\escalar s'_j.\star\}_{j\in J}$.
                  Then,
                  $\sum_{t'\in \spdv[t]} t'\lra^*(\sum_i\escalar p_i\produ[\mathcal S]\escalar s_i).\star$ and
                  $\sum_{u'\in \spdv[u]} u'\lra^*(\sum_j\escalar p'_j\produ[\mathcal S]\escalar s'_j).\star$.

                  Then, using Corollary~\ref{cor:soundness}, we have
                  \begin{align*}
                    {\semS{\sum_i\escalar p_i\produ[\mathcal S]\escalar s_i}} 
										& =\sem{\vdash (\sum_i\escalar p_i\produ[\mathcal S]\escalar s_i).\star:\one}
                    = \sem{\vdash t:\one} 
                    =\sem{\vdash u:\one} \\
                    &=\sem{\vdash (\sum_j\escalar p'_j\produ[\mathcal S]\escalar s'_j).\star:\one}
                    =\semS{\sum_j\escalar p'_j\produ[\mathcal S]\escalar s'_j}. 
									 \end{align*}
                  Therefore, since $\semS{\cdot}$ is a monomorphism, $(\sum_i\escalar p_i\produ[\mathcal S]\escalar s_i).\star = (\sum_j\escalar p'_j\produ[\mathcal S]\escalar s'_j).\star$, thus $t\approx u$.
            \item
                  If $A=\top$, then by Theorem~\ref{thm:intros}, $t\lra^*\langle\rangle$ and $u\lra^*\langle\rangle$, and thus, $t\sim u$ and so $t\approx u$.
          \end{itemize}
    \item Let $K\neq[\cdot]$.
          By Corollary~\ref{cor:soundness}, we have
          \[
            \sem{\vdash\sum_{t'\in \spdv[t]} t':A} =\sem{\vdash t:A} = \sem{\vdash u:A} =\sem{\vdash\sum_{u'\in \spdv[u]} u':A}.
          \]
          Hence, by composition,
          \[
            \sem{\vdash K[\sum_{t'\in \spdv[t]} t']:\tau} =\sem{\vdash K[\sum_{u'\in \spdv[u]} u']:\tau}.
          \]
          Since $\tau$ is smaller than $A$, otherwise $K$ would have been $[\cdot]$, we can apply the induction hypothesis to conclude that
          \begin{equation}
            \label{eq:conc}
            K[\sum_{t'\in \spdv[t]} t'] \approx K[\sum_{u'\in \spdv[u]} u'].
          \end{equation}

          Since any reduction path started from $K[\sum_{t'\in \spdv[t]} t']$ or from $K[\sum_{u'\in \spdv[u]} u']$ use only reductions $\lra$, Equation~\eqref{eq:conc} implies
          $1_{\mathcal S}\bullet K[\sum_{t'\in \spdv[t]} t'] \sim 1_{\mathcal S}\bullet K[\sum_{u'\in \spdv[u]} u']$, which implies
          the Equation~\eqref{eq:goal}.
          \qedhere
  \end{itemize}
\end{proof}

\section{Conclusion}\label{sec:conclusion}
In this paper, we have presented a categorical characterisation for the proof language \OC, an extension of IMALL with the generalised probabilistic connective $\odot$. We have shown that the essential structure of a symmetric monoidal closed category with biproducts suffices for modelling the \OC, when there exists a monomorphism from the semiring of scalars to the semiring $\Hom II$. A key element in our approach was the abstract definition of the weighted codiagonal map, which underpins the representation of generalised probabilities.  We established soundness and adequacy proofs for this model.

In particular, Corollary~\ref{cor:soundness} gives a summary of the approach: the interpretation of a term is the same as the interpretation of the weighted linear combination of the values it achieves. The map $\sumwe pq$ gives us an abstract representation for this.

Our work offers an alternative approach to existing models relying on probabilistic coherence spaces, cones, or compactness requirements. 
It also generalises the model $\mathcal R^\Pi$ of PCF$^{\mathcal R}$ given in~\cite{LairdManzonettoMcCuskerPaganiLICS13} in two ways: we give a categorical characterisation, and we give a language where the sums and scalar product not only serve as a way to represent probabilities but also as a way to represent vectors and matrices. Furthermore, the categorical model for \OC paves the way for future investigations into the connections between linear logic, verifiable quantum algorithms, and the development of probabilistic programming languages.

\section*{Acknowledgments}
We thank Damiano Mazza for help us understand why the probabilistic coherence spaces are not an example of our model.

\bibliographystyle{elsarticle-harv} 
\bibliography{biblio.bib}

\appendix
\section{Proof of Theorem~\ref{thm:intros}}\label{app:intros}
\xrecap{Theorem}{Introduction}{thm:intros}
{
  Let $\vdash t:A$ and $t$ irreducible.
  \begin{itemize}
    \item If $A=\one$, then $t=\star$.
    \item If $A=B\otimes C$, then $t=u\otimes v$.
    \item If $A=B\multimap C$, then $t=\lambda x.u$.
    \item If $A=\top$, then $t=\langle\rangle$.
    \item $A$ cannot be equal to $\zero$.
    \item If $A=B\with C$, then $t=\pair uv$.
    \item If $A=B\oplus C$, then $t=\inl(l)$, $t=\inr(r)$.
    \item If $A=B\odot C$, then $t=\super uv$.
  \end{itemize}
}
\begin{proof}
  By induction on $t$. If $t$ is one of $\star$, $u\otimes v$, $\lambda x.u$, $\langle\rangle$, $\pair uv$, $\inl(l)$, $\inr(r)$, or $\super uv$, then we are done. 
  \begin{itemize}
    \item $t$ cannot be a variable or $\elimzero(u)$ since it is closed.
    \item Let $t=\elimone(u,v)$, then $\vdash u:\one$. Thus, by the induction hypothesis, $u=\star$, but then $\elimone(u,v)$ is reducible, which is absurd.

    \item Let $t=\elimtens(u,xy.v)$, then $\vdash u:A\otimes B$. Thus, by the induction hypothesis, $u=u_1\otimes u_2$, but then $\elimtens(u,xy.v)$ is reducible, which is absurd.

    \item Let $t=uv$, then $\vdash u:B\multimap A$ and $\vdash v:B$. Thus, by the induction hypothesis, $u=\lambda x.s$, but then $uv$ is reducible, which is absurd.

    \item Let $t=\fst(u)$, then $\vdash u:A\with B$. Thus, by the induction hypothesis, $u=\pair{s_1}{s_2}$, but then $\fst(u)$ is reducible, which is absurd.
    \item Let $t=\snd(u)$. This case is analogous to the previous one.
    \item Let $t=\elimoplus(u,x.s_1,y.s_2)$, then $\vdash u:B\oplus C$. Thus, by the induction hypothesis, $u=\inl(r)$ or $u=\inr(r)$, but then $\elimoplus(u,x.s_1,y.s_2)$ is reducible, which is absurd.
    \item Let $t=\fstsup(u)$, then $\vdash u:A\odot B$. Thus, by the induction hypothesis, $u=\super{s_1}{s_2}$, but then $\fstsup(u)$ is reducible, which is absurd.
    \item Let $t=\sndsup(u)$. This case is analogous to the previous one.

    \item Let $t=\elimsup(u,x.s_1,y.s_2)$, then $\vdash u:B\odot C$. Thus, by the induction hypothesis, $u=\super rv$, but then $\elimsup(u,x.s_1,y.s_2)$ is reducible, which is absurd.
      \qedhere
  \end{itemize}
\end{proof}

\section{Proof of Lemma~\ref{lem:distrib}}\label{app:distrib}
\recap{Lemma}{lem:distrib}{
  Let $F:\SM\to\SM$ be a semiadditive functor. Then there is a natural
  isomorphism
  \[
    F(A)\oplus F(B)\;\cong\;F(A\oplus B),
  \]
  where the arrows are given by
  \begin{align*}
    F(A\oplus B)\xlra{f} F(A)\oplus F(B) & \quad\text{with } f = \pair{F(\pi_1)}{F(\pi_2)},\\
    F(A)\oplus F(B)\xlra{f^{-1}} F(A\oplus B) & \quad\text{with } f^{-1} = \coproducto{F(i_1)}{F(i_2)}.
  \end{align*}
}

\begin{proof}~

  \begin{itemize}
      \item First we check that $f^{-1}\circ f=\Id$
	\begin{align*}
	  f^{-1}\circ f
	  &=f^{-1}\circ\Id\circ f\\
	  &=f^{-1}\circ(i_1\circ\pi_1\oplus i_2\circ\pi_2)\circ \pair{{F(\pi_1)}}{{F(\pi_2)}}\\
	  &=f^{-1}\circ(i_1\circ\pi_1\circ \pair{{F(\pi_1)}}{{F(\pi_2)}}\oplus i_2\circ\pi_2\circ \pair{{F(\pi_1)}}{{F(\pi_2)}})\\
	  &=f^{-1}\circ(i_1\circ{{F(\pi_1)}}\oplus i_2\circ{{F(\pi_2)}})\\
	  &=f^{-1}\circ i_1\circ{{F(\pi_1)}}\oplus f^{-1}\circ i_2\circ{{F(\pi_2)}} \\
	  &=\coproducto{{F(i_1)}}{{F(i_2)}}\circ i_1\circ F(\pi_1) \oplus \coproducto{{F(i_1)}}{{F(i_2)}}\circ i_2\circ F(\pi_2)\\
	  &={{F(i_1)}}\circ F(\pi_1) \oplus {{F(i_2)}}\circ F(\pi_2)\\
	  &=F(i_1\circ \pi_1) \oplus F(i_2\circ \pi_2)\\
	  &=F(i_1\circ \pi_1 \oplus  i_2\circ \pi_2)\\
	  &=F(\Id)\\
	  &=\Id
	\end{align*}
      \item To check $f\circ f^{-1}=\Id$ we check instead
	\[
	  \left\{
	    \begin{array}{l}
	      f\circ f^{-1}\circ i_1 = i_1\\
	      f\circ f^{-1}\circ i_2 = i_2
	    \end{array}
	  \right.
	\]
	For the first equation, we have
	\begin{align*}
	  f\circ f^{-1}\circ i_1
	  &=\pair{F(\pi_1)}{F(\pi_2)}\circ\coproducto{{F(i_1)}}{{F(i_2)}}\circ i_1\\
	  &=\pair{F(\pi_1)}{F(\pi_2)}\circ{{F(i_1)}}\\
	  &=\pair{F(\pi_1)\circ{{F(i_1)}}}{F(\pi_2)\circ{{F(i_1)}}}\\
	  &=\pair{F(\pi_1\circ i_1)}{F(\pi_2\circ i_1)}\\
	  &=\pair{F(\Id)}{F(0)}\\
	  &=\pair{\Id}{0}\\
	  (*) &=i_1
	\end{align*}
	Where the equality $(*)$ is given by the fact that $\pi_1\circ i_1 = \Id = \pi_1\circ\pair{\Id}{0}$ and $\pi_2\circ i_1 = 0 = \pi_2\circ\pair{\Id}0$.

	The second equation is analogous.
      \item We check that $f$ is a natural transformation.
	\[
	  \begin{tikzcd}[labels=description,row sep=1cm,column sep=2cm,
	      ]
	      F(A\oplus B)\ar[r,"f"]\ar[d,"F(g\oplus h)"] & F(A)\oplus F(B)\ar[d,"F(g)\oplus F(h)"]\\
	      F(C\oplus D)\ar[r,"f"] & F(C)\oplus F(D)
	  \end{tikzcd}
	\]
	\begin{align*}
	  f\circ F(g\oplus h)
	  &=\pair{F(\pi_1)}{F(\pi_2)}\circ F(g\oplus h)\\
	  &=\pair{F(\pi_1\circ(g\oplus h))}{F(\pi_2\circ(g\oplus h))}\\
	  &=\pair{F(g\circ\pi_1)}{F(h\circ\pi_2)}\\
	  &=\pair{F(g)\circ F(\pi_1)}{F(h)\circ F(\pi_2)}\\
	  &=\pair{F(g)\circ \pi_1\circ\pair{F(\pi_1)}{F(\pi_2)}}{F(h)\circ \pi_2\circ\pair{F(\pi_1)}{F(\pi_2)}}\\
	  &=\pair{F(g)\circ \pi_1\circ f}{F(h)\circ \pi_2\circ f}\\
	  &=\pair{F(g)\circ \pi_1}{F(h)\circ \pi_2}\circ f\\
	  &=(F(g)\oplus F(h))\circ f
	\end{align*}
      \item Finally, $f^{-1}$ is also a natural transformation, since it is the inverse of a natural transformation.
	\qedhere
  \end{itemize}
\end{proof}

\section{Proof of Lemma~\ref{lem:s-nat}}\label{app:s-nat}
\xrecap{Lemma}{Scalar}{lem:s-nat}
{
  For any map $I\xlra{s} I$, the map $\hat s_A$ is a natural transformation.
}

\begin{proof}
      Consequence of the commutation of the following diagram.
      \[
	\begin{tikzcd}[labels=description,row sep=1cm,column sep=3cm,
	    execute at end picture={
	      \path 
	      (\tikzcdmatrixname-1-1) -- (\tikzcdmatrixname-2-1) coordinate[pos=0.5](aux1)
	      (\tikzcdmatrixname-1-2) -- (\tikzcdmatrixname-2-2) coordinate[pos=0.5](aux2)
	      (\tikzcdmatrixname-1-3) -- (\tikzcdmatrixname-2-3) coordinate[pos=0.5](aux3)
	      (\tikzcdmatrixname-1-4) -- (\tikzcdmatrixname-2-4) coordinate[pos=0.5](aux4)
	      (aux1) -- (aux2) node[midway,red]{\small (Naturality of $\rho$)}
	      (aux2) -- (aux3) node[midway,red]{\small (Functoriality of $\otimes$)}
	      (aux3) -- (aux4) node[midway,red]{\small (Naturality of $\rho^{-1}$)};
	    }
	  ]
	  A\ar[r,"\rho"]\ar[d,"f"] & A\otimes I\ar[r,"\Id\otimes{s}"]\ar[d,dashed,"f\otimes\Id"] & A\otimes I\ar[r,"\rho^{-1}"]\ar[d,dashed,"f\otimes\Id"] & A\ar[d,"f"]\\
	  B\ar[r,"\rho"] & B\otimes I\ar[r,"\Id\otimes{s}"] & B\otimes I\ar[r,"\rho^{-1}"] & B\\
	\end{tikzcd}
      \]
\end{proof}

\section{Proof of Lemma~\ref{lem:PropsS}}\label{app:PropsS}
\xrecap{Lemma}{Properties of the scalar map}{lem:PropsS}
{
 Let $s$ be any map $I\xlra{s} I$.
 Then,
  \begin{enumerate}
    \item $\hat s_{I} = s$.
    \item $\hat s_{A\otimes B} = \hat s_A\otimes\Id_B$.
    \item $\hat s_{A\oplus B} = \hat s_A\oplus\hat s_B$.
  \end{enumerate}
}

\begin{proof}
  ~
  \begin{enumerate}
    \item Consequence of the commutation of the following diagram.
      \[
	\begin{tikzcd}[labels=description,column sep=2cm,
	    execute at end picture={
	      \path 
	      (\tikzcdmatrixname-1-1) -- (\tikzcdmatrixname-1-4) node[midway,red,yshift=4mm]{\small (Def.)}
	      (\tikzcdmatrixname-1-1) -- (\tikzcdmatrixname-1-4) node[midway,red,yshift=-4mm]{\small (Naturality of $\rho^{-1}$)};
	    }
	  ]
	  I
	  \arrow[rrr, "\hat s_I" description,sloped,
	    rounded corners,
	    to path={[pos=0.75]
	      -- ([yshift=2mm]\tikztostart.north)
	      |- ([yshift=5mm]\tikztotarget.north)\tikztonodes
	    -- (\tikztotarget.north)}
	  ]
	  \arrow[rrr, "s" description,sloped,
	    rounded corners,
	    to path={[pos=0.75]
	      -- ([yshift=-2mm]\tikztostart.south)
	      |- ([yshift=-5mm]\tikztotarget.south)\tikztonodes
	    -- (\tikztotarget.south)}
	  ]
	  \ar[r,dashed,"\rho",bend left=10]
	  & I\otimes I\ar[r,dashed,"\Id\otimes s"]\ar[l,dashed,"\rho^{-1}",bend left=10]
	  & I\otimes I\ar[r,dashed,"\rho^{-1}"]
	  & I
	\end{tikzcd}
      \]
    \item Consequence of the commutation of the following diagram.
      \[
	\begin{tikzcd}[labels=description,column sep=2cm,row sep=1cm,
	    execute at end picture={
	      \path 
	      (\tikzcdmatrixname-2-1) -- (\tikzcdmatrixname-3-1) coordinate[pos=0.5](aux2)
	      (\tikzcdmatrixname-2-2) -- (\tikzcdmatrixname-3-2) coordinate[pos=0.5](aux3)
	      (\tikzcdmatrixname-1-1) -- (\tikzcdmatrixname-3-1) node[sloped,pos=0.5,red,yshift=-1.1cm]{\small (Coherence)}
	      (aux2) -- (aux3) node[midway,red]{\small (Naturality of $\sigma$)}
	      (\tikzcdmatrixname-1-2) -- (\tikzcdmatrixname-3-2) node[sloped,pos=0.5,red,yshift=1.1cm]{\small (Coherence)};
	    }
	  ]
	  A\otimes B\ar[d,"\rho_{A\otimes B}"]\ar[dd,"\rho_A\otimes\Id",sloped,out=180,in=180] & A\otimes B\\
	  A\otimes B\otimes I\ar[r,"\Id\otimes\Id\otimes s"] & A\otimes B\otimes I\ar[u,"\rho^{-1}_{A\otimes B}"]  \\
	  A\otimes I\otimes B\ar[r,"\Id\otimes s\otimes\Id"]\ar[u,"\Id\otimes\sigma",dashed] & A\otimes I\otimes B\ar[uu,"\rho^{-1}\otimes\Id",sloped,out=0,in=0]\ar[u,"\Id\otimes\sigma",dashed]
	\end{tikzcd}
      \]
    \item Consequence of the commutation of the following diagram.
      \[
		\begin{tikzcd}[labels=description,column sep=2.3cm,row sep=1cm,
	    execute at end picture={
	      \path 
	      (\tikzcdmatrixname-2-1) -- (\tikzcdmatrixname-3-1) coordinate[pos=0.5](aux2)
	      (\tikzcdmatrixname-2-2) -- (\tikzcdmatrixname-3-2) coordinate[pos=0.5](aux3)
	      (\tikzcdmatrixname-1-1) -- (\tikzcdmatrixname-3-1) node[pos=0.75,red,xshift=-1cm]{\small (*)}
	      (aux2) -- (aux3) node[midway,red]{\small (Lemma~\ref{lem:distrib})}
	      (\tikzcdmatrixname-1-2) -- (\tikzcdmatrixname-3-2) node[pos=0.75,red,xshift=1cm]{\small (**)};
	    }
			]
		  A\oplus B\ar[dd,"\rho_A\oplus\rho_B",sloped,out=180,in=180]\ar[d,"\rho_{A\oplus B}"] & A\oplus B \\
		  (A\oplus B)\otimes I\ar[r,"\Id\otimes s"]\ar[d,"d",dashed] & (A\oplus B)\otimes I\ar[u,"\rho^{-1}_{A\oplus B}"]\ar[d,"d",dashed] &  \\
		  (A\otimes I)\oplus (B\otimes I)\ar[r,"(\Id\otimes s)\oplus(\Id\otimes s)"] & (A\otimes I)\oplus (B\otimes I)\ar[uu,"\rho_A^{-1}\oplus\rho_B^{-1}",sloped,out=0,in=0]
		\end{tikzcd}
      \]
      Where the commutation of the diagram ${\color{red}(*)}$ is proved as follows.
      \begin{align*}
	d\circ\rho_{A\oplus B}
	&=\pair{\pi_1\otimes\Id}{\pi_2\otimes\Id}\circ\rho_{A\oplus B}\\
	&=\pair{(\pi_1\otimes\Id)\circ\rho_{A\oplus B}}{(\pi_2\otimes\Id)\circ\rho_{A\oplus B}}\\
	\textrm{(Naturality of $\rho$)}&=\pair{\rho_A\circ\pi_1}{\rho_B\circ\pi_2}\\
	&=\rho_A\oplus\rho_B
      \end{align*}
      The commutation of the diagram ${\color{red}(**)}$ is a direct consequence of the commutation of the diagram ${\color{red}(*)}$.
      Indeed, since $d\circ\rho_{A\oplus B} = \rho_A\oplus\rho_B$, we have $(\rho_A^{-1}\oplus\rho_B^{-1})\circ d\circ\rho_{A\oplus B} = \Id$, thus, $(\rho_A^{-1}\oplus\rho_B^{-1})\circ d=\rho_{A\oplus B}^{-1}$.
      \qedhere
  \end{enumerate}
\end{proof}

\section{Proof of Lemma~\ref{lem:tauNatural}}\label{app:tauNatural}
\xrecap{Lemma}{The map $\tau$}{lem:tauNatural}
{
	The following map in the arrows of $\SM$ is a natural transformation with respect to $I$.
  \[
    \tau=\home AB\otimes I\xlra{\varphi_{A,\home AB\otimes I,B\otimes I}(\varepsilon\otimes\Id)}\home A{B\otimes I}
  \]
  where $\varphi_{A,\home AB\otimes I,B\otimes I}$ is the map given by the adjunction 
  \[
    \Hom{X\otimes Y}{Z}\overset{\varphi_{X,Y,Z}}{\underset{\varphi^{-1}_{X,Y,Z}}{\leftrightarrows}}\Hom{Y}{\home XZ}
  \]
by taking $X = A$, $Y=\home AB\otimes I$, and $Z=B\otimes I$.
}
\begin{proof}
  We need to prove the commutation of the following diagram.
  \[
    \begin{tikzcd}[labels=description,column sep=2cm,row sep=1cm]
      \home AB\otimes I\ar[r,"\tau"]\ar[d,"\Id\otimes s"] & \home A{B\otimes I}\ar[d,"\home A{\Id\otimes s}"]\\
      \home AB\otimes I\ar[r,"\tau"] & \home A{B\otimes I}
    \end{tikzcd}
  \]

  Since $\varphi$ is natural, the following diagram commutes.
  \[
    \begin{tikzcd}[labels=description,column sep=1.5cm,row sep=1.5cm]
      \Hom{A\otimes\home AB\otimes I}{B\otimes I}\ar[r,"\varphi"]\ar[d,"\Hom{\Id}{\Id\otimes s}"] & \Hom{\home AB\otimes I}{\home A{B\otimes I}}\ar[d,"\Hom{\Id}{\home A{\Id\otimes s}}"]\\
      \Hom{A\otimes\home AB\otimes I}{B\otimes I}\ar[r,"\varphi"] & \Hom{\home AB\otimes I}{\home A{B\otimes I}}
    \end{tikzcd}
  \]
  Thus, 
  $\home A{\Id\otimes s}\circ\tau=\home A{\Id\otimes s}\circ \varphi(\varepsilon\otimes\Id) = \varphi(\Hom{\Id}{\Id\otimes s}(\varepsilon\otimes\Id))$.

  Therefore, we must prove that
  \[
    \varphi(\Hom{\Id}{\Id\otimes s}(\varepsilon\otimes\Id)) = \varphi(\varepsilon\otimes\Id)\circ(\Id\otimes s) = \tau\circ(\Id\otimes s)
  \]
  We prove instead that
  \[
    \Hom{\Id}{\Id\otimes s}(\varepsilon\otimes\Id) =\varphi^{-1}(\tau\circ(\Id\otimes s))
  \]
  where $\varphi^{-1}_{X,Y,Z}(g) = X\otimes Y\xlra{\Id\otimes g} X\otimes\home XZ\xlra{\varepsilon}Z$.

  We have
  \begin{align*}
    \Hom{\Id}{\Id\otimes s}(\varepsilon\otimes\Id)
    &=(\Id\otimes s)\circ(\varepsilon\otimes\Id)\\
    &=\varepsilon\otimes s\\
    (*) &=\varepsilon\circ((\Id\otimes\tau)\circ(\Id\otimes\Id\otimes s))\\
    &=\varepsilon\circ(\Id\otimes(\tau\circ(\Id\otimes s)))\\
    &=\varphi^{-1}(\tau\circ(\Id\otimes s))
  \end{align*}
  where the equality $(*)$ is justified by the commutation of the following diagram.
  \[
    \begin{tikzcd}[labels=description,column sep=3cm,row sep=2cm,
	execute at end picture={
	  \path 
	  (\tikzcdmatrixname-2-1) -- (\tikzcdmatrixname-1-2) coordinate[pos=0.5](aux)
	  (\tikzcdmatrixname-1-1) -- (aux) node[pos=0.5,red]{\small (Functoriality of $\otimes$)}
	  (aux) -- (\tikzcdmatrixname-2-2) node[pos=0.5,red]{\small ($\varepsilon\otimes\Id=\varphi^{-1}(\tau)$)};
	}
      ]
      A\otimes \home AB\otimes I\ar[d,"\Id\otimes\Id\otimes s"]\ar[r,"\varepsilon\otimes s"] & B\otimes I\\
      A\otimes\home AB\otimes I\ar[ru,"\varepsilon\otimes\Id",sloped]\ar[r,"\Id\otimes\tau"] & A\otimes \home A{B\otimes I}\ar[u,"\varepsilon"]
    \end{tikzcd}
  \]
\end{proof}

%

\section{Proof of Lemma~\ref{lem:sumwe-nat}}\label{app:sumwe-nat}
\xrecap{Lemma}{Weighted codiagonal}{lem:sumwe-nat}
{
  Let $I\xlra p I$ and $I\xlra q I$ be two maps.
  The map
    $A\oplus A\xlra{\sumwe{p}{q}} A$ defined by $\sumwe{p}{q} = \coproducto{\hat p}{\hat q}$
	is a natural transformation.
}

\begin{proof}
      Consequence of the commutation of the following diagram.
      \[
	\begin{tikzcd}[labels=description,column sep=2cm,row sep=1cm]
	  A\oplus A\ar[d,"f\oplus f"]\ar[r,"\sumwe pq"] & A\ar[d,"f"]\\
	  B\oplus B\ar[r,"\sumwe pq"] & B
	\end{tikzcd}
      \]
      We have
      \begin{align*}
	f\circ\sumwe pq
	&=f\circ\coproducto{\hat p}{\hat q}\\
	&=\coproducto{f\circ\hat p}{f\circ\hat q}\\
	\textrm{(Lemma~\ref{lem:s-nat})}&=\coproducto{\hat p\circ f}{\hat q\circ f}\\
	&= \coproducto{\coproducto{\hat p}{\hat q}\circ i_1\circ f}{\coproducto{\hat p}{\hat q}\circ i_2\circ f}\\
	&= \coproducto{\hat p}{\hat q}\circ\coproducto{i_1\circ f}{i_2\circ f}\\
	& = \sumwe pq\circ (f\oplus f)
	\qedhere
      \end{align*}
\end{proof}
\section{Proof of Lemma~\ref{lem:PropsSumwe}}\label{app:PropsSumwe}
\recap{Lemma}{lem:PropsSumwe}
{
	For any $I\xlra p I$ and $I\xlra q I$, we have
  \begin{enumerate}
	\item $(\sumwe pq\oplus\sumwe pq)\circ(\Id\oplus\sigma\oplus\Id) = \sumwe pq$.
	\item 
  $(\Id\oplus\sigma\oplus\Id)\circ(\Delta\oplus\Delta)=\Delta$
  \end{enumerate}
}

\begin{proof}
	~
	\begin{enumerate}
		\item 
    Consequence of the commutation of the following diagram.
      \[
	\begin{tikzcd}[labels=description,row sep=1cm,column sep=1.5cm]
	  (A\oplus A)\oplus(B\oplus B)\ar[rd,"\sumwe pq\oplus\sumwe pq",sloped]&&\ar[ll,"\Id\oplus\sigma\oplus\Id"] (A\oplus B)\oplus(A\oplus B)\ar[dl,"\sumwe pq",sloped]\\
	  &A\oplus B
	\end{tikzcd}
      \]
      We have
	  \begin{align*}
		(\sumwe{p}{q}\oplus\sumwe{p}{q})
		\circ(\Id\oplus\sigma\oplus\Id)
	&=\coproducto{\hat p_A}{\hat q_A}\oplus\coproducto{\hat p_B}{\hat q_B}
		\circ(\Id\oplus\sigma\oplus\Id)
		\\
	&= (\nabla\circ({\hat p_A}\oplus{\hat q_A}))\oplus (\nabla\circ({\hat p_B}\oplus{\hat q_B}))
		\circ(\Id\oplus\sigma\oplus\Id)
		\\
	 &= \nabla\circ({\hat p_A}\oplus{\hat q_A}\oplus {\hat p_B}\oplus{\hat q_B})\circ(\Id\oplus\sigma\oplus\Id)
	 \\
	 &= \nabla\circ({\hat p_A}\oplus{\hat p_B}\oplus {\hat q_A}\oplus{\hat q_B})
		\\
	\textrm{(Lemma~\ref{lem:PropsS})}&=\nabla\circ({\hat p_{A\oplus B}}\oplus{\hat q_{A\oplus B}})\\
	&=\coproducto{\hat p_{A\oplus B}}{\hat q_{A\oplus B}}\\
	&=\sumwe pq
	  \end{align*}
	\item 
  Consequence of the commutation of the following diagram
  \[
    \begin{tikzcd}[labels=description,column sep=1cm,row sep=1cm]
	A\oplus B\ar[rr,"\Delta\oplus\Delta"]\ar[dr,"\Delta",sloped]&& (A\oplus A)\oplus (B\oplus B)\ar[dl,"\Id\oplus\sigma\oplus\Id",sloped]\\
	&(A\oplus B)\oplus (A\oplus B)
    \end{tikzcd}
  \]
  To check
  $(\Id\oplus\sigma\oplus\Id)\circ(\Delta\oplus\Delta)=\Delta$,
   we check instead
  \[
    \left\{
      \begin{array}{l}
	\pi_1\circ(\Id\oplus\sigma\oplus\Id)\circ(\Delta\oplus\Delta)=\pi_1\circ\Delta\\
	\pi_2\circ(\Id\oplus\sigma\oplus\Id)\circ(\Delta\oplus\Delta)=\pi_2\circ\Delta
      \end{array}
    \right.
  \]
  We have
  \begin{align*}
    &\pi_{A\oplus B,A\oplus B}^1\circ(\Id\oplus\sigma\oplus\Id)\circ(\Delta\oplus\Delta)\\
    (*)&=(\pi_{A,A}^1\oplus \pi_{B,B}^1)\circ(\Delta\oplus\Delta)\\
    &=(\pi_{A,A}^1\circ\Delta)\oplus (\pi_{B,B}^1\circ\Delta)\\
    &=\Id_A\oplus\Id_A\\
    &=\Id_{A\oplus B}\\
    &=\pi_{A\oplus B,A\oplus B}^1\circ\Delta
  \end{align*}
  where the equality $(*)$ is justified as follows, using the fact that 
  \begin{equation}
    \label{eq}
    f_A\oplus g_B = \pair{f_A\circ\pi^1_{A\oplus B}}{g_B\circ \pi^2_{A\oplus B}}
  \end{equation}
  \begin{align*}
    &\pi_{A\oplus B,A\oplus B}^1\circ(\Id_A\oplus\sigma_{A\oplus B}^{B\oplus A}\oplus\Id_B)\\
    &=\pi_{A\oplus B,A\oplus B}^1\circ(\Id_A\oplus\pair{\pi^2_{A,B}}{\pi^1_{A,B}}\oplus\Id_B)\\
   \eqref{eq}
	&=\pi_{A\oplus B,A\oplus B}^1\circ(
      \pair{\Id_A\circ\pi^1_{A,A\oplus B}}
      {\pair{\pi^2_{A,B}}{\pi^1_{A,B}}\circ\pi^2_{A,A\oplus B}}
    \oplus\Id_B)\\
    &=\pi_{A\oplus B,A\oplus B}^1\circ(
      \pair{\pi^1_{A,A\oplus B}}
      {\pair{\pi^2_{A,B}}{\pi^1_{A,B}}\circ\pi^2_{A,A\oplus B}}
    \oplus\Id_B)\\
    \eqref{eq}
	&=
	\begin{aligned}[t]
	&\pi_{A\oplus B,A\oplus B}^1\circ
	\\
	&
      \pair{
	\pair{\pi^1_{A,A\oplus B}}{\pair{\pi^2_{A,B}}{\pi^1_{A,B}}\circ\pi^2_{A,A\oplus B}}
	\circ\pi^1_{A\oplus A\oplus B,B}
      }
      {
	\Id_B
	\circ\pi^2_{A\oplus A\oplus B,B}
      }
	\end{aligned}
    \\
    &=
	\begin{aligned}[t]
	&\pi_{A\oplus B,A\oplus B}^1\circ
	\\
	&
      \pair{
	\pair{\pi^1_{A,A\oplus B}}{\pair{\pi^2_{A,B}}{\pi^1_{A,B}}\circ\pi^2_{A,A\oplus B}}
      \circ\pi^1_{A\oplus A\oplus B,B}
      }
      {
	\pi^2_{A\oplus A\oplus B,B}
      }
	\end{aligned}
    \\
    &=
	\begin{aligned}[t]
	&\pi_{A\oplus B,A\oplus B}^1\circ
	\\
	&
      \pair{
	\pair
	{\pi^1_{A,A\oplus B}\circ\pi^1_{A\oplus A\oplus B,B}}
	{\pair{\pi^2_{A,B}}{\pi^1_{A,B}}\circ\pi^2_{A,A\oplus B}\circ\pi^1_{A\oplus A\oplus B,B}}
      }
      {
	\pi^2_{A\oplus A\oplus B,B}
      }
	\end{aligned}
    \\
    &=
	\begin{aligned}[t]
	&\pi_{A\oplus B,A\oplus B}^1\circ
	\\
	\langle
	\langle
	&
	{\pi^1_{A,A\oplus B}\circ\pi^1_{A\oplus A\oplus B,B}}
	,
	\\
	&
	{
	  \pair
	  {\pi^2_{A,B}\circ\pi^2_{A,A\oplus B}\circ\pi^1_{A\oplus A\oplus B,B}}
	  {\pi^1_{A,B}\circ\pi^2_{A,A\oplus B}\circ\pi^1_{A\oplus A\oplus B,B}}
	}
	\rangle
	  ,
	  \\
	  &
      {
	\pi^2_{A\oplus A\oplus B,B}
      }
	  \rangle
	\end{aligned}
    \\
    &=
    \pair
    {\pi^1_{A,A\oplus B}\circ\pi^1_{A\oplus A\oplus B,B}}
    {\pi^2_{A,B}\circ\pi^2_{A,A\oplus B}\circ\pi^1_{A\oplus A\oplus B,B}}
    \\
    &=
    \pair
    {\pi^1_{A,A\oplus B\oplus B}}
    {\pi^2_{A,B}\circ\pi^2_{A\oplus A,B\oplus B}}
    \\
    &=\pair
    {
      \pi_{A,A}^1\circ\pi^1_{A\oplus A,B\oplus B}
    }
    {
      \pi_{B,B}^1\circ\pi^2_{A\oplus A,B\oplus B}
    }\\
    &=\pi_{A,A}^1\oplus \pi_{B,B}^1
  \end{align*}
  The case with $\pi^2$ is analogous.
  \qedhere
	\end{enumerate}
\end{proof}

\section{Proof of Lemma~\ref{lem:prop-times-sumwe}}\label{app:prop-times-sumwe}

\recap{Lemma}{lem:prop-times-sumwe}
{
  If $(p,q)\in\we$, then
  $\sumwe pq\circ\delta = \sumwe pq\oplus\Id$.
}

\begin{proof}
  Consequence of the commutation of the following diagram.
  \[
    \begin{tikzcd}[labels=description,row sep=1cm,column sep=2cm,
	execute at end picture={
	  \path 
	  (\tikzcdmatrixname-1-2) -- (\tikzcdmatrixname-2-2) coordinate[pos=0.5](aux)
	  (\tikzcdmatrixname-1-1) -- (aux) node[pos=0.5,red,sloped]{\small (Lemma~\ref{lem:nablaPQInvDelta})}
	  (aux) -- (\tikzcdmatrixname-1-3) node[pos=0.45,red,sloped]{\small (Lemma~\ref{lem:PropsSumwe})}
	  (\tikzcdmatrixname-1-1) -- (\tikzcdmatrixname-1-3) node[pos=0.5,red,sloped,yshift=5mm]{\small (Def.)};
	}
      ]
      (A\oplus A)\oplus B\ar[dr,sloped,"\sumwe pq\oplus\Id"]\ar[rr,"\delta",bend left=15]\ar[r,dashed,"\Id\oplus\Delta"] & (A\oplus A)\oplus(B\oplus B)\ar[r,dashed,"\Id\oplus\sigma\oplus\Id"],\ar[d,"\sumwe pq\oplus\sumwe pq",dashed]& (A\oplus B)\oplus(A\oplus B)\ar[dl,"\sumwe pq",sloped]\\
      &A\oplus B
    \end{tikzcd}
  \]
\end{proof}

\section{Proof of Lemma~\ref{lem:deltaDelta}}\label{app:deltaDelta}
\recap{Lemma}{lem:deltaDelta}
{
  $\Delta = \delta\circ(\Delta\oplus\Id)$.
}

\begin{proof}
  Consequence of the commutation of the following diagram.
  \[
    \begin{tikzcd}[labels=description,column sep=1cm,row sep=1cm]
      A\oplus B\ar[rr,"\Delta"]\ar[dr,"\Delta\oplus\Id",sloped]&& (A\oplus B)\oplus (A\oplus B)\\
      &(A\oplus A)\oplus B\ar[ur,"\delta",sloped]
    \end{tikzcd}
  \]
  We have
  \begin{align*}
    \delta\circ(\Delta\oplus\Id)
    &=(\Id\oplus\sigma\oplus\Id)\circ(\Id\oplus\Delta)\circ(\Delta\oplus\Id)\\
    &=(\Id\oplus\sigma\oplus\Id)\circ(\Delta\oplus\Delta)\\
    \textrm{(Lemma~\ref{lem:DeltaSigma})}&=\Delta
    \qedhere
  \end{align*}
\end{proof}
\section{Proof of Lemma~\ref{lem:subst}}\label{app:subst}
\noindent{\bf Lemma~\ref{lem:subst} (Substitution)}{\bf .}
{\em
  If $\Gamma,x:A\vdash t:B$ and $\Theta\vdash v:A$, then
  $\sem{\Gamma\vdash (v/x)t:B}=\sem{\Gamma,x:A\vdash t:B}\circ (\Id\otimes\sem{\Gamma\vdash v:A})$.
  That is, the following diagram commutes.
  \begin{center}
    \begin{tikzcd}[labels=description,row sep=5mm,column sep=2cm]
      \sem\Gamma\otimes\sem\Theta\ar[dr,"\Id\otimes v",sloped]\ar[rr,"(v/x)t"] && \sem B\\
      &\sem\Gamma\otimes \sem A\ar[ru,"t",sloped]
    \end{tikzcd}
  \end{center}
}

\begin{proof}
  By induction on $t$. 
  To avoid cumbersome notation, we write $A$ instead of $\sem A$.
  \begin{itemize}
    \item 
      Let $t=x$. Then, $\Gamma=\emptyset$, and $A=B$. Then, the commuting diagram is the following.
      \[
	\begin{tikzcd}[labels=description,row sep=1cm,column sep=2cm]
	   I\otimes\Theta\cong\Theta\ar[dr,"\Id\otimes v",sloped]\ar[rr,"(v/x)x=v"] &&  A\\
	  & I\otimes  A= A\ar[ru,"\Id",sloped]
	\end{tikzcd}
      \]

    \item The case $t=y\neq x$ is no possible since $\Gamma,x:A\neq y:B$.
    \item Let $\escalar s.\star$, it is not possible since $\Gamma,x:A\neq\emptyset$.
    \item Let $t=r\plus u$.
      Then, the commuting diagram is the following.
      
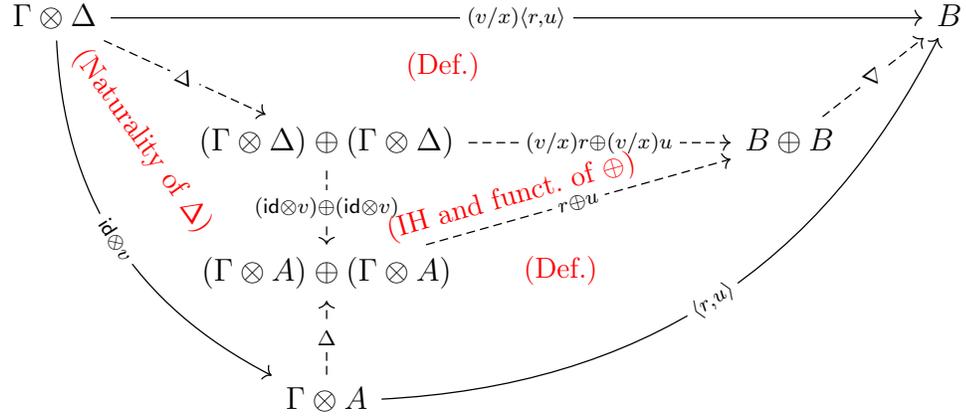
\[
	\begin{tikzcd}[labels=description,column sep=3.5cm,row sep=1cm,
	    execute at end picture={
	      \path (\tikzcdmatrixname-1-1) -- (\tikzcdmatrixname-1-4) coordinate[pos=0.5](aux)
	      (aux) -- (\tikzcdmatrixname-2-2) node[midway,red]{\small (Def.)};
	      \path (\tikzcdmatrixname-1-1) -- (\tikzcdmatrixname-4-2) node[pos=0.3,red,sloped]{\small (Naturality of $\Delta$)};
	      \path (\tikzcdmatrixname-3-2) -- (\tikzcdmatrixname-2-3) node[pos=0.3,red,sloped,yshift=3mm]{\small (IH and funct.~of $\oplus$)};
	      \path (\tikzcdmatrixname-4-2) -- (\tikzcdmatrixname-2-3) node[midway,red]{\small (Def.)};
	    }
	  ]
	  \Gamma\otimes\Theta&[-2.5cm]&&[-2.5cm] B \\
	  & {(\Gamma\otimes\Theta)\oplus(\Gamma\otimes\Theta)} & {B\oplus B} \\
	  & {(\Gamma\otimes A)\oplus(\Gamma\otimes A)} \\
	  & {\Gamma\otimes A}
	  \arrow["\Delta", dashed, from=4-2, to=3-2]
	  \arrow["{\Id\otimes v}", from=1-1, to=4-2,sloped,bend right]
	  \arrow["{(\Id\otimes v)\oplus(\Id\otimes v)}"', dashed, from=2-2, to=3-2]
	  \arrow["\Delta"', dashed, from=1-1, to=2-2,sloped]
	  \arrow["{\pair ru}"', from=4-2, to=1-4,sloped,bend right]
	  \arrow["\nabla", dashed, from=2-3, to=1-4,sloped]
	  \arrow["{(v/x)r\oplus (v/x)u}", dashed, from=2-2, to=2-3]
	  \arrow["{r\oplus u}", dashed, from=3-2, to=2-3,sloped]
	  \arrow["{(v/x)\pair ru}", from=1-1, to=1-4]
	\end{tikzcd}
      \]

    \item Let $t=\escalar s\bullet r$.
      Then, the commuting diagram is the following.

      
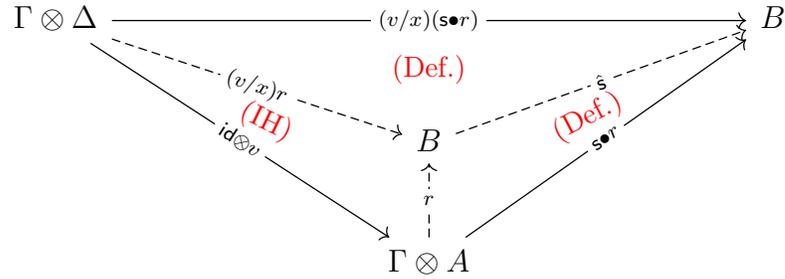
\[
	\begin{tikzcd}[labels=description,column sep=3.5cm,row sep=1cm,
	    execute at end picture={
	      \path 
	      (\tikzcdmatrixname-1-1) -- (\tikzcdmatrixname-1-3) coordinate[midway](aux1)
	      (\tikzcdmatrixname-2-2) -- (\tikzcdmatrixname-3-2) coordinate[midway](aux2)
	      (aux1) -- (\tikzcdmatrixname-2-2) node[midway,red]{\small (Def.)}
	      (\tikzcdmatrixname-1-1) -- (aux2) node[midway,red,sloped]{\small (IH)}
	      (aux2) -- (\tikzcdmatrixname-1-3) node[midway,red,sloped]{\small (Def.)};
	    }
	  ]
	  \Gamma\otimes\Theta&& B \\
	  & B \\
	  & {\Gamma\otimes A}
	  \arrow["{\Id\otimes v}", from=1-1, to=3-2,sloped]
	  \arrow["{\escalar s\bullet r}"', from=3-2, to=1-3,sloped]
	  \arrow["{(v/x)(\escalar s\bullet r)}"{description}, from=1-1, to=1-3]
	  \arrow["r", dashed, from=3-2, to=2-2]
	  \arrow["{\hat{\escalar s}}", dashed, from=2-2, to=1-3,sloped]
	  \arrow["{(v/x)r}", dashed, from=1-1, to=2-2,sloped]
	\end{tikzcd}
      \]
    \item Let $t=\elimone(r,u)$, so $\Gamma=\Gamma_1,\Gamma_2$.

	\begin{itemize}
	  \item Let $x\in FV(r)$, then, the commuting diagram is the following.
	    
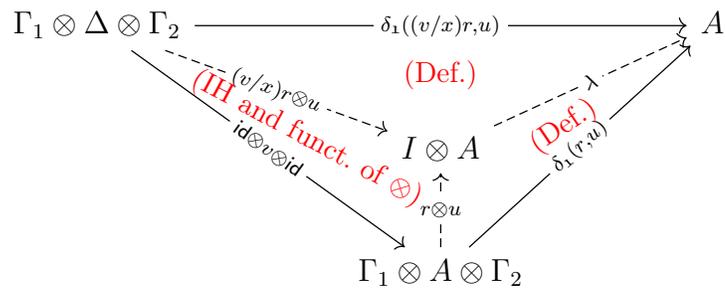
\[
	      \begin{tikzcd}[labels=description,row sep=1cm,column sep=2cm,
		  execute at end picture={
		    \path (\tikzcdmatrixname-1-1) -- (\tikzcdmatrixname-1-3) coordinate[pos=0.5](aux)
		    (aux) -- (\tikzcdmatrixname-2-2) node[midway,red]{\small (Def.)};
		    \path (\tikzcdmatrixname-2-2) -- (\tikzcdmatrixname-3-2) coordinate[pos=0.5](aux)
		    (aux) -- (\tikzcdmatrixname-1-3) node[midway,red,sloped]{\small (Def.)};
		    \path (\tikzcdmatrixname-1-1) -- (aux) node[pos=0.55,red,sloped]{\small (IH and funct.~of $\otimes$)};
		  }
		]
		{{\Gamma_1}\otimes{\Theta}\otimes{\Gamma_2}} && { A} \\
		& { I\otimes{A}} \\
		& {{\Gamma_1}\otimes A\otimes{\Gamma_2}}
		\arrow["{\elimone((v/x)r,u)}", from=1-1, to=1-3]
		\arrow["{\Id\otimes v\otimes\Id}"', from=1-1, to=3-2,sloped]
		\arrow["{\elimone(r,u)}"', from=3-2, to=1-3,sloped]
		\arrow["{(v/x)r\otimes u}"', from=1-1, to=2-2,sloped,dashed]
		\arrow["\lambda"', from=2-2, to=1-3,sloped,dashed]
		\arrow["{r\otimes u}"{description}, from=3-2, to=2-2,dashed]
	      \end{tikzcd}
	    \]

	  \item
	    Let $x\in FV(u)$, then, the commuting diagram is the following.

	    
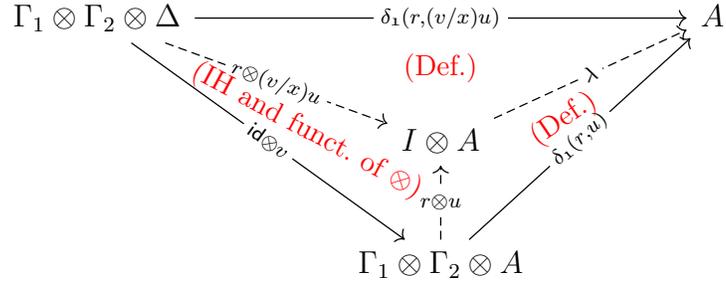
\[
	      \begin{tikzcd}[labels=description,row sep=1cm,column sep=2cm,
		  execute at end picture={
		    \path (\tikzcdmatrixname-1-1) -- (\tikzcdmatrixname-1-3) coordinate[pos=0.5](aux)
		    (aux) -- (\tikzcdmatrixname-2-2) node[midway,red]{\small (Def.)};
		    \path (\tikzcdmatrixname-2-2) -- (\tikzcdmatrixname-3-2) coordinate[pos=0.5](aux)
		    (aux) -- (\tikzcdmatrixname-1-3) node[midway,red,sloped]{\small (Def.)};
		    \path (\tikzcdmatrixname-1-1) -- (aux) node[pos=0.55,red,sloped]{\small (IH and funct.~of $\otimes$)};
		  }
		]
		{{\Gamma_1}\otimes{\Gamma_2}\otimes{\Theta}} && { A} \\
		& { I\otimes{A}} \\
		& {{\Gamma_1}\otimes{\Gamma_2}\otimes A}
		\arrow["{\elimone(r,(v/x)u)}", from=1-1, to=1-3]
		\arrow["{\Id\otimes v}"', from=1-1, to=3-2,sloped]
		\arrow["{\elimone(r,u)}"', from=3-2, to=1-3,sloped]
		\arrow["{r\otimes (v/x)u}"', from=1-1, to=2-2,sloped,dashed]
		\arrow["\lambda"', from=2-2, to=1-3,sloped,dashed]
		\arrow["{r\otimes u}"{description}, from=3-2, to=2-2,dashed]
	      \end{tikzcd}
	    \]
	\end{itemize}

      \item Let $t=r\otimes u$.
	  \begin{itemize}
	    \item Let $x\in FV(u)$, so $\Gamma=\Gamma_1,\Gamma_2$. Then, the commuting diagram is the following.
	      
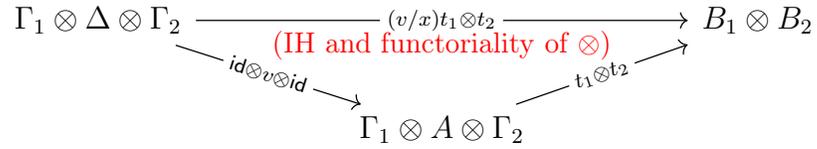
\[
		\begin{tikzcd}[labels=description,column sep=2cm,row sep=8mm,
		    execute at end picture={
		      \path (\tikzcdmatrixname-1-1) -- (\tikzcdmatrixname-1-3) coordinate[midway](aux)
		      (aux) -- (\tikzcdmatrixname-2-2) node[pos=0.3,red]{\small (IH and functoriality of $\otimes$)};
		    }
		  ]
		  {{\Gamma_1}\otimes{\Theta}\otimes{\Gamma_2}} && {{B_1}\otimes{B_2}} \\
		  & {{\Gamma_1}\otimes{A}\otimes{\Gamma_2}}
		  \arrow["{t_1\otimes t_2}", from=2-2, to=1-3,sloped]
		  \arrow["{\Id\otimes v\otimes\Id}", from=1-1, to=2-2,sloped]
		  \arrow["{(v/x)t_1\otimes t_2}", from=1-1, to=1-3]
		\end{tikzcd}
	      \]

	    \item Let $x\in FV(r)$. This case is analogous to the previous one.
	  \end{itemize}
	
      \item Let $t=\elimtens(r,yz.u)$. Then $\Gamma=\Gamma_1,\Gamma_2$.
	  \begin{itemize}
	    \item Let $x\in FV(r)$. Then, the commuting diagram is the following
	      
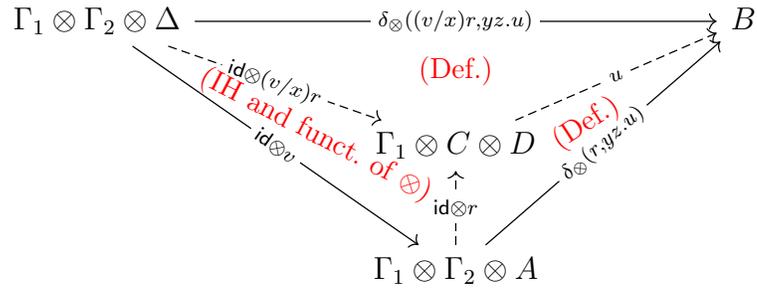
\[
		\begin{tikzcd}[labels=description,row sep=1cm,column sep=2.2cm,
		  execute at end picture={
		    \path (\tikzcdmatrixname-1-1) -- (\tikzcdmatrixname-1-3) coordinate[pos=0.5](aux)
		    (aux) -- (\tikzcdmatrixname-2-2) node[midway,red]{\small (Def.)};
		    \path (\tikzcdmatrixname-2-2) -- (\tikzcdmatrixname-3-2) coordinate[pos=0.5](aux)
		    (aux) -- (\tikzcdmatrixname-1-3) node[midway,red,sloped]{\small (Def.)};
		    \path (\tikzcdmatrixname-1-1) -- (aux) node[pos=0.55,red,sloped]{\small (IH and funct.~of $\otimes$)};
		  }
		]
		  {{\Gamma_1}\otimes{\Gamma_2}\otimes{\Theta}} && {{B}} \\
		  & {{\Gamma_1}\otimes{C}\otimes{D}} \\
		  & {{\Gamma_1}\otimes{\Gamma_2}\otimes{A}}
		  \arrow["{\elimtens(r,yz.u)}", from=3-2, to=1-3,sloped]
		  \arrow["{\Id\otimes v}", from=1-1, to=3-2,sloped]
		  \arrow["{\elimtens((v/x)r,yz.u)}", from=1-1, to=1-3]
		  \arrow["{\Id\otimes (v/x)r}", from=1-1, to=2-2,sloped,dashed]
		  \arrow["u", from=2-2, to=1-3,sloped,dashed]
		  \arrow["{\Id\otimes r}", from=3-2, to=2-2,dashed]
		\end{tikzcd}
	      \]

	    \item 
	      Let $x\in FV(u)$.
	      Then, the commuting diagram is the following
	      
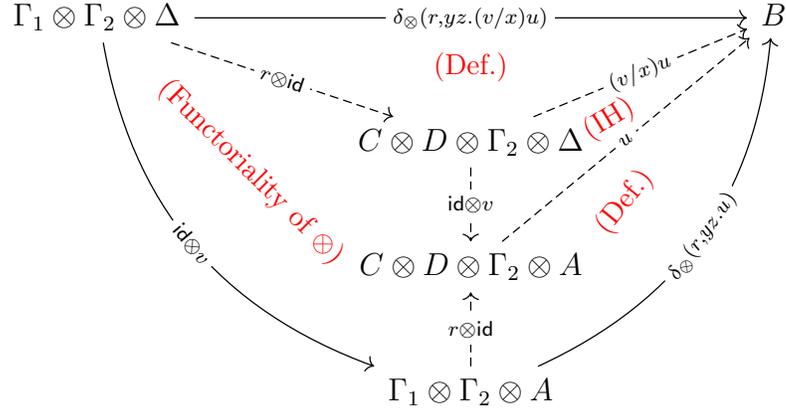
\[
		\begin{tikzcd}[labels=description,column sep=2cm,row sep=1cm,
		  execute at end picture={
		    \path (\tikzcdmatrixname-1-1) -- (\tikzcdmatrixname-1-3) coordinate[pos=0.5](aux)
		    (aux) -- (\tikzcdmatrixname-2-2) node[midway,red]{\small (Def.)};
		    \path (\tikzcdmatrixname-2-2) -- (\tikzcdmatrixname-3-2) coordinate[pos=0.5](aux)
		    (aux) -- (\tikzcdmatrixname-1-3) node[midway,red,sloped]{\small (IH)};
		    \path (\tikzcdmatrixname-4-2) -- (\tikzcdmatrixname-1-3) node[midway,red,sloped]{\small (Def.)};
		    \path (\tikzcdmatrixname-1-1) -- (\tikzcdmatrixname-4-2) node[pos=0.4,red,sloped]{\small (Functoriality of $\otimes$)};
		  }
		  ]
		  {{\Gamma_1}\otimes{\Gamma_2}\otimes{\Theta}} && B \\
		  & {{C}\otimes{D}\otimes{\Gamma_2}\otimes{\Theta}} \\
		  & {{C}\otimes{D}\otimes{\Gamma_2}\otimes{A}} \\
		  & {{\Gamma_1}\otimes{\Gamma_2}\otimes{A}}
		  \arrow["{\elimtens(r,yz.u)}"{description}, from=4-2, to=1-3,sloped,bend right]
		  \arrow["{\Id\otimes v}"{description}, from=1-1, to=4-2,sloped,bend right]
		  \arrow["{\elimtens(r,yz.(v/x)u)}"{description}, from=1-1, to=1-3]
		  \arrow["r\otimes\Id"{description}, from=1-1, to=2-2,sloped,dashed]
		  \arrow["{(v/x)u}"{description}, from=2-2, to=1-3,sloped,dashed]
		  \arrow["r\otimes\Id"{description}, from=4-2, to=3-2,dashed]
		  \arrow["{\Id\otimes v}"{description}, from=2-2, to=3-2,dashed]
		  \arrow["u"{description}, from=3-2, to=1-3,sloped,dashed]
		\end{tikzcd}
	      \]
	  \end{itemize}

    \item 
      Let $t=\lambda{y}.r$, so $B=C\multimap D$. Then, the commuting diagram is the following.
      
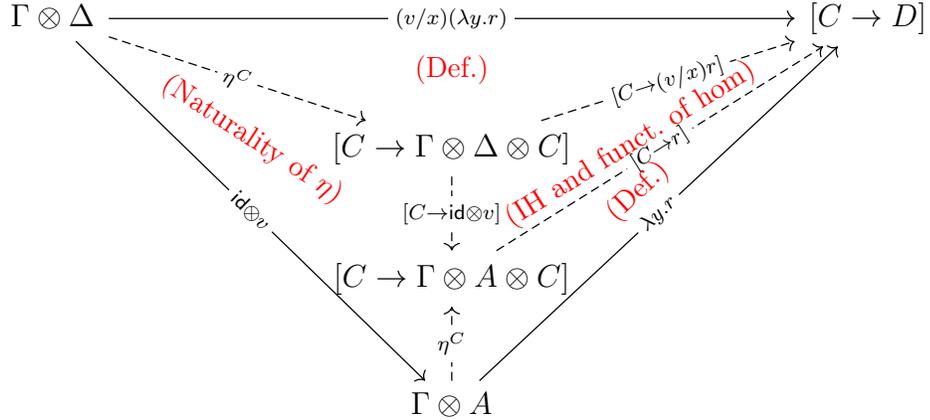
\[
	\begin{tikzcd}[labels=description,column sep=2.8cm,row sep=1cm,
	    execute at end picture={
	      \path (\tikzcdmatrixname-1-1) -- (\tikzcdmatrixname-1-3) coordinate[pos=0.5](aux)
	      (aux) -- (\tikzcdmatrixname-2-2) node[midway,red]{\small (Def.)};
	      \path (\tikzcdmatrixname-1-1) -- (\tikzcdmatrixname-3-2) node[midway,red,sloped,yshift=1mm]{\small (Naturality of $\eta$)};
	      \path (\tikzcdmatrixname-3-2) -- (\tikzcdmatrixname-1-3) node[midway,red,sloped,yshift=2.5mm,xshift=-3mm]{\small (IH and funct.~of hom)};
	      \path (\tikzcdmatrixname-4-2) -- (\tikzcdmatrixname-1-3) node[midway,red,sloped,yshift=4mm]{\small (Def.)};
	    }
	  ]
	  \Gamma\otimes\Theta\ar[dddr,"\Id\otimes v",sloped]\ar[rr,"(v/x)(\lambda y.r)"]\arrow[rd,"\eta^{  C}",dashed,sloped] & &\home{  C}{  D}\\
	  & \home{  C}{\Gamma\otimes\Theta\otimes  C}\arrow[ur,dashed,"\home{  C}{(v/x)r}",sloped]\ar[d,dashed,"\home{  C}{\Id\otimes v}"] & \\
	  & \home{  C}{\Gamma\otimes  A\otimes  C}\ar[uur,"\home{  C}r",dashed,sloped]& \\
	  & \Gamma\otimes  A\ar[u,"\eta^{  C}",dashed]\ar[uuur,"\lambda y.r",sloped]&
	\end{tikzcd}
      \]

    \item
      Let $t=ru$, so $\Gamma=\Gamma_1,\Gamma_2$.
      \begin{itemize}
	\item Let $x\in FV(r)$, so $\Gamma_1\vdash u:C$ and $\Gamma_2,x:A\vdash r:C\multimap B$.
	  Then, the commuting diagram is the following
	  
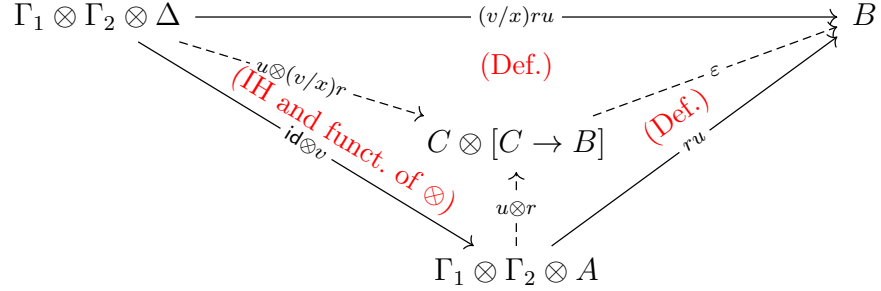
\[
	    \begin{tikzcd}[labels=description,column sep=2.9cm,row sep=1cm,
		execute at end picture={
		  \path (\tikzcdmatrixname-1-1) -- (\tikzcdmatrixname-1-3) coordinate[pos=0.5](aux)
		  (aux) -- (\tikzcdmatrixname-2-2) node[midway,red]{\small (Def.)};
		  \path (\tikzcdmatrixname-1-1) -- (\tikzcdmatrixname-3-2) node[midway,red,sloped,yshift=3mm,xshift=4mm]{\small (IH and funct.~of $\otimes$)};
		  \path (\tikzcdmatrixname-3-2) -- (\tikzcdmatrixname-1-3) node[midway,red,sloped,yshift=4mm]{\small (Def.)};
		}
	      ]
	      {{\Gamma_1}}\otimes{{\Gamma_2}}\otimes\Theta\ar[ddr,"\Id\otimes v",sloped]\ar[rr,"(v/x)ru"]\arrow[rd,"u\otimes(v/x)r",dashed,sloped] & &  B\\
	      &   C\otimes\home{  C}{  B}\arrow[ur,dashed,"\varepsilon",sloped] & \\
	      & {{\Gamma_1}}\otimes{{\Gamma_2}}\otimes  A\ar[u,"u\otimes r",dashed]\ar[uur,"ru",sloped]&
	    \end{tikzcd}
	  \]

	\item Let $x\in FV(u)$, so $\Gamma_1,x:A\vdash u:C$ and $\Gamma_2\vdash r:C\multimap B$. 
	  Then, the commuting diagram is the following
	  
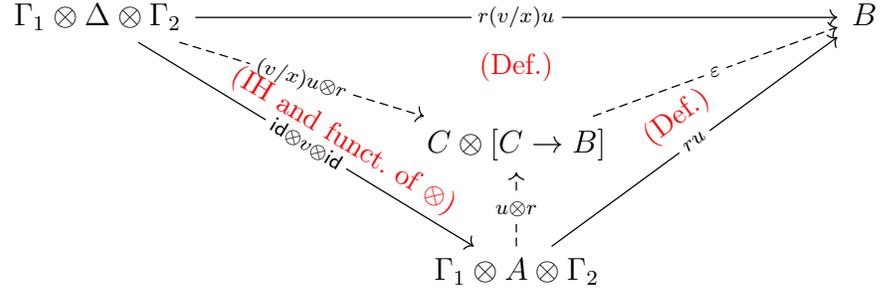
\[
	    \begin{tikzcd}[labels=description,column sep=2.9cm,row sep=1cm,
		execute at end picture={
		  \path (\tikzcdmatrixname-1-1) -- (\tikzcdmatrixname-1-3) coordinate[pos=0.5](aux)
		  (aux) -- (\tikzcdmatrixname-2-2) node[midway,red]{\small (Def.)};
		  \path (\tikzcdmatrixname-1-1) -- (\tikzcdmatrixname-3-2) node[midway,red,sloped,yshift=3mm,xshift=4mm]{\small (IH and funct.~of $\otimes$)};
		  \path (\tikzcdmatrixname-3-2) -- (\tikzcdmatrixname-1-3) node[midway,red,sloped,yshift=4mm]{\small (Def.)};
		}
	      ]
	      {{\Gamma_1}}\otimes\Theta\otimes{{\Gamma_2}}\ar[ddr,"\Id\otimes v\otimes\Id",sloped]\ar[rr,"r(v/x)u"]\arrow[rd,"(v/x)u\otimes r",dashed,sloped] & &  B\\
	      &   C\otimes\home{  C}{  B}\arrow[ur,dashed,"\varepsilon",sloped] & \\
	      & {{\Gamma_1}}\otimes  A\otimes{{\Gamma_2}}\ar[u,"u\otimes r",dashed]\ar[uur,"ru",sloped]&
	    \end{tikzcd}
	  \]

      \end{itemize}

  \item Let $t=\langle\rangle$. Then, the commuting diagram is the following.
      
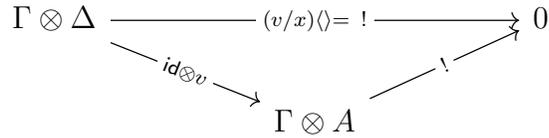
\[
	\begin{tikzcd}[labels=description,column sep=2cm]
	  \Gamma\otimes\Theta&& { 0 } \\
	  & {\Gamma\otimes A}
	  \arrow["{(v/x)\langle\rangle=\ !}"{description}, from=1-1, to=1-3]
	  \arrow["{\Id\otimes v}"{description}, from=1-1, to=2-2,sloped]
	  \arrow["{!}"{description}, from=2-2, to=1-3,sloped]
	\end{tikzcd}
      \]

  \item Let $t = \elimzero(r)$.
    \begin{itemize}
      \item Let $x\in FV(r)$, so $\Gamma=\Gamma_1,\Gamma_2$ and $\Gamma_1,x:A\vdash r:\zero$. Then, the commuting diagram is the following
	
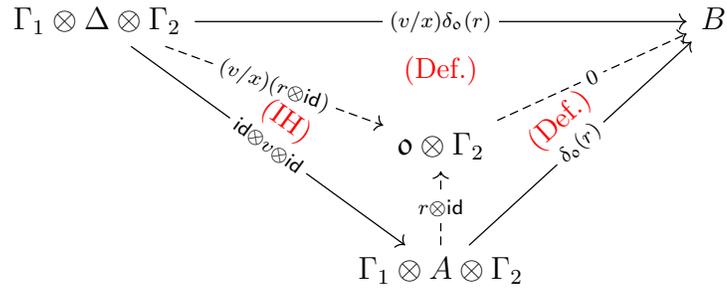
\[
	  \begin{tikzcd}[labels=description,column sep=2cm,row sep=1cm,
	      execute at end picture={
		\path (\tikzcdmatrixname-1-1) -- (\tikzcdmatrixname-1-3) coordinate[pos=0.5](aux)
		(aux) -- (\tikzcdmatrixname-2-2) node[midway,red]{\small (Def.)};
		\path (\tikzcdmatrixname-1-1) -- (\tikzcdmatrixname-3-2) node[midway,red,sloped,yshift=4mm]{\small (IH)};
		\path (\tikzcdmatrixname-3-2) -- (\tikzcdmatrixname-1-3) node[midway,red,sloped,yshift=4mm]{\small (Def.)};
	      }
	    ]
	    {{\Gamma_1}}\otimes\Theta\otimes{{\Gamma_2}}\ar[ddr,"\Id\otimes v\otimes\Id",sloped]\ar[rr,"(v/x){\elimzero(r)}"]\arrow[rd,"(v/x)(r\otimes\Id)",dashed,sloped] & &  B\\
	    & \zero\otimes{{\Gamma_2}}\arrow[ur,dashed,"0",sloped] & \\
	    &{{\Gamma_1}}\otimes  A\otimes{{\Gamma_2}}\ar[u,"r\otimes \Id",dashed]\ar[uur,"{\elimzero(r)}",sloped]&
	  \end{tikzcd}
	\]
      \item Let $x\notin FV(r)$, so $\Gamma={\Gamma_1},{\Gamma_2}$ and ${\Gamma_1}\vdash r:\zero$. Then, the commuting diagram is the following
	
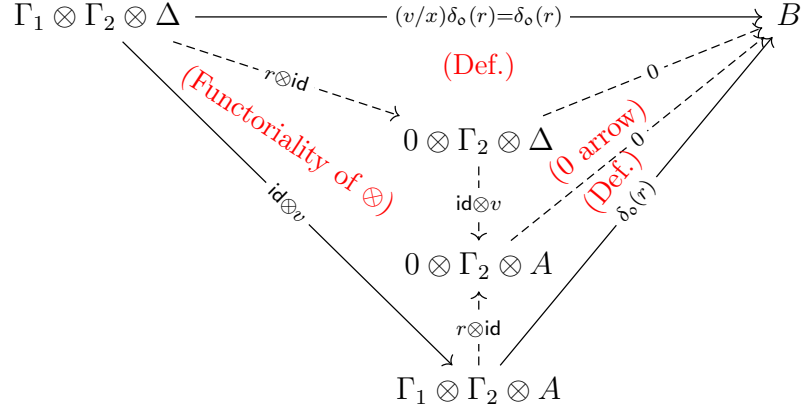
\[
	  \begin{tikzcd}[labels=description,column sep=2.5cm,row sep=1cm,
	      execute at end picture={
		\path (\tikzcdmatrixname-1-1) -- (\tikzcdmatrixname-1-3) coordinate[pos=0.5](aux)
		(aux) -- (\tikzcdmatrixname-2-2) node[midway,red]{\small (Def.)};
		\path (\tikzcdmatrixname-1-1) -- (\tikzcdmatrixname-3-2) node[midway,red,sloped]{\small (Functoriality of $\otimes$)};
		\path (\tikzcdmatrixname-3-2) -- (\tikzcdmatrixname-1-3) node[pos=0.4,red,sloped,yshift=3mm]{\small ($0$ arrow)};
		\path (\tikzcdmatrixname-4-2) -- (\tikzcdmatrixname-1-3) node[midway,red,sloped,yshift=4mm]{\small (Def.)};
	      }
	    ]
	    {{\Gamma_1}\otimes{\Gamma_2}\otimes{\Theta}} && { B} \\
	    & { 0 \otimes{\Gamma_2}\otimes{\Theta}} \\
	    & { 0 \otimes{\Gamma_2}\otimes{A}} \\
	    & {{\Gamma_1}\otimes{\Gamma_2}\otimes{A}}
	    \arrow["{\Id\otimes v}"{description}, from=1-1, to=4-2,sloped]
	    \arrow["{(v/x)\elimzero(r)=\elimzero(r)}"{description}, from=1-1, to=1-3]
	    \arrow["{\elimzero(r)}"{description}, from=4-2, to=1-3,sloped]
	    \arrow["r\otimes\Id"{description}, from=1-1, to=2-2,sloped,dashed]
	    \arrow["r\otimes\Id"{description}, from=4-2, to=3-2,dashed]
	    \arrow["{\Id\otimes v}"{description}, from=2-2, to=3-2,dashed]
	    \arrow["0"{description}, from=2-2, to=1-3,sloped,dashed]
	    \arrow["0"{description}, from=3-2, to=1-3,sloped,dashed]
	  \end{tikzcd}
	\]
    \end{itemize}

  \item 
    Let $t=\pair ru$, so $B=C\odot D$. Then, the commuting diagram is the following.
      \[
	\begin{tikzcd}[labels=description,column sep=2.8cm,row sep=1cm,
	    execute at end picture={
	      \path (\tikzcdmatrixname-1-1) -- (\tikzcdmatrixname-1-3) coordinate[pos=0.5](aux)
	      (aux) -- (\tikzcdmatrixname-2-2) node[midway,red]{\small (Def.)};
	      \path (\tikzcdmatrixname-1-1) -- (\tikzcdmatrixname-3-2) node[midway,red,sloped,yshift=1mm]{\small (Naturality of $\Delta$)};
	      \path (\tikzcdmatrixname-3-2) -- (\tikzcdmatrixname-1-3) node[midway,red,sloped,yshift=2.5mm,xshift=-3mm]{\small (IH and funct. of $\oplus$)};
	      \path (\tikzcdmatrixname-4-2) -- (\tikzcdmatrixname-1-3) node[midway,red,sloped,yshift=4mm]{\small (Def.)};
	    }
	  ]
	  \Gamma\otimes\Theta\ar[dddr,"\Id\otimes v",sloped]\ar[rr,"(v/x)\pair ru"]\arrow[rd,"\Delta",dashed,sloped] & &{  C}\oplus{ D}\\
	  & (\Gamma\otimes\Theta)\oplus(\Gamma\otimes\Theta)\arrow[ur,dashed,"(v/x)r\oplus(v/x)u",sloped]\ar[d,dashed,"(\Id\otimes v)\oplus(\Id\otimes v)"] & \\
	  & (\Gamma\otimes  A)\oplus(\Gamma\otimes  A)\ar[uur,"r\oplus u",dashed,sloped]& \\
	  & \Gamma\otimes  A\ar[u,"\Delta",dashed]\ar[uuur,"\pair ru",sloped]&
	\end{tikzcd}
      \]

  \item
    Let $t=\fst(r)$.  Then, the commuting diagram is the following
      \[
	\begin{tikzcd}[labels=description,column sep=3cm,row sep=1cm,
	    execute at end picture={
	      \path (\tikzcdmatrixname-1-1) -- (\tikzcdmatrixname-1-3) coordinate[pos=0.5](aux)
	      (aux) -- (\tikzcdmatrixname-2-2) node[midway,red]{\small (Def.)};
	      \path (\tikzcdmatrixname-1-1) -- (\tikzcdmatrixname-3-2) node[midway,red,sloped,yshift=4mm]{\small (IH)};
	      \path (\tikzcdmatrixname-3-2) -- (\tikzcdmatrixname-1-3) node[midway,red,sloped,yshift=4mm]{\small (Def.)};
	    }
	  ]
	  \Gamma\otimes\Theta\ar[ddr,"\Id\otimes v",sloped]\ar[rr,"(v/x)\fst(r)"]\arrow[rd,"(v/x)r",dashed,sloped] & &  B\\
	  &   B\oplus  C\arrow[ur,dashed,"\fst",sloped] & \\
	  & \Gamma\otimes  A\ar[u,"r",dashed]\ar[uur,"\fst(r)",sloped]&
	\end{tikzcd}
      \]
  \item
    Let $t=\snd(r)$. This case is analogous to the previous case.

  \item 
      Let $t=\inl(r)$, so $B=C+D$.  Then, the commuting diagram is the following
      \[
	\begin{tikzcd}[labels=description,column sep=3cm,row sep=1cm,
	    execute at end picture={
	      \path (\tikzcdmatrixname-1-1) -- (\tikzcdmatrixname-1-3) coordinate[pos=0.5](aux)
	      (aux) -- (\tikzcdmatrixname-2-2) node[midway,red]{\small (Def.)};
	      \path (\tikzcdmatrixname-1-1) -- (\tikzcdmatrixname-3-2) node[midway,red,sloped,yshift=4mm]{\small (IH)};
	      \path (\tikzcdmatrixname-3-2) -- (\tikzcdmatrixname-1-3) node[midway,red,sloped,yshift=4mm]{\small (Def.)};
	    }
	  ]
	  \Gamma\otimes\Theta\ar[ddr,"\Id\otimes v",sloped]\ar[rr,"(v/x)\inl(r)"]\arrow[rd,"(v/x)r",dashed,sloped] & & C\oplus D\\
	  &   C\arrow[ur,dashed,"i_1",sloped] & \\
	  & \Gamma\otimes  A\ar[u,"r",dashed]\ar[uur,"\inl(r)",sloped]&
	\end{tikzcd}
      \]

  \item Let $t=\inr(r)$. This case is analogous to the previous case.

  \item Let $t=\elimoplus(r,{y}.s_1,{z}.s_2)$. Then, $\Gamma=\Gamma_1,\Gamma_2$.
      \begin{itemize}
	\item Let $x\in FV(r)$, so $\Gamma_1,x:A\vdash r:C\vee D$,
	  $y:C,\Gamma_2\vdash s_1:B$, and $z:D,\Gamma_2\vdash s_2:B$.  Then,
	  the commuting diagram is the following
	  \[
	    \begin{tikzcd}[labels=description,column sep=1cm,row sep=8mm,
		execute at end picture={
		  \path (\tikzcdmatrixname-1-1) -- (\tikzcdmatrixname-1-3) coordinate[midway](aux)
		  (aux) -- (\tikzcdmatrixname-2-2) node[midway,red]{\small (Def.)};
		  \path (\tikzcdmatrixname-4-2) -- (\tikzcdmatrixname-1-3) node[midway,red,yshift=-5mm]{\small (Def.)};
		  \path (\tikzcdmatrixname-1-1) -- (\tikzcdmatrixname-4-2) node[midway,red,sloped,yshift=-5mm]{\small (IH and funct.~of $\otimes$)};
		}
	      ]
	      { {\Gamma_1}}\otimes \Theta\otimes{ {\Gamma_2}}\arrow[rdd,bend right=15,"(v/x)r\otimes\Id",dashed,sloped]\ar[rddd,"\Id\otimes v\otimes\Id",sloped,bend right]\ar[rr,"(v/x){\elimoplus(r,y.s_1,z.s_2)}"] & &   B\\
	      & ( C\otimes{ {\Gamma_2}})\oplus(D\otimes{ {\Gamma_2}})\arrow[ur,dashed,"\coproducto{s_1}{s_2}",sloped]&\\
	      & ( C\oplus D)\otimes{ {\Gamma_2}}\ar[u,"d",dashed] & \\
	      & { {\Gamma_1}}\otimes   A\otimes{ {\Gamma_2}}\ar[u,"r\otimes\Id",dashed]\ar[uuur,"{\elimoplus(r,y.s_1,z.s_2)}",sloped,bend right] & 
	    \end{tikzcd}
	  \]

	\item
	  Let $x\in FV(s_1)\cup FV(s_2)$, so $\Gamma_1\vdash r:C\vee D$,
	  $y:C,\Gamma_2,x:A\vdash s_1:B$, and $z:D,\Gamma_2,x:A\vdash s_2:B$.  Then,
	  the commuting diagram is the following
	  \[
	    \begin{tikzcd}[labels=description,column sep=-4mm,row sep=1.3cm,
		execute at end picture={
		  \path (\tikzcdmatrixname-1-1) -- (\tikzcdmatrixname-1-5) coordinate[pos=0.5](aux)
		  (\tikzcdmatrixname-2-2) -- (\tikzcdmatrixname-2-4) coordinate[pos=0.5](aux2)
		  (aux) -- (aux2) node[midway,red]{\small (Def.)};
		  \path (\tikzcdmatrixname-3-4) -- (\tikzcdmatrixname-4-3) node[midway,red]{\small (Def.)};
		  \path (\tikzcdmatrixname-2-2) -- (\tikzcdmatrixname-2-4) coordinate[pos=0.5](aux)
		  (\tikzcdmatrixname-3-2) -- (\tikzcdmatrixname-3-4) coordinate[pos=0.5](aux2)
		  (aux) -- (aux2) node[midway,red]{\small (Lemma~\ref{lem:distrib})};
		  \path (\tikzcdmatrixname-2-4) -- (\tikzcdmatrixname-3-4) coordinate[pos=0.5](aux)
		  (aux) -- (\tikzcdmatrixname-1-5) node[midway,red,yshift=-8mm,xshift=-4mm]{\small (IH)};
		  \path (\tikzcdmatrixname-2-2) -- (\tikzcdmatrixname-3-2) node[midway,red,sloped,yshift=-1cm]{\small (Funct.~of $\otimes$)};
		}
	      ]
	      {{\Gamma_1}}\otimes{{\Gamma_2}}\otimes\Theta\arrow[rrddd,"\Id\otimes v",
		rounded corners,
		to path={[pos=0.5]
		  -- ([yshift=-4.5cm]\tikztostart.south)\tikztonodes
		  -- ([xshift=-1cm]\tikztotarget.west)
		-- (\tikztotarget.west)}
	      ] \arrow[rrrr, "{(v/x){\elimoplus(r,y.s_1,z.s_2)}}"] \arrow[rd, dashed,sloped,"r\otimes\Id"] &[-9mm] &[-1cm] &  &[5mm]   B \\
	      & (  C\oplus  D)\otimes{{\Gamma_2}}\otimes\Theta\arrow[rr, dashed,sloped,"d"] \arrow[d, dashed,"\Id\otimes v"] & & (  C\otimes{{\Gamma_2}}\otimes\Theta)\oplus ( D\otimes{{\Gamma_2}}\otimes\Theta) \arrow[d, dashed,"\Id\otimes v\oplus\Id\otimes v"] \arrow[ru, dashed,sloped,"{\coproducto{(v/x)s_1}{(v/x)s_2}}"] &   \\[1cm]
	      & (  C\oplus  D)\otimes{{\Gamma_2}}\otimes   A \arrow[rr, dashed,"d"] & & (  C\otimes{{\Gamma_2}}\otimes   A)\oplus ( D\otimes{{\Gamma_2}}\otimes   A) \arrow[ruu, dashed,bend right,looseness=0.8,sloped,"{\coproducto{s_1}{s_2}}"] &   \\
	      & & {{\Gamma_1}}\otimes{{\Gamma_2}}\otimes\Theta\arrow[lu, dashed,"r\otimes\Id",sloped] \arrow[rruuu, "{{\elimoplus(r,y.s_1,z.s_2)}}" description,sloped,
		rounded corners,
		to path={[pos=0.5]
		  -- ([xshift=4.2cm]\tikztostart.east)
		  -- ([yshift=-4.5cm]\tikztotarget.south)
		-- (\tikztotarget.south)\tikztonodes}
	      ]
	      & &
	    \end{tikzcd}
	  \]
      \end{itemize}
    
    \item Let $t=\super ru$. Analogous to case $t=\pair ru$.

    \item Let $t=\fstsup(r)$. Analogous to case $t=\fst(r)$.

    \item Let $t=\sndsup(r)$. Analogous to case $t=\snd(r)$.

  \item
      Let $t=\elimsup(r,{y}.s_1,{z}.s_2)$. Then, $\Gamma=\Gamma_1,\Gamma_2$.
      \begin{itemize}
	\item Let $x\in FV(r)$, so 
	  $\Gamma_1,x:A\vdash r:C\odot D$,
	  $y:C,\Gamma_2\vdash s_1:B$, and $z:D,\Gamma_2\vdash s_2:B$. 
	  Then, the commuting diagram is the following
	  
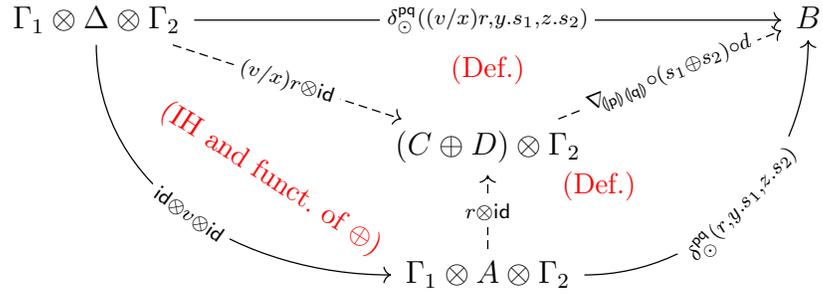
\[
	    \begin{tikzcd}[labels=description,column sep=2.5cm,row sep=1cm,
		execute at end picture={
		  \path (\tikzcdmatrixname-1-1) -- (\tikzcdmatrixname-1-3) coordinate[pos=0.5](aux)
		  (aux) -- (\tikzcdmatrixname-2-2) node[midway,red]{\small (Def.)};
		  \path (\tikzcdmatrixname-3-2) -- (\tikzcdmatrixname-1-3) node[pos=0.3,red]{\small (Def.)};
		  \path (\tikzcdmatrixname-1-1) -- (\tikzcdmatrixname-3-2) node[midway,red,sloped,yshift=-5mm]{\small (IH and funct.~of $\oplus$)};
		}
	      ]
	      {{\Gamma_1}}\otimes{{\Theta}}\otimes{\Gamma_2}\ar[ddr,"\Id\otimes v\otimes\Id",sloped,out=270,in=180]\ar[rd,"(v/x)r\otimes\Id",dashed,sloped]\ar[rr,"{\elimsup((v/x)r,{y}.s_1,{z}.s_2)}"] & &  B\\
	      & ({C}\oplus{D})\otimes {{\Gamma_2}}\arrow[ur,dashed,"\sumwe{\semS{\escalar p}}{\semS{\escalar q}}\circ(s_1\oplus s_2)\circ d",sloped] & \\
	      & {{\Gamma_1}}\otimes{ A}\otimes {\Gamma_2}\ar[uur,"{\elimsup(r,{y}.s_1,{z}.s_2)}",sloped,out=0,in=270]\ar[u,dashed,"r\otimes\Id"]&
	    \end{tikzcd}
	  \]
	\item Let $x\in FV(s_1)\cup FV(s_2)$, so 
	  $\Gamma_1\vdash r:C\odot D$,
	  $y:C,\Gamma_2,x:A\vdash s_1:B$, and $z:D,\Gamma_2,x:A\vdash s_2:B$.
	  Then, the commuting diagram is the following
	  \[
	    \hspace{-1cm}\begin{tikzcd}[labels=description,row sep=1cm,column sep=1.5cm,
		execute at end picture={
		  \path (\tikzcdmatrixname-1-1) -- (\tikzcdmatrixname-3-2) node[midway,red,sloped,yshift=-3mm]{\small (Lemma~\ref{lem:distrib})};
		  \path (\tikzcdmatrixname-1-2) -- (\tikzcdmatrixname-2-2) coordinate[pos=0.5](aux)
		  (aux) -- (\tikzcdmatrixname-2-3) node[pos=0.4,red,yshift=-2mm]{\small (IH and funct.~of $\oplus$)};
		  \path (\tikzcdmatrixname-3-1) -- (\tikzcdmatrixname-3-2) node[pos=0.33,red]{\small (Funct.~of $\otimes$)};
		  \path (\tikzcdmatrixname-1-1) -- (\tikzcdmatrixname-1-3) node[pos=0.42,red,yshift=5mm]{\small (Def.)};
		  \path (\tikzcdmatrixname-3-2) -- (\tikzcdmatrixname-2-2) node[pos=-1,red]{\small (Def.)};
		}
	      ]
	      {({C}\oplus{D})\otimes{\Gamma_2}\otimes{\Theta}} &[-5mm] {({C}\otimes{\Gamma_2}\otimes{\Theta})\oplus({D}\otimes{\Gamma_2}\otimes{\Theta})} & {\ } \\
	      & {({C}\otimes{\Gamma_2}\otimes{A})\oplus({D}\otimes{\Gamma_2}\otimes{A})}& { B\oplus B} \\
	      {{\Gamma_1}\otimes{\Gamma_2}\otimes{\Theta}}
	      \arrow[rrd, "{\elimsup(r,y.(v/x)s_1,z.(v/x)s_2)}" description,sloped,
		rounded corners,
		to path={[pos=0.75]
		  -- ([xshift=-5mm]\tikztostart.west)
		  |- ([xshift=4mm,yshift=6.2cm]\tikztotarget.east)\tikztonodes
		|- (\tikztotarget.east)}
	      ]
	      & {({C}\oplus{D})\otimes{\Gamma_2}\otimes{A}} \\
	      {{\Gamma_1}\otimes{\Gamma_2}\otimes{A}} && { B}
	      \arrow["{\Id\otimes v}", from=3-1, to=4-1]
	      \arrow["{\elimsup(r,y.s_1,z.s_2)}", from=4-1, to=4-3]
	      \arrow["r\otimes\Id"', from=3-1, to=1-1,dashed]
	      \arrow["d"', from=1-1, to=1-2,dashed]
	      \arrow["{(v/x)s_1\oplus(v/x)s_2}"', from=1-2, to=2-3,dashed,sloped]
	      \arrow["\sumwe{\semS{\escalar p}}{\semS{\escalar q}}"', from=2-3, to=4-3,dashed]
	      \arrow["r\otimes\Id"', from=4-1, to=3-2,sloped,dashed]
	      \arrow["d"', from=3-2, to=2-2,dashed]
	      \arrow["{s_1\oplus s_2}"', from=2-2, to=2-3,dashed]
	      \arrow["{(\Id\otimes v)\oplus(\Id\otimes v)}"{description}, from=1-2, to=2-2,dashed]
	      \arrow["{\Id\otimes v}"{description}, from=1-1, to=3-2,sloped,dashed,bend right]
	    \end{tikzcd}
	  \]
	  \qedhere
      \end{itemize}
\end{itemize}
  \end{proof}

\section{Proof of Theorem~\ref{thm:soundness}}\label{app:soundness}

\noindent {\bf Theorem \ref{thm:soundness}} (Soundness){\bf .} 
\emph{Let $\Gamma\vdash t:A$.}
\begin{itemize}
  \item\emph{
      If $t\lra r$, by any rule but $({\relimsup \ell})$ and $({\relimsup r})$,  then
      \[
	\sem{\Gamma\vdash t:A}=\sem{\Gamma\vdash r:A}
      \]
    }
  \item\emph{
    If $t\lra[\escalar p] r_1$ by rule $({\relimsup \ell})$ and $t\lra[\escalar q] r_2$ by rule $({\relimsup r})$, 
       then.
      \[
	\sem{\Gamma\vdash t:A} = \sumwe{\semS{\escalar p}}{\semS{\escalar q}}\circ(\sem{\Gamma\vdash r_1:A}\oplus\sem{\Gamma\vdash r_2:A})\circ\Delta
      \]
    }

    \emph{
      That is,
    }
    \[
      \begin{tikzcd}[labels=description,row sep=1cm,column sep=1.5cm]
	  \sem\Gamma\ar[r,"t"]\ar[d,"\Delta"] & \sem A\\
	  \sem\Gamma\oplus\sem\Gamma\ar[r,"r_1\oplus r_2"] & \sem A\oplus\sem A\ar[u,"{\sumwe{\semS{\escalar p}}{\semS{\escalar q}}}"]
      \end{tikzcd}
    \]
\end{itemize}
\begin{proof}
  By induction on the relation $\lra[\escalar p]$, using implicitly the coherence maps when needed.
  To avoid cumbersome notation, we write $A$ instead of $\sem A$.

  {\bf Basic cases:} 

  \begin{itemize}
    \item 
      \(
	\vcenter{
	  \infer{\Gamma\vdash\elimone(\escalar s.\star,t):A}
	  {
	    \infer{\vdash\escalar s.\star:\one}{}
	    &
	    \Gamma\vdash t:A
	  }
	}
	\qquad\lra\qquad
	\vcenter{
	  \infer{\Gamma\vdash\escalar s\bullet t:A}
	  {\Gamma\vdash t:A}
	}
      \)
      \[
	\begin{tikzcd}[labels=description,row sep=1cm,column sep=1.9cm,
	    execute at end picture={
	      \path (\tikzcdmatrixname-1-1) -- (\tikzcdmatrixname-1-5) node[midway,red,yshift=6mm]{\small (Def.)};
	      \path (\tikzcdmatrixname-1-1) -- (\tikzcdmatrixname-1-5) node[midway,red,yshift=-3.8cm]{\small (Def.)};
	      \path (\tikzcdmatrixname-1-2) -- (\tikzcdmatrixname-1-4) coordinate[pos=0.5](aux)
	      (aux) -- (\tikzcdmatrixname-2-3) node[pos=0.4,red]{\small (Funct.~of $\otimes$)};
	      \path
	      (\tikzcdmatrixname-1-1) -- (\tikzcdmatrixname-1-2) coordinate[pos=0.5](aux1)
	      (\tikzcdmatrixname-2-3) -- (\tikzcdmatrixname-3-3) coordinate[pos=0.5](aux2)
	      (aux1) -- (aux2) node[midway,red,sloped]{\small (Naturality of $\lambda^{-1}$)}
	      (\tikzcdmatrixname-1-4) -- (\tikzcdmatrixname-1-5) coordinate[pos=0.5](aux1)
	      (aux2) -- (aux1) node[midway,red,sloped]{\small (Naturality of $\lambda$)};
	    }
	  ]
	  \Gamma \arrow[r, "\lambda^{-1}",dashed] \arrow[rrrr, bend left=15,sloped,"{\elimone(\escalar s.\star,t)}"] 
	  \arrow[rrrr, "{\escalar s\bullet t}" description,sloped,
	    rounded corners,
	    to path={[pos=0.75]
	      -- ([yshift=-3mm]\tikztostart.south)
	      |- ([yshift=-4cm]\tikztotarget.south)\tikztonodes
	    -- (\tikztotarget.south)}
	  ]
	  \arrow[rrdd, dashed,"t",sloped] &  I\otimes\Gamma \arrow[rr, dashed,"\hat{\escalar s}\otimes t"] \arrow[dr, dashed,"\Id\otimes t",sloped] & &  I\otimes A \arrow[r, dashed,"\lambda"]  & A \\
	  &{\ }     &  I\otimes A \arrow[ru, dashed,"\hat{\escalar s}\otimes\Id",sloped]  & {\ }  &   \\
	  & & A \arrow[rruu, dashed,"\hat{\escalar s}",sloped] \arrow[u, "\lambda^{-1}",dashed] 
	\end{tikzcd}
      \]

    \item
      \(
      \vcenter{
	\infer{\Gamma,\Theta_1,\Theta_2\vdash\elimtens(t\otimes u,xy.r):C}
	{
	  \Gamma,x:A,y:B\vdash r:C
	  &
	  \infer{\Theta_1,\Theta_2\vdash t\otimes u:A\otimes B}
	  {
	    \Theta_1\vdash t:A & \Theta_2\vdash u:B
	  }
	}
      }
      \)

      \hfill\(
      \lra\qquad
      \Gamma,\Theta_1,\Theta_2\vdash (t/x,u/y)r:C
      \)
      \[
	\begin{tikzcd}[column sep=1cm,row sep=1cm,labels=description,
	    execute at end picture={
	      \path (\tikzcdmatrixname-1-1) -- (\tikzcdmatrixname-1-3) coordinate[midway](aux)
	      (aux) -- (\tikzcdmatrixname-2-2) node[midway,red]{\small (Lemma~\ref{lem:subst})};
	      \path (\tikzcdmatrixname-3-2) -- (\tikzcdmatrixname-1-3) node[midway,red,sloped]{\small (Lemma~\ref{lem:subst})};
	      \path  (\tikzcdmatrixname-1-1) -- (\tikzcdmatrixname-3-2) node[midway,red,sloped]{\small (Funct.~of $\oplus$)};
	      \path (\tikzcdmatrixname-1-1) -- (\tikzcdmatrixname-3-2) node[midway,red,xshift=-2cm]{\small (Def.)};
	    }
	    ]
	    {\Gamma\otimes{\Theta_1}\otimes{\Theta_2}}
	    \arrow[rr, "{\elimtens(t\otimes u,xy.r)}" description,sloped,
	    rounded corners,
	    to path={[pos=0.75]
	    -- ([xshift=-1mm]\tikztostart.west)
	    |- ([xshift=3mm,yshift=-4cm]\tikztotarget.east)\tikztonodes
	    |- (\tikztotarget.east)}
	    ]
	    && { C} \\
	    & {\Gamma\otimes{A}\otimes{\Theta_2}} \\
	    & {\Gamma\otimes{A}\otimes{B}}
	    \arrow["{(t/x,u/y)r}", from=1-1, to=1-3]
	    \arrow["{\Id\otimes t\otimes u}"', from=1-1, to=3-2,sloped,dashed,bend right]
	    \arrow["r"', from=3-2, to=1-3,sloped,dashed,bend right]
	    \arrow["{\Id\otimes t\otimes\Id}"', from=1-1, to=2-2,sloped,dashed]
	    \arrow["{\Id\otimes u}"', from=2-2, to=3-2,dashed]
	    \arrow["{(u/y)r}", from=2-2, to=1-3,sloped,dashed]
	\end{tikzcd}
      \]

    \item 
      \(
      \vcenter{
	\infer{\Gamma,\Theta\vdash(\lambda{x}.t)u:B}
	{
	  \infer{\Gamma\vdash\lambda{x}.t:A\multimap B}{\Gamma,x:A\vdash u:B}
	  \qquad
	  \Theta\vdash u:A
	}
      }
      \qquad\lra\qquad
      \vcenter{
	$\Gamma,\Theta\vdash(u/x)t:B$
      }
      \)
      \[
	\begin{tikzcd}[column sep=1.3cm,row sep=2cm,labels=description,
	    execute at end picture={
	      \path (\tikzcdmatrixname-1-1) -- (\tikzcdmatrixname-1-5) node[midway,red,yshift=6mm]{\small (Def.)};
	      \path (\tikzcdmatrixname-1-1) -- (\tikzcdmatrixname-1-5) node[midway,red,yshift=-3cm]{\small (Lemma~\ref{lem:subst})};
	      \path (\tikzcdmatrixname-1-1) -- (\tikzcdmatrixname-2-3) coordinate[pos=0.5](aux)
	      (aux) -- (\tikzcdmatrixname-1-2) node[midway,red]{\small (Funct.~of $\otimes$)};
	      \path (\tikzcdmatrixname-2-3) -- (\tikzcdmatrixname-1-5) coordinate[pos=0.5](aux)
	      (aux) -- (\tikzcdmatrixname-1-2) node[midway,red,xshift=5mm]{\small (Naturaliry of $\varepsilon$)};
	      \path (\tikzcdmatrixname-1-2) -- (\tikzcdmatrixname-2-3) node[midway,red,sloped,yshift=3mm]{\small (Adj.~ax.)};
	    }
	    ]
	    \Gamma\otimes\Theta \arrow[rrrr, "(\lambda x.t)u", bend left=15,sloped] \arrow[r, "\eta^{ B}\otimes u", dashed] \arrow[rrd, "\Id\otimes u", dashed,sloped]
	    \arrow[rrrr, "{(u/x)t}" description,sloped,
	    rounded corners,
	    to path={[pos=0.75]
	    -- ([yshift=-3mm]\tikztostart.south)
	    |- ([yshift=-3.5cm]\tikztotarget.south)\tikztonodes
	    -- (\tikztotarget.south)}
	    ]
	    & \home{ B}{\Gamma\otimes B}\otimes B \arrow[rr, "\home{ B}{t}\otimes\Id", dashed] \arrow[rd, "\varepsilon", dashed, bend left=40,sloped] &[-2.1cm]  & \home{ B}{ A}\otimes B\arrow[r, "\varepsilon", dashed] &  A \\
	    & & \Gamma\otimes B\arrow[rru, "t", dashed,sloped] \arrow[lu, "\eta^{ B}\otimes\Id", dashed,sloped] & &
	\end{tikzcd}
      \]

    \item
      \(
      \vcenter{
	\infer{\Gamma\vdash\fst\pair tu:A}
	{
	  \infer{\Gamma\vdash\pair tu:A\odot B}
	  {
	    \Gamma\vdash t:A 
	    &
	    \Gamma\vdash u:B
	  }
	}
      }
      \qquad\lra\qquad
      \vcenter{
	$\Gamma\vdash t:A$
      }
      \)
      \[
	\begin{tikzcd}[column sep=2.3cm,row sep=3em,labels=description,
	    execute at end picture={
	      \path (\tikzcdmatrixname-1-1) -- (\tikzcdmatrixname-2-1) node[midway,red]{\small (*)};
	      \path (\tikzcdmatrixname-1-1) -- (\tikzcdmatrixname-1-2) coordinate[pos=0.5](aux)
	      (\tikzcdmatrixname-2-1) -- (\tikzcdmatrixname-2-2) coordinate[pos=0.5](aux2)
	      (aux) -- (aux2) node[midway,red]{\small (Coherence)};
	      \path (\tikzcdmatrixname-1-1) -- (\tikzcdmatrixname-1-2) node[midway,red,yshift=-2.2cm]{\small (Def.)};
	    }
	    ]
	    \Gamma \arrow[r, "t"] \arrow[d, "\Delta", dashed, bend right]
	    \arrow[r, "{\fst\pair tu}" description,sloped,
	    rounded corners,
	    to path={[pos=0.75]
	    -- ([xshift=-1cm]\tikztostart.west)
	    |- ([xshift=1cm,yshift=-2.5cm]\tikztotarget.east)\tikztonodes
	    |- (\tikztotarget.east)}
	    ]
	    &   A   \\
	    \Gamma\oplus\Gamma \arrow[r, "t\oplus u", dashed] \arrow[u, "\pi_1", dashed, bend right] &   A\oplus   B  \arrow[u, "\pi_1", dashed] & 
	\end{tikzcd}
      \]
      {\color{red}(*) $\pi_1\circ\Delta = \Id_\Gamma$}
      \medskip

    \item 
      \(
      \vcenter{
	\infer{\Gamma\vdash\snd\pair tu:B}
	{
	  \infer{\Gamma\vdash\pair tu:A\odot B}
	  {
	    \Gamma\vdash t:A 
	    &
	    \Gamma\vdash u:B
	  }
	}
      }
      \qquad\lra\qquad
      \vcenter{
	$\Gamma\vdash u:B$
      }
      \)\\

      Analogous to the previous case.

    \item
      \(
      \vcenter{
	\infer{\Gamma,\Theta\vdash\elimoplus(\inl(t),{x}.u,{y}.v):C}
	{
	  \infer{\Gamma\vdash\inl(t):A\oplus B}
	  {
	    \Gamma\vdash t:A 
	  }
	  &
	  \Theta,x:A\vdash u:C
	  &
	  \Theta,y:B\vdash v:C
	}
      }
      \)

      \hfill $\lra\qquad \Gamma,\Theta\vdash (t/x)u:C$

      \[
	\begin{tikzcd}[column sep=1cm,row sep=3em,labels=description,
	    execute at end picture={
	      \path (\tikzcdmatrixname-1-1) -- (\tikzcdmatrixname-1-4) node[midway,red,yshift=-3.5cm]{\small (Def.)};
	      \path (\tikzcdmatrixname-1-1) -- (\tikzcdmatrixname-2-1) coordinate[pos=0.5](aux)
	      (aux) -- (\tikzcdmatrixname-1-4) node[midway,red,xshift=-1cm]{\small (Lemma~\ref{lem:subst})};
	      \path (\tikzcdmatrixname-1-4) -- (\tikzcdmatrixname-2-4) coordinate[pos=0.5](aux)
	      (\tikzcdmatrixname-2-1) -- (aux) node[midway,red,xshift=1cm]{\small (Def.~of coproduct)};
	      \path (\tikzcdmatrixname-2-1) -- (\tikzcdmatrixname-3-1) coordinate[pos=0.5](aux)
	      (aux) -- (\tikzcdmatrixname-2-4) node[midway,red]{\small (*)};
	    }
	    ]
	    \Gamma\otimes\Theta \arrow[rrr, "(t/x)u"] \arrow[d, "t\otimes\Id", dashed] 
	    \arrow[rrr, "{\elimoplus(\inl(t),x.u,y.v)}" description,sloped,
	    rounded corners,
	    to path={[pos=0.75]
	    -- ([xshift=-0.8cm]\tikztostart.west)
	    |- ([xshift=1.8cm,yshift=-4.5cm]\tikztotarget.east)\tikztonodes
	    |- (\tikztotarget.east)}
	    ]
	    & &[1.5cm] & C \\
	    A\otimes\Theta \arrow[d, "i_1\otimes\Id", dashed] \arrow[sloped,rrru, "u", dashed] \arrow[rrr, "i_1", dashed]  & & &( A\otimes\Theta)\oplus( B\otimes\Theta) \arrow[u, "\coproducto uv", dashed] \\
	    ( A\oplus B)\otimes\Theta \arrow[rrru, "d", dashed, sloped] & & & 
	\end{tikzcd}
      \]
      The commutation of the diagram {\color{red}(*)} is justified as follows.
      \begin{align*}
	d\circ (i_1\otimes\Id)
	&=\pair{\pi_1\otimes\Id}{\pi_2\otimes\Id}\circ (i_1\otimes\Id)\\
	&=\pair{(\pi_1\otimes\Id)\circ (i_1\otimes\Id)}{(\pi_2\otimes\Id)\circ (i_1\otimes\Id)}\\
	&=\pair{(\pi_1\circ i_1)\otimes\Id)}{(\pi_2\circ i_1)\otimes\Id}\\
	&=\pair{\Id\otimes\Id}{0\otimes\Id}\\
	&=\pair{\Id}{0}\\
	&=i_1
      \end{align*}

    \item
      \(
      \vcenter{
	\infer{\Gamma,\Theta\vdash\elimoplus(\inr(t),{x}.u,{y}.v):C}
	{
	  \infer{\Gamma\vdash\inr(t):A\oplus B}
	  {
	    \Gamma\vdash t:B 
	  }
	  &
	  \Theta,x:A\vdash u:C
	  &
	  \Theta,y:B\vdash v:C
	}
      }
      \)

      \hfill $\lra\qquad\Gamma,\Theta\vdash (t/y)v:C$

      Analogous to the previous case.

    \item
      \(
      \vcenter{
	\infer{\Gamma\vdash\pi^\odot_1\super tu:A}
	{
	  \infer{\Gamma\vdash\super tu:{A\odot B}}
	  {
	    \Gamma\vdash t:A
	    &
	    \Gamma\vdash u:B
	  }
	}
      }
      \qquad\lra\qquad
      \vcenter{
	$\Gamma\vdash t:A$
      }
      \)
      Since the interpretation of the derivations are the same to the case of $\pi_1\pair tu$, this case is analogous to that one.

    \item 
      \(
      \vcenter{
	\infer{\Gamma\vdash\pi^\odot_2\super tu:B}
	{
	  \infer{\Gamma\vdash\super tu:{A\odot B}}
	  {
	    \Gamma\vdash t:A
	    &
	    \Gamma\vdash u:B
	  }
	}
      }
      \qquad\lra\qquad
      \vcenter{
	$\Gamma\vdash u:B$
      }
      \)

      Analogous to the previous case.

    \item
      \begin{tikzcd}[column sep=-3cm,row sep=1.5cm]
	& \vcenter{
	  \infer{\Gamma,\Theta\vdash\elimsup(\super{t_1}{t_2},{x}.u,{y}.v):C}
	  {
	    \infer{\Gamma\vdash\super{t_1}{t_2}:A\odot B}
	    {
	      \Gamma\vdash t_1:A 
	      &
	      \Gamma\vdash t_2:B
	    }
	    &
	    x:A,\Theta\vdash u:C
	    &
	    y:B,\Theta\vdash v:C
	  }
	}
	\ar[dl, start anchor={south}, end anchor={north}, pos=.9,"\semS{\escalar p}"']\ar[dr, start anchor={south}, end anchor={north}, pos=.9,"\semS{\escalar q}"] & \\
	\Gamma,\Theta\vdash ({t_1}/x)r:C & &  \Gamma,\Theta\vdash ({t_2}/y)s:C
      \end{tikzcd}

      \[
	\begin{tikzcd}[column sep=1.1cm,row sep=3em,labels=description,
	    execute at end picture={
	      \path (\tikzcdmatrixname-1-1) -- (\tikzcdmatrixname-2-1) coordinate[pos=0.5](aux)
	      (aux) -- (\tikzcdmatrixname-1-3) node[midway,red]{\small (Corollary~\ref{cor:PropsFinv})};
	      \path (\tikzcdmatrixname-3-1) -- (\tikzcdmatrixname-1-3) node[midway,red]{\small (Lemma~\ref{lem:distrib})};
	      \path (\tikzcdmatrixname-1-3) -- (\tikzcdmatrixname-2-3) coordinate[pos=0.5](aux)
	      (\tikzcdmatrixname-3-2) -- (aux) node[midway,red,sloped]{\small (Lemma~\ref{lem:subst})};
	      \path (\tikzcdmatrixname-3-2) -- (\tikzcdmatrixname-3-3) node[midway,red]{\small (Def.)};
	    }
	    ]
	    \Gamma\otimes\Theta
	    \arrow[rrdd, "{\elimsup(\super{t_1}{t_2},x.u,y.v)}" description,sloped,
	    rounded corners,
	    to path={[pos=0.75]
	    -- ([xshift=-1cm]\tikztostart.west)
	    |- ([yshift=-5mm]\tikztotarget.south)\tikztonodes
	    -- (\tikztotarget)}
	    ]
	    && {(\Gamma\otimes\Theta)\oplus(\Gamma\otimes\Theta)} \\
	    {(\Gamma\oplus\Gamma)\otimes\Theta} && { C\oplus C} \\
	    {( A\oplus B)\otimes\Theta} & {( A\otimes\Theta)\oplus( B\otimes\Theta)} & { C}
	    \arrow["\Delta", from=1-1, to=1-3]
	    \arrow["\sumwe{\semS{\escalar p}}{\semS{\escalar q}}", from=2-3, to=3-3]
	    \arrow["{(t_1/x)r\oplus (t_2/y)s}", from=1-3, to=2-3]
	    \arrow["\Delta\otimes\Id"', from=1-1, to=2-1,dashed]
	    \arrow["{(t_1\oplus t_2)\otimes\Id}"', from=2-1, to=3-1,dashed]
	    \arrow["d"', from=3-1, to=3-2,dashed]
	    \arrow["{r\oplus s}"', from=3-2, to=2-3,sloped,dashed]
	    \arrow["d", from=2-1, to=1-3,sloped,dashed]
	    \arrow["{(t_1\otimes\Id)\oplus(t_2\otimes \Id)}", from=1-3, to=3-2,sloped,dashed]
	\end{tikzcd}
      \]
    \item 
      \(
	\vcenter{
	  \infer{\vdash \escalar s_1.\star\plus\ \escalar s_2.\star:\one}
	  {
	    \vdash \escalar s_1.\star:\one 
	    &
	    \vdash \escalar s_2.\star:\one 
	  }
	}
	\qquad\lra\qquad
	\vcenter{
	  $\vdash (\escalar s_1\masS[\mathcal S] \escalar s_2).\star:\one$
	}
      \)

      Since $\llparenthesis\cdot\rrparenthesis$ is a homomorphism, we have
      \[
	\sem{\vdash(\escalar s_1\masS[\mathcal S] \escalar s_2).\star:\one } = {\semS{\escalar s_1\masS[\mathcal S]\escalar s_2}}=\semS{\escalar s_1}+\semS{\escalar s_2}
	=\sem{\vdash \escalar s_1.\star\plus\ \escalar s_2.\star:\one}
      \]

    \item 
      \(
	\vcenter{
	  \infer
	  {\Gamma,\Theta\vdash\elimtens(t\plus u,xy.v):A}
	  {
	    \infer
	    {\Theta\vdash t\plus u:B\otimes C}
	    {\Theta\vdash t:B\otimes C\quad\Theta\vdash u:B\otimes C}
	    &
	    \Gamma,x:B,y:C\vdash v:A
	  }
	}
	\qquad\lra
      \)

      \hfill\(
	\vcenter{
	  \infer
	  {\Gamma,\Theta\vdash\elimtens(t,xy.v)\plus\elimtens(u,xy.v):A}
	  {
	    \infer
	    {\Gamma,\Theta\vdash\elimtens(t,xy.v):A}
	    {\Theta\vdash t:B\otimes C & \Gamma,x:B,y:C\vdash v:A}
	    &
	    \infer
	    {\Gamma,\Theta\vdash\elimtens(u,xy.v):A}
	    {\Theta\vdash u:B\otimes C & \Gamma,x:B,y:C\vdash v:A}
	  }
	}
      \)

      \[
	\begin{tikzcd}[labels=description,column sep=5mm,
	    execute at end picture={
	      \path (\tikzcdmatrixname-1-1) -- (\tikzcdmatrixname-2-1) coordinate[midway](aux)
	      (aux) -- (\tikzcdmatrixname-1-4) node[pos=0.4,sloped,red]{\small (Corollary~\ref{cor:PropsFinv})};
	      \path 
	      (\tikzcdmatrixname-2-1) -- (\tikzcdmatrixname-3-1) coordinate[midway](aux1)
	      (\tikzcdmatrixname-1-4) -- (\tikzcdmatrixname-3-4) coordinate[midway](aux2)
	      (aux1) -- (aux2) node[midway,red]{\small (Lemma~\ref{lem:distrib})};
	      \path 
	      (\tikzcdmatrixname-3-1) -- (\tikzcdmatrixname-7-1) coordinate[pos=0.875](aux1)
	      (\tikzcdmatrixname-6-4) -- (\tikzcdmatrixname-7-4) coordinate[midway](aux2)
	      (aux1) -- (aux2) node[midway,red]{\small (Naturality of $\nabla$)};
	      \path (\tikzcdmatrixname-1-1) -- (\tikzcdmatrixname-1-4) node[midway,red,yshift=5mm]{\small (Def.)}; 
	      \path (\tikzcdmatrixname-7-1) -- (\tikzcdmatrixname-7-4) node[midway,red,yshift=-5mm]{\small (Def.)}; 
	      \path (\tikzcdmatrixname-3-1) -- (\tikzcdmatrixname-6-4) node[midway,red]{\small (Corollary~\ref{cor:PropsF})}; 
	    }
	  ]
	  \Gamma\otimes\Theta
	  \arrow[rrrdddddd, "{\elimtens(t\plus u,xy.v)}" description,sloped,
	    rounded corners,
	    to path={[pos=0.25]
	      -- ([yshift=5mm]\tikztostart.north)
	      -| ([xshift=2.2cm]\tikztotarget.east)\tikztonodes
	    -- (\tikztotarget)}
	  ]
	  \arrow[rrrdddddd, "{\elimtens(t,xy.v)\plus\elimtens(u,xy.v)}" description,sloped,
	    rounded corners,
	    to path={[pos=0.75]
	      -- ([xshift=-2cm]\tikztostart.west)
	      |- ([yshift=-7mm]\tikztotarget.south)\tikztonodes
	    -- (\tikztotarget)}
	  ]
	  &&& {\Gamma\otimes(\Theta\oplus\Theta)} \\
	  {(\Gamma\otimes\Theta)\oplus (\Gamma\otimes\Theta)} \\
	  {(\Gamma\otimes B\otimes C)\oplus (\Gamma\otimes B\otimes C)} &&& {\Gamma\otimes((B\otimes C)\oplus (B\otimes C))} \\
	  &  &  \\
	  &  &  \\
	  &&& {\Gamma\otimes B\otimes C} \\
	  {A\oplus A} &&& A
	  \arrow["v", dashed, from=6-4, to=7-4]
	  \arrow["\Id\otimes\Delta", dashed, from=1-1, to=1-4]
	  \arrow["{\Id\otimes (t\oplus u)}", dashed, from=1-4, to=3-4]
	  \arrow["\Id\otimes\nabla", dashed, from=3-4, to=6-4]
	  \arrow["\Delta", dashed, from=1-1, to=2-1]
	  \arrow["{(\Id\otimes t)\oplus (\Id\otimes u)}", dashed, from=2-1, to=3-1]
	  \arrow["{v\oplus v}", dashed, from=3-1, to=7-1]
	  \arrow["\nabla", dashed, from=7-1, to=7-4]
	  \arrow["\nabla", out=270,in=180, looseness=1.2, sloped, dashed, from=3-1, to=6-4]
	  \arrow["d", dashed, from=1-4, to=2-1,sloped]
	  \arrow["d", dashed, from=3-4, to=3-1]
	\end{tikzcd}
      \]

    \item 
      \(
	\vcenter{
	  \infer{\Gamma\vdash \lambda{x}.t\plus\lambda{x}.u:A\multimap B}
	  {
	    \infer{\Gamma\vdash\lambda{x}.t:A\multimap B}{
	      \Gamma,x:A\vdash t:B
	    }
	    &
	    \infer{\Gamma\vdash\lambda{x}.u:A\multimap B}{
	      \Gamma,x:A\vdash u:B
	    }
	  }
	}
      \)

      \hfill\(
	\lra\qquad
	\vcenter{
	  \infer{\Gamma\vdash \lambda{x}.(t\plus u):A\multimap B}
	  {
	    \infer{\Gamma,x:A\vdash t\plus u:B}{
	      \Gamma,x:A\vdash t:B
	      &
	      \Gamma,x:A\vdash u:B
	    }
	  }
	}
      \)

      \[
	\hspace{-1cm}
	\begin{tikzcd}[column sep=5mm,row sep=3em,labels=description,
	    execute at end picture={
	      \path (\tikzcdmatrixname-1-1) -- (\tikzcdmatrixname-2-1) coordinate[pos=0.5](aux)
	      (aux) -- (\tikzcdmatrixname-1-3) node[midway,red,xshift=-5mm]{\small (Naturality of $\Delta$)};
	      \path (\tikzcdmatrixname-1-1) -- (\tikzcdmatrixname-1-3) node[midway,red,yshift=4mm]{\small (Def.)};
	      \path (\tikzcdmatrixname-4-1) -- (\tikzcdmatrixname-4-3) node[midway,red,yshift=-4mm]{\small (Def.)};
	      \path (\tikzcdmatrixname-3-1) -- (\tikzcdmatrixname-4-1) coordinate[pos=0.5](aux1)
	      (\tikzcdmatrixname-1-3) -- (\tikzcdmatrixname-2-3) coordinate[pos=0.5](aux2)
	      (aux1) -- (aux2) node[midway,red,sloped]{\small (Lemma~\ref{lem:distrib})};
	      \path (\tikzcdmatrixname-2-1) -- (\tikzcdmatrixname-3-1) coordinate[pos=0.5](aux)
	      (aux) -- (\tikzcdmatrixname-1-3) node[midway,red,xshift=-5mm,sloped]{\small (Corollary~\ref{cor:PropsFinv})};
	      \path (\tikzcdmatrixname-4-1) -- (\tikzcdmatrixname-2-3) coordinate[pos=0.5](aux)
	      (aux) -- (\tikzcdmatrixname-4-3) node[midway,red]{\small (Corollary~\ref{cor:PropsF})};
	    }
	  ]
	  \Gamma
	  \arrow[rrddd, "{\lambda{x}.t\plus\lambda{x}.r}" description,sloped,
	    rounded corners,
	    to path={[pos=0.25]
	      -- ([yshift=5mm]\tikztostart.north)
	      -| ([xshift=2cm]\tikztotarget.east)\tikztonodes
	    -- (\tikztotarget)}
	  ]
	  \arrow[rrddd, "{\lambda{x}.(t\plus r)}" description,sloped,
	    rounded corners,
	    to path={[pos=0.75]
	      -- ([xshift=-23mm]\tikztostart.west)
	      |- ([yshift=-5mm]\tikztotarget.south)\tikztonodes
	    -- (\tikztotarget)}
	  ]
	  \arrow[r, "\Delta", dashed] \arrow[d, "\eta^{ A}", dashed] &[5mm] \Gamma\oplus\Gamma \arrow[r, "\eta^{ A}\oplus\eta^{ A}", dashed] &[1cm] \home{ A}{\Gamma\otimes  A}\oplus \home{ A}{\Gamma\otimes  A} \arrow[d, "\home{ A}{t}\oplus \home{ A}{u}", dashed] \\
	  \home{ A}{\Gamma\otimes  A}\arrow[rru,"\Delta",dashed,sloped] \arrow[d, "\home{ A}{\Delta}", dashed] & & \home{ A}{ B}\oplus \home{ A}{ B}\ar[dd,"{\nabla}",dashed] \\
	  \home{ A}{(\Gamma\otimes  A)\oplus(\Gamma\otimes  A)} \arrow[d, "\home{ A}{t\oplus u}", dashed] \arrow[rruu,sloped, "\gamma", dashed] &&   \\
	  \home{ A}{  B\oplus   B}\arrow[rr,dashed,"\home{ A}{\nabla}",sloped]  \arrow[rruu, "\gamma",sloped, dashed] & & \home{ A}{ B} 
	\end{tikzcd}
      \]

    \item 
      \(
	\vcenter{
	  \infer{\Gamma\vdash\langle\rangle\plus\langle\rangle:\top}
	  {
	    \infer{\Gamma\vdash\langle\rangle:\top}{}
	    &
	    \infer{\Gamma\vdash\langle\rangle:\top}{}
	  }
	}
	\qquad\lra\qquad
	\infer{\Gamma\vdash\langle\rangle:\top}{}
      \)
      \[
	\begin{tikzcd}[labels=description,column sep=2cm,row sep=1cm,
	    execute at end picture={
	      \path
	      (\tikzcdmatrixname-1-1) -- (\tikzcdmatrixname-2-1) coordinate[midway](aux1)
	      (\tikzcdmatrixname-1-2) -- (\tikzcdmatrixname-2-2) coordinate[midway](aux2)
	      (aux1) -- (aux2) node[midway,red]{\small (Terminal object)};
	      \path (\tikzcdmatrixname-1-1) -- (\tikzcdmatrixname-1-2) node[midway,red,yshift=4mm]{\small (Def.)};
	      \path (\tikzcdmatrixname-2-1) -- (\tikzcdmatrixname-2-2) node[midway,red,yshift=-4mm]{\small (Def.)};
	    }
	  ]
	  \Gamma
	  \arrow[r, "{\langle\rangle\plus\langle\rangle}" description,sloped,
	    rounded corners,
	    to path={[pos=0.75]
	      -- ([xshift=-5mm]\tikztostart.west)
	      |- ([xshift=5mm,yshift=-2.5cm]\tikztotarget.east)\tikztonodes
	    |- (\tikztotarget.east)}
	  ]
	  \arrow[r, "{\langle\rangle}" description,sloped,
	    rounded corners,
	    to path={[pos=0.75]
	      -- ([yshift=3mm]\tikztostart.north)
	      |- ([yshift=5mm]\tikztotarget.north)\tikztonodes
	    |- (\tikztotarget.north)}
	  ]
	  & { 0 } \\
	  \Gamma\oplus\Gamma & { 0 \oplus 0 }
	  \arrow["{!}", from=1-1, to=1-2,dashed]
	  \arrow["\nabla"', from=2-2, to=1-2,dashed]
	  \arrow["\Delta"', from=1-1, to=2-1,dashed]
	  \arrow["{{!}\oplus{!}}"', from=2-1, to=2-2,dashed]
	\end{tikzcd}
      \]

    \item
      \(
	\vcenter{
	  \infer{\Gamma\vdash \pair{t_1}{t_2}\plus\pair{u_1}{u_2}:A\with B}
	  {
	    \infer{\Gamma\vdash\pair{t_1}{t_2}:A\with B}{
	      \Gamma\vdash t_1:A & \Gamma\vdash t_2:B
	    }
	    &
	    \infer{\Gamma\vdash\pair{u_1}{u_2}:A\with B}{
	      \Gamma\vdash u_1:A & \Gamma\vdash u_2:B
	    }
	  }
	}
      \)

      \hfill\(
	\lra\qquad
	\vcenter{
	  \infer{\Gamma\vdash \pair{t_1\plus u_1}{t_2\plus u_2}:A\with B}
	  {
	    \infer{\Gamma\vdash t_1\plus u_1:A}
	    {
	      \Gamma\vdash t_1:A & \Gamma\vdash u_1:A
	    }
	    &
	    \infer{\Gamma\vdash t_2\plus u_2:A}
	    {
	      \Gamma\vdash t_2:A & \Gamma\vdash u_2:A
	    }
	  }
	}
      \)
      \[
	\begin{tikzcd}[column sep=3cm,row sep=1cm,labels=description,
	    execute at end picture={
	      \path (\tikzcdmatrixname-1-1) -- (\tikzcdmatrixname-2-1) node[midway,red,xshift=-1.2cm]{\small (Def.)};
	      \path (\tikzcdmatrixname-1-1) -- (\tikzcdmatrixname-2-1) node[midway,red,xshift=3cm]{\small (Def.)};
	      \path (\tikzcdmatrixname-2-1) -- (\tikzcdmatrixname-3-1) coordinate[pos=0.5](aux)
	      (aux) -- (\tikzcdmatrixname-2-2) node[midway,red,red,sloped]{\small (Coherence)};
	      \path (\tikzcdmatrixname-2-2) -- (\tikzcdmatrixname-3-2) coordinate[pos=0.5](aux)
	      (aux) -- (\tikzcdmatrixname-3-1) node[midway,red,sloped]{\small (Lemma~\ref{lem:PropsSumwe})};
	    }
	  ]
	  \Gamma\oplus\Gamma\arrow[dashed,d,"\Delta\oplus\Delta"] 
	  & \Gamma
	  \arrow[dd, "{\pair{t_1\plus u_1}{t_2\plus u_2}}" description,sloped,
	    rounded corners,
	    to path={[pos=0.75]
	      -- ([yshift=3mm]\tikztostart.north)
	      -| ([yshift=-5mm,xshift=-8.5cm]\tikztotarget.south)\tikztonodes
	    -| (\tikztotarget.south)}
	  ]
	  \arrow[dd, "{\pair{t_1}{t_2}\plus\pair{u_1}{u_2}}" description,sloped,
	    rounded corners,
	    to path={[pos=0.75]
	      -- ([xshift=3mm]\tikztostart.east)
	      -| ([xshift=1cm]\tikztotarget.east)\tikztonodes
	    |- (\tikztotarget.east)}
	  ]
	  \arrow[dashed,l,"\Delta"] \\
	  \Gamma\oplus\Gamma\oplus\Gamma\oplus\Gamma\arrow[dashed,r,"t_1\oplus t_2\oplus u_1\oplus u_2"]\arrow[dashed,d,"t_1\oplus u_1\oplus t_2\oplus u_2"] &  A\oplus B\oplus A\oplus B\arrow[d,dashed,"{\nabla}"]\arrow[dashed,dl,"\Id\oplus\sigma\oplus\Id",sloped] \\
	  A\oplus A\oplus B\oplus B\arrow[dashed,r,"{\nabla}\oplus{\nabla}"]  &  A\oplus B
	\end{tikzcd}
      \]

    \item 
      \(
	\vcenter{\infer{\Gamma,\Theta\vdash\elimoplus(t\plus u,x.v,y.w):A}{\infer{\Gamma\vdash t\plus u:B\oplus C}{\Gamma\vdash t:B\oplus C & \Gamma\vdash u:B\oplus C} &  x:B,\Theta\vdash v:A &  y:C,\Theta\vdash w:A}}
      \)

      \hfill\(
	\lra\qquad
	\vcenter{\infer{\Gamma,\Theta\vdash\elimoplus(t,x.v,y.w)\plus\elimoplus(u,x.v,y.w):A}{D_1 & D_2}}
      \)

      where
      \begin{align*}
	D_1&= \vcenter{\infer{\Gamma,\Theta\vdash\elimoplus(t,x.v,y.w):A}{\Gamma\vdash t:B\oplus C &  x:B,\Theta\vdash v:A &  y:C,\Theta\vdash w:A}}\\
	D_2&=\vcenter{\infer{\Gamma,\Theta\vdash\elimoplus(u,x.v,y.w):A}{\Gamma\vdash u:B\oplus C &  x:B,\Theta\vdash v:A &  y:C,\Theta\vdash w:A}}
      \end{align*}

      \[
	\begin{tikzcd}[labels=description,column sep=5mm,row sep=1cm,
	    execute at end picture={
	      \path (\tikzcdmatrixname-1-1) -- (\tikzcdmatrixname-1-2) node[midway,red,yshift=5mm]{\small (Def.)};
	      \path (\tikzcdmatrixname-1-1) -- (\tikzcdmatrixname-2-1) coordinate[midway](aux)
	      (aux) -- (\tikzcdmatrixname-1-2) node[midway,sloped,red]{\small (Corollary~\ref{cor:PropsFinv})};
	      \path
	      (\tikzcdmatrixname-2-1) -- (\tikzcdmatrixname-3-1) coordinate[midway](aux1)
	      (\tikzcdmatrixname-1-2) -- (\tikzcdmatrixname-2-2) coordinate[midway](aux2)
	      (aux1) -- (aux2) node[midway,sloped,red]{\small (Lemma~\ref{lem:distrib})};
	      \path (\tikzcdmatrixname-2-2) -- (\tikzcdmatrixname-3-2) coordinate[pos=0.7](aux)
	      (\tikzcdmatrixname-3-1) -- (aux)  node[pos=0.4,red]{\small (Corollary~\ref{cor:PropsF})};
	      \path
	      (\tikzcdmatrixname-3-1) -- (\tikzcdmatrixname-4-1) coordinate[midway](aux1)
	      (\tikzcdmatrixname-3-2) -- (\tikzcdmatrixname-4-2) coordinate[midway](aux2)
	      (aux1) -- (aux2) node[midway,red]{\small (Naturality of $\nabla$)};
	      \path
	      (\tikzcdmatrixname-4-1) -- (\tikzcdmatrixname-5-1) coordinate[midway](aux1)
	      (\tikzcdmatrixname-4-2) -- (\tikzcdmatrixname-5-2) coordinate[midway](aux2)
	      (aux1) -- (aux2) node[midway,red]{\small (Naturality of $\nabla$)};
	      \path (\tikzcdmatrixname-5-1) -- (\tikzcdmatrixname-5-2) node[midway,red,yshift=-5mm]{\small (Def.)};
	    }
	  ]
	  \Gamma\otimes\Theta
	  \arrow[rdddd, "{\elimoplus(t\plus u,x.v,y.w)}",sloped,
	    rounded corners,
	    to path={[pos=0.25]
	      -- ([yshift=5mm]\tikztostart.north)
	      -| ([xshift=2.2cm]\tikztotarget.east)\tikztonodes
	    -- (\tikztotarget)}
	  ]
	  \arrow[rdddd, "{\elimoplus(t,x.v,y.w)\plus\elimoplus(u,x.v,y.w)}",sloped,
	    rounded corners,
	    to path={[pos=0.75]
	      -- ([xshift=-2.9cm]\tikztostart.west)
	      |- ([yshift=-5mm]\tikztotarget.south)\tikztonodes
	    -- (\tikztotarget)}
	  ]
	  & {(\Gamma\oplus\Gamma)\otimes\Theta} \\
	  {(\Gamma\otimes\Theta)\oplus(\Gamma\otimes\Theta)} & {((B\oplus C)\oplus(B\oplus C))\otimes\Theta} \\
	  {((B\oplus C)\otimes\Theta)\oplus((B\oplus C)\otimes\Theta)} & {(B\oplus C)\otimes\Theta} \\
	  {(B\otimes\Theta\oplus C\otimes\Theta)\oplus(B\otimes\Theta\oplus C\otimes\Theta)} & {B\otimes\Theta\oplus C\otimes\Theta} \\
	  {A\oplus A} & A
	  \arrow["{\coproducto vw}"', from=4-2, to=5-2,dashed]
	  \arrow["d"', from=3-2, to=4-2,dashed]
	  \arrow["\Delta\otimes\Id"', from=1-1, to=1-2,dashed]
	  \arrow["{(t\oplus u)\otimes\Delta}"', from=1-2, to=2-2,dashed]
	  \arrow["\Delta", from=1-1, to=2-1,dashed]
	  \arrow["{(t\otimes\Id)\oplus(u\otimes\Id)}", from=2-1, to=3-1,dashed]
	  \arrow["{d\oplus d}", from=3-1, to=4-1,dashed]
	  \arrow["{\coproducto vw\oplus\coproducto vw}", from=4-1, to=5-1,dashed]
	  \arrow["\nabla", from=5-1, to=5-2,dashed]
	  \arrow["d", from=1-2, to=2-1,dashed,sloped]
	  \arrow["d", from=2-2, to=3-1,dashed,sloped]
	  \arrow["\nabla", from=3-1, to=3-2,dashed]
	  \arrow["\nabla\otimes\Id",from=2-2, to=3-2,dashed]
	  \arrow["\nabla", from=4-1, to=4-2,dashed]
	\end{tikzcd}
      \]

    \item 
      \(
	\vcenter{
	  \infer{\Gamma\vdash \super{t_1}{t_2}\plus\super{u_1}{u_2}:A\odot B}
	  {
	    \infer{\Gamma\vdash\super{t_1}{t_2}:A\odot B}{
	      \Gamma\vdash t_1:A & \Gamma\vdash t_2:B
	    }
	    &
	    \infer{\Gamma\vdash\super{u_1}{u_2}:A\odot B}{
	      \Gamma\vdash u_1:A & \Gamma\vdash u_2:B
	    }
	  }
	}
      \)

      \hfill\(
	\lra\qquad
	\vcenter{
	  \infer{\Gamma\vdash \super{t_1\plus u_1}{t_2\plus u_2}:A\odot B}
	  {
	    \infer{\Gamma\vdash t_1\plus u_1:A}
	    {
	      \Gamma\vdash t_1:A & \Gamma\vdash u_1:A
	    }
	    &
	    \infer{\Gamma\vdash t_2\plus u_2:A}
	    {
	      \Gamma\vdash t_2:A & \Gamma\vdash u_2:A
	    }
	  }
	}
      \)

      This case is analogous to that of pairs.

    \item 
      \(
	\vcenter{\infer{\vdash \escalar s_1\bullet \escalar s_2.\star:\one}{\infer{\vdash \escalar s_2.\star:\one}{}}}
	\qquad\lra\qquad
	\vcenter{\infer{\vdash (\escalar s_1\produ[\mathcal S]\escalar s_2).\star:\one}{}}
      \)

      By Lemma~\ref{lem:PropsS}, we have that for any $I\xlra s I$, $\hat s=s$. In addition, we have that $\llparenthesis\cdot\rrparenthesis$ is a homomorphism. Thus,
      \[
	\sem{\vdash(\escalar s_1\produ[\mathcal S] \escalar s_2).\star:\one } = {\semS{\escalar s_1\produ[\mathcal S]\escalar s_2}}
	=\semS{\escalar s_1}\circ\semS{\escalar s_2}
	=\widehat{\semS{\escalar s_1}}\circ\semS{\escalar s_2}
	=\sem{\vdash \escalar s_1\bullet \escalar s_2.\star:\one}
      \]

    \item 
      \(
	\infer{\Gamma,\Theta\vdash\elimtens(\escalar s\bullet t,xy.v):A}{\Gamma,x:B,y:C\vdash v:A &  \infer{\Theta\vdash\escalar s\bullet t:B\otimes C}{\Theta\vdash t:B\otimes C}}
      \)

      \hfill\(
	\lra\qquad
	\infer{\Gamma,\Theta\vdash\escalar s\bullet \elimtens(t,xy.v):A}{\infer{\Gamma,\Theta\vdash\elimtens(t,xy.v):A}{\Gamma,x:B,y:C\vdash v:A &  {\Theta\vdash t:B\otimes C}}}
      \)

      \[
	\begin{tikzcd}[labels=description,column sep=1cm,
	    execute at end picture={
	      \path (\tikzcdmatrixname-2-1) -- (\tikzcdmatrixname-2-4) node[midway,yshift=1cm,red]{\small (Def.)};
	      \path (\tikzcdmatrixname-2-1) -- (\tikzcdmatrixname-2-4) node[midway,red]{\small (Lemma~\ref{lem:PropsS})};
	      \path
	      (\tikzcdmatrixname-2-1) -- (\tikzcdmatrixname-3-1) coordinate[pos=0.7](aux1)
	      (\tikzcdmatrixname-2-4) -- (\tikzcdmatrixname-3-4) coordinate[pos=0.7](aux2)
	      (aux1) -- (aux2) node[midway,red]{\small (Lemma~\ref{lem:s-nat})};
	      \path (\tikzcdmatrixname-3-1) -- (\tikzcdmatrixname-3-4) node[midway,yshift=-4mm,red]{\small (Def.)};
	    }
	    ]
	    \Gamma\otimes\Theta 
	    \arrow[rrrdd, "{\elimtens(\escalar s\bullet t,xy.v)}",sloped,
	    rounded corners,
	    to path={[pos=0.25]
	    -| ([xshift=8mm]\tikztotarget.east)\tikztonodes
	    -- (\tikztotarget)}
	    ]
	    \arrow[rrrdd, "{\escalar s\bullet \elimtens(t,xy.v)}",sloped,
	    rounded corners,
	    to path={[pos=0.75]
	    -- ([xshift=-4mm]\tikztostart.west)
	    |- ([yshift=-5mm]\tikztotarget.south)\tikztonodes
	    -- (\tikztotarget)}
	    ]
	    \\
	    {\Gamma\otimes B\otimes C} &&[1.5cm]& {\Gamma\otimes B\otimes C} \\
	    A &&& A
	    \arrow["{\Id\otimes t}", dashed, from=1-1, to=2-1]
	    \arrow["{\Id\otimes\widehat{\semS{\escalar s}}}", dashed, from=2-1, to=2-4,bend left=10]
	    \arrow["\widehat{\semS{s}}", dashed, from=2-1, to=2-4,bend right=10]
	    \arrow["v", dashed, from=2-4, to=3-4]
	    \arrow["v", dashed, from=2-1, to=3-1]
	    \arrow["{\widehat{\semS{\escalar s}}}"{description}, dashed, from=3-1, to=3-4]
	\end{tikzcd}
      \]

    \item 
      \(
	\vcenter{
	  \infer{\Gamma\vdash{\escalar s}\bullet\lambda{x}.t:A\multimap B}
	  {
	    \infer{\Gamma\vdash\lambda{x}.t:A\multimap B}
	    {
	      \Gamma,x:A\vdash t:B
	    }
	  }
	}
	\qquad\lra\qquad
	\vcenter{
	  \infer{\Gamma\vdash\lambda{x}.{\escalar s}\bullet t:A\multimap B}
	  {
	    \infer{\Gamma,x:A\vdash{\escalar s}\bullet t:B}
	    {
	      \Gamma,x:A\vdash t:B
	    }
	  }
	}
      \)
      \[
	\begin{tikzcd}[row sep=1.5cm,column sep=3cm,labels=description,
	    execute at end picture={
	      \path (\tikzcdmatrixname-2-1) -- (\tikzcdmatrixname-2-2) node[midway,red]{\small Lemma~\ref{lem:HomeScalEqScal}};
	      \path (\tikzcdmatrixname-1-1) -- (\tikzcdmatrixname-1-2) node[midway,red,yshift=5mm]{\small (Def.)};
	      \path (\tikzcdmatrixname-1-1) -- (\tikzcdmatrixname-1-2) node[midway,red,yshift=-7mm]{\small (Def.)};
	    }
	  ]
	  \Gamma 
	  \arrow[d, "{\lambda x.\escalar s\bullet t}",
	    rounded corners,
	    to path={[pos=0.75]
	      -- ([yshift=5mm]\tikztostart.north)
	      -| ([yshift=-1cm,xshift=6.8cm]\tikztotarget.south)\tikztonodes
	    -| (\tikztotarget.south)}
	  ]
	  \arrow[d, "{{\escalar s}\bullet\lambda x.t}"]
	  \arrow[r, "\eta^{ A}", dashed] & \home{ A}{\Gamma\otimes  A} \arrow[d, "\home{ A}t", dashed]                                           \\
	  \home{ A}{ B}  & \home{ A}{ B} \arrow[l, "\home{ A}{\hat s}", dashed, bend left] \arrow[l, "\hat s", dashed, bend right]
	\end{tikzcd}
      \]

    \item 
      \(
	\vcenter{
	  \infer{\Gamma\vdash \escalar s\bullet \langle\rangle:\top}
	  {
	    \infer{\Gamma\vdash\langle\rangle:\top}{}
	  }
	}
	\qquad\lra\qquad
	\infer{\Gamma\vdash\langle\rangle:\top}{}
      \)
      \[
	\begin{tikzcd}[labels=description,column sep=2cm,row sep=1cm,
	    execute at end picture={
	      \path
	      (\tikzcdmatrixname-1-1) -- (\tikzcdmatrixname-1-3) coordinate[midway](aux)
	      (aux) -- (\tikzcdmatrixname-2-2) node[pos=0.3,red]{\small (Terminal object)};
	      \path (\tikzcdmatrixname-1-1) -- (\tikzcdmatrixname-1-3) node[midway,red,yshift=4mm]{\small (Def.)};
	      \path (\tikzcdmatrixname-1-1) -- (\tikzcdmatrixname-2-2) node[midway,red,xshift=-1cm]{\small (Def.)};
	    }
	  ]
	  \Gamma
	  \arrow[rr, "{\escalar s\bullet\langle\rangle}",
	    rounded corners,
	    to path={[pos=0.75]
	      -- ([xshift=-5mm]\tikztostart.west)
	      |- ([xshift=5mm,yshift=-2.5cm]\tikztotarget.east)\tikztonodes
	    |- (\tikztotarget.east)}
	  ]
	  \arrow[rr, "{\langle\rangle}",
	    rounded corners,
	    to path={[pos=0.75]
	      -- ([yshift=3mm]\tikztostart.north)
	      |- ([yshift=5mm]\tikztotarget.north)\tikztonodes
	    |- (\tikztotarget.north)}
	  ]
	  && { 0 } \\
	  & { 0 } &
	  \arrow["{!}", from=1-1, to=1-3,dashed]
	  \arrow["{!}"', from=1-1, to=2-2,dashed,sloped]
	  \arrow["{\hat{\escalar s}}"', from=2-2, to=1-3,dashed,sloped]
	\end{tikzcd}
      \]

    \item 
      \(
	\vcenter{
	  \infer{\Gamma\vdash{s}\bullet\pair tu:A\with B}
	  {
	    \infer{\Gamma\vdash\pair tu:A\with B}
	    {
	      \Gamma\vdash t:A
	      &
	      \Gamma\vdash u:B
	    }
	  }
	}
	\qquad\lra\qquad
	\vcenter{
	  \infer{\Gamma\vdash\pair{{s}\bullet t}{{s}\bullet u}:A\with B}
	  {
	    \infer{\Gamma\vdash{s}\bullet t:A}
	    {
	      \Gamma\vdash t:A
	    }
	    &
	    \infer{\Gamma\vdash{s}\bullet u:B}
	    {
	      \Gamma\vdash u:B
	    }
	  }
	}
      \)
      \[
	\begin{tikzcd}[row sep=2cm,column sep=4cm,labels=description,
	    execute at end picture={
	      \path (\tikzcdmatrixname-1-1) -- (\tikzcdmatrixname-1-2) node[midway,red,yshift=5mm]{\small (Def.)};
	      \path (\tikzcdmatrixname-1-1) -- (\tikzcdmatrixname-1-2) node[midway,red,yshift=-8mm]{\small (Def.)};
	      \path (\tikzcdmatrixname-2-1) -- (\tikzcdmatrixname-2-2) node[midway,red]{\small Lemma~\ref{lem:PropsS}};
	    }
	  ]
	 \Gamma 
	  \arrow[d, "{s\bullet \pair{t}{u}}",
	    rounded corners,
	    to path={[pos=0.75]
	      -- ([yshift=5mm]\tikztostart.north)
	      -| ([yshift=-1.2cm,xshift=6.5cm]\tikztotarget.south)\tikztonodes
	    -| (\tikztotarget.south)}
	  ]
	  \arrow[d, "{\pair{s\bullet t}{s\bullet u}}"]
	  \arrow[r, "\Delta", dashed] & \Gamma\oplus\Gamma \arrow[d, "t\oplus u", dashed]                                           \\
	  A\oplus  B                          &  A\oplus  B \arrow[l, "\hat s\oplus\hat s", dashed, bend left=20] \arrow[l, "\hat s", dashed, bend right=20]
	\end{tikzcd}
      \]

    \item 
      \(
	\infer{\Gamma,\Theta\vdash\elimoplus(\escalar s\bullet t,x.v,y.w):A}{\infer{\Gamma\vdash\escalar s\bullet t:B\oplus C}{\Gamma\vdash t:B\oplus C} & \Theta,x:B\vdash v:A & \Theta,y:C\vdash w:A}
      \)

      \hfill\(
	\lra\qquad
	\vcenter{\infer{\Gamma,\Theta\vdash\escalar s\bullet \elimoplus(t,x.v,y.w):A}{\infer{\Gamma,\Theta\vdash\elimoplus(t,x.v,y.w):A}{{\Gamma\vdash t:B\oplus C} & \Theta,x:B\vdash v:A & \Theta,y:C\vdash w:A }}}
      \)
      \[
	\begin{tikzcd}[labels=description,column sep=4cm,row sep=1cm,
	    execute at end picture={
	      \path (\tikzcdmatrixname-2-1) -- (\tikzcdmatrixname-2-2) node[midway,yshift=1cm,red]{\small (Def.)};
	      \path (\tikzcdmatrixname-4-1) -- (\tikzcdmatrixname-4-2) node[midway,yshift=-4mm,red]{\small (Def.)};
	      \path (\tikzcdmatrixname-2-1) -- (\tikzcdmatrixname-2-2) node[midway,red]{\small (Lemma~\ref{lem:PropsS})};
	      \path 
	      (\tikzcdmatrixname-2-1) -- (\tikzcdmatrixname-3-1) coordinate[pos=0.6](aux1)
	      (\tikzcdmatrixname-2-2) -- (\tikzcdmatrixname-3-2) coordinate[pos=0.6](aux2)
	      (aux1) -- (aux2) node[midway,red]{\small (Lemma~\ref{lem:s-nat})};
	      \path 
	      (\tikzcdmatrixname-3-1) -- (\tikzcdmatrixname-4-1) coordinate[pos=0.5](aux1)
	      (\tikzcdmatrixname-3-2) -- (\tikzcdmatrixname-4-2) coordinate[pos=0.5](aux2)
	      (aux1) -- (aux2) node[midway,red]{\small (Lemma~\ref{lem:s-nat})};
	    }
	  ]
	  \Gamma\otimes\Theta 
	  \arrow[rddd, "{\elimoplus(\escalar s\bullet t,x.v,y.w)}",sloped,
	    rounded corners,
	    to path={[pos=0.25]
	      -| ([xshift=15mm]\tikztotarget.east)\tikztonodes
	    -- (\tikztotarget)}
	  ]
	  \arrow[rddd, "{\escalar s\bullet \elimoplus(t,x.v,y.w)}",sloped,
	    rounded corners,
	    to path={[pos=0.75]
	      -- ([xshift=-10mm]\tikztostart.west)
	      |- ([yshift=-5mm]\tikztotarget.south)\tikztonodes
	    -- (\tikztotarget)}
	  ]
	  \\
	  {(B\oplus C)\otimes\Theta} & {(B\oplus C)\otimes\Theta} \\
	    {(B\otimes\Theta)\oplus (C\otimes\Theta)} & {(B\otimes\Theta)\oplus (C\otimes\Theta)} \\[1cm]
	  A & A
	  \arrow["t\otimes\Id", from=1-1, to=2-1,dashed]
	  \arrow["{\widehat{\semS{\escalar s}}\otimes\Id}", bend left=10,dashed, from=2-1, to=2-2]
	  \arrow["{\widehat{\semS{\escalar s}}}",bend right=10,dashed, from=2-1, to=2-2]
	  \arrow["d", from=2-2, to=3-2,dashed]
	  \arrow["{\coproducto vw}", from=3-2, to=4-2,dashed]
	  \arrow["d", from=2-1, to=3-1,dashed]
	  \arrow["{\coproducto vw}", from=3-1, to=4-1,dashed]
	  \arrow["{\widehat{\semS{\escalar s}}}", dashed, from=4-1, to=4-2]
	  \arrow["{\widehat{\semS{\escalar s}}}", dashed, from=3-1, to=3-2]
	\end{tikzcd}
      \]

    \item 
      \(
	\vcenter{
	  \infer{\Gamma\vdash{s}\bullet\super tu:A\odot B}
	  {
	    \infer{\Gamma\vdash\super tu:A\odot B}
	    {
	      \Gamma\vdash t:A
	      &
	      \Gamma\vdash u:B
	    }
	  }
	}
	\qquad\lra\qquad
	\vcenter{
	  \infer{\Gamma\vdash\super{{s}\bullet t}{{s}\bullet u}:A\odot B}
	  {
	    \infer{\Gamma\vdash{s}\bullet t:A}
	    {
	      \Gamma\vdash t:A
	    }
	    &
	    \infer{\Gamma\vdash{s}\bullet u:B}
	    {
	      \Gamma\vdash u:B
	    }
	  }
	}
      \)

      This case is analogous to that of pairs.

  \end{itemize}
  {\bf  Inductive cases:}
  The cases where the reduction is determinist (that is, any case but those due to rules ${\relimsup \ell}$ and ${\relimsup r}$), are trivial by composition. Therefore, the only interesting case is that of the non-deterministic rules. 

  The interesting case corresponds to the reductions
  \[
    \infer{K[t]\lra[\escalar p] K[r_1]}{t\lra[p] r_1}
    \qquad
    \infer{K[t]\lra[\escalar q] K[r_2]}{t\lra[q] r_2}
  \]

  We must show that
  \[
    \sem{\Gamma\vdash K[t]:A} = \sumwe{\semS{\escalar p}}{\semS{\escalar q}}\circ(\sem{\Gamma\vdash K[r_1]:A}\oplus\sem{\Gamma\vdash K[r_2]:A})\circ\Delta
  \]
  We proceed by induction on the shape of $K[]$.
  \begin{itemize}
    \item If $K[]=[]$, then this is the basic case of the non-deterministic rules.
      \item If $K[]=K'[]\plus u$. Let 
	\begin{align*}
	  f_1& =\sem{\Gamma\vdash K'[r_1]:A} 
	  & f_2 &=\sem{\Gamma\vdash K'[r_2]:A}\\
	  g&=\sem{\Gamma\vdash u:A}
	\end{align*}
	Then,
	\begin{align*}
	  \sem{\Gamma\vdash K[t]:A}
	  &= \sem{\Gamma\vdash K'[t]\plus u:A}\\
	  &= \nabla\circ(\sem{\Gamma\vdash K'[t]:A}\oplus g)\circ\Delta\\
	  \textrm{(by IH)} &= \nabla\circ((\sumwe{\semS{\escalar p}}{\semS{\escalar q}}\circ(f_1\oplus f_2)\circ\Delta)\oplus( g))\circ\Delta\\
	  \textrm{\color{red}(*)}&= \sumwe{\semS{\escalar p}}{\semS{\escalar q}}\circ((\nabla\circ(f_1\oplus g)\circ\Delta)\oplus(\nabla\circ(f_2\oplus g)\circ\Delta))\circ\Delta\\
	  &= \sumwe{\semS{\escalar p}}{\semS{\escalar q}}\circ(\sem{\Gamma\vdash K[r_1]:A}\oplus\sem{\Gamma\vdash K[r_2]:A})\circ\Delta
	\end{align*}
	Where the equality {\color{red}(*)} is justified by the following commuting diagram.
	\[
	  \begin{tikzcd}[labels=description,column sep=4cm,row sep=1cm,
	    execute at end picture={
	      \path 
	      (\tikzcdmatrixname-1-1) -- (\tikzcdmatrixname-2-1) coordinate[pos=0.5](aux1)
	      (\tikzcdmatrixname-1-2) -- (\tikzcdmatrixname-2-2) coordinate[pos=0.5](aux2)
	      (aux1) -- (aux2) node[midway,red,sloped]{\small Naturality of $\Delta$)};
	      \path (\tikzcdmatrixname-3-1) -- (\tikzcdmatrixname-2-2) node[midway,red,yshift=4mm,xshift=-7mm,sloped]{\small (Lemma~\ref{lem:deltaDelta})};
	      \path 
	      (\tikzcdmatrixname-3-1) -- (\tikzcdmatrixname-4-1) coordinate[pos=0.5](aux1)
	      (\tikzcdmatrixname-2-2) -- (\tikzcdmatrixname-3-2) coordinate[pos=0.5](aux2)
	      (aux1) -- (aux2) node[midway,red,sloped]{\small (Lemma~\ref{lem:delta-nat})};
	      \path (\tikzcdmatrixname-5-1) -- (\tikzcdmatrixname-3-2) node[midway,red,yshift=4mm,xshift=-4mm,sloped]{\small (Lemma~\ref{lem:prop-times-sumwe})};
	      \path (\tikzcdmatrixname-5-1) -- (\tikzcdmatrixname-4-2) node[pos=0.7,red]{\small (Lemma~\ref{lem:sumwe-nat})};
	    }
	  ]
	  \Gamma\ar[d,"\Delta"]\ar[r,"\Delta"] & \Gamma\oplus\Gamma\ar[d,"\Delta\oplus\Delta"] \\
	  \Gamma\oplus\Gamma\ar[r,"\Delta",dashed]\ar[d,"\Delta\oplus\Id"] & (\Gamma\oplus\Gamma)\oplus(\Gamma\oplus\Gamma)\ar[d,"(f_1\oplus g)\oplus(f_2\oplus g)"] \\
	    (\Gamma\oplus\Gamma)\oplus\Gamma\ar[ru,"\delta",dashed]\ar[d,"(f_1\oplus f_2)\oplus g"] &( A\oplus A)\oplus( A\oplus A)\ar[d,"\nabla\oplus\nabla"]\ar[ddl,"\sumwe{\semS{\escalar p}}{\semS{\escalar q}}",sloped,dashed]\\
	  ( A\oplus A)\oplus A\ar[ru,"\delta",dashed]\ar[d,"\sumwe{\semS{\escalar p}}{\semS{\escalar q}}\oplus\Id"] &  A\oplus A\ar[d,"\sumwe{\semS{\escalar p}}{\semS{\escalar q}}"]\\
	   A\oplus A\ar[r,"\nabla"]  & A
	  \end{tikzcd}
	\]

      \item If $K[]=u\plus K'[]$. This case is analogous to the case $K'[]\plus s$.
      \item If $K[]=\escalar s\bullet K'[]$. Let 
	\begin{align*}
	  f_1 &=\sem{\Gamma\vdash K'[r_1]:A} 
	  & f_2&=\sem{\Gamma\vdash K'[r_2]:A}
	\end{align*}
	Then,
	\begin{align*}
	  \sem{\Gamma\vdash K[t]:A}
	  &= \sem{\Gamma\vdash \escalar \escalar s\bullet K'[t]}\\
	  &= \widehat{\semS{\escalar s}}\circ\sem{\Gamma\vdash K'[t]:A}\\
	  \textrm{(by IH)} &= \widehat{\semS{\escalar s}}\circ(\sumwe{\semS{\escalar p}}{\semS{\escalar q}}\circ(f_1\oplus f_2)\circ\Delta)\\
	  \textrm{\color{red}(*)}&= \sumwe{\semS{\escalar p}}{\semS{\escalar q}}\circ((\widehat{\semS{\escalar s}}\circ f_1)\oplus(\widehat{\semS{\escalar s}}\circ f_2))\circ\Delta\\
	  &= \sumwe{\semS{\escalar p}}{\semS{\escalar q}}\circ(\sem{\Gamma\vdash \escalar s\bullet K'[r_1]:A}\oplus\sem{\Gamma\vdash \escalar s\bullet K'[r_2]:A})\circ\Delta\\
	  &= \sumwe{\semS{\escalar p}}{\semS{\escalar q}}\circ(\sem{\Gamma\vdash K[r_1]:A}\oplus\sem{\Gamma\vdash K[r_2]:A})\circ\Delta
	\end{align*}
	Where the equality {\color{red}(*)} is justified by the following commuting diagram.
	\[
	  \begin{tikzcd}[labels=description,column sep=4cm,row sep=1cm,
	      execute at end picture={
		\path
		(\tikzcdmatrixname-2-1) -- (\tikzcdmatrixname-3-1) coordinate[pos=0.5](aux1)
		(\tikzcdmatrixname-2-2) -- (\tikzcdmatrixname-3-2) coordinate[pos=0.5](aux2)
		(aux1) -- (aux2) node[midway,red]{\small (Lemma~\ref{lem:sumwe-nat})};
	      }
	    ]
	    \Gamma\oplus\Gamma\ar[d,"f_1\oplus f_2"] & \Gamma\ar[l,"\Delta"]\\
	    A\oplus A\ar[r,"\sumwe{\semS{\escalar p}}{\semS{\escalar q}}"]\ar[d,"\widehat{\semS s}\oplus\widehat{\semS s}"] & A\ar[d,"\widehat{\semS s}"]\\
	    A\oplus A\ar[r,"\sumwe{\semS{\escalar p}}{\semS{\escalar q}}"] & A
	  \end{tikzcd}
	\]
      \qedhere
    \item If $K[]=\elimone(K'[],u)$.
      Then $\Gamma=\Gamma_1,\Gamma_2$.
      Let 
      \begin{align*}
	f_1 & =\sem{\Gamma_1\vdash K'[r_1]:\one} & f_2&=\sem{\Gamma_1\vdash K'[r_2]:\one}\\
	g&=\sem{\Gamma_2\vdash u:A}
      \end{align*}
      Then,
      \begin{align*}
	\sem{\Gamma\vdash K[t]:A} 
	&=\sem{\Gamma\vdash\elimone(K'[t],u):A}\\
	&=\sem{\Gamma_1,\Gamma_2\vdash\elimone(K'[t],u):A}\\
	&=\lambda\circ(\sem{\Gamma_1\vdash K'[t]:\one}\otimes g)\\
	\textrm{(by IH)} &= \lambda\circ((\sumwe{\semS{\escalar p}}{\semS{\escalar q}}\circ (f_1\oplus f_2)\circ\Delta)\otimes g)\\
	\textrm{\color{red}(*)}&=\sumwe{\semS{\escalar p}}{\semS{\escalar q}}\circ(
	(\lambda\circ(f_1\otimes g))
	\oplus
	(\lambda\circ(f_2\otimes g))
	)\circ\Delta\\
	&=
	\begin{aligned}[t]
	  \sumwe{\semS{\escalar p}}{\semS{\escalar q}}\circ(&\sem{\Gamma\vdash \elimone(K'[r_1],u):A}\\
	  \oplus &\sem{\Gamma\vdash \elimone(K'[r_2],u):A})\circ\Delta
	\end{aligned}
	\\
	&=\sumwe{\semS{\escalar p}}{\semS{\escalar q}}\circ(\sem{\Gamma\vdash K[r_1]:A}\oplus\sem{\Gamma\vdash K[r_2]:A})\circ\Delta
      \end{align*}
      Where the equality {\color{red}(*)} is justified by the following commuting diagram.
      \[
	\begin{tikzcd}[labels=description,column sep=4cm,row sep=1cm,
	    execute at end picture={
	      \path (\tikzcdmatrixname-2-1) -- (\tikzcdmatrixname-1-2) node[midway,red,yshift=4mm,xshift=-6mm,sloped]{\small (Corollary~\ref{cor:PropsFinv})};
	      \path 
	      (\tikzcdmatrixname-2-1) -- (\tikzcdmatrixname-3-1) coordinate[pos=0.5](aux1)
	      (\tikzcdmatrixname-1-2) -- (\tikzcdmatrixname-2-2) coordinate[pos=0.5](aux2)
	      (aux1) -- (aux2) node[midway,red,sloped]{\small (Lemma~\ref{lem:distrib})};
	      \path (\tikzcdmatrixname-4-1) -- (\tikzcdmatrixname-2-2) node[midway,red,yshift=4mm,xshift=-4mm,sloped]{\small ({Corollary~\ref{cor:PropsF}})};
	      \path (\tikzcdmatrixname-4-1) -- (\tikzcdmatrixname-2-2) node[midway,red,yshift=-4mm,xshift=-4mm,sloped]{\small {(Lemma~\ref{lem:sumwe-nat})}};
	    }
	    ]
	    {\Gamma_1}\otimes{\Gamma_2}\ar[d,"\Delta\otimes\Id"]\ar[r,"\Delta"] & ({\Gamma_1}\otimes{\Gamma_2})\oplus({\Gamma_1}\otimes{\Gamma_2})\ar[d,"(f_1\otimes g)\oplus(f_2\otimes g)"]\\
	    ({\Gamma_1}\oplus{\Gamma_1})\otimes{\Gamma_2}\ar[ru,sloped,"d",dashed]\ar[d,"(f_1\oplus f_2)\otimes g"]& ( I\otimes  A)\oplus( I\otimes  A)\ar[d,"\lambda\oplus\lambda"]\ar[ddl,"\sumwe{\semS{\escalar p}}{\semS{\escalar q}}",sloped,dashed]\\
	    ( I\oplus I)\otimes  A\ar[d,"\sumwe{\semS{\escalar p}}{\semS{\escalar q}}\otimes\Id"]\ar[ru,sloped,"d",dashed]&  A\oplus  A\ar[d,"\sumwe{\semS{\escalar p}}{\semS{\escalar q}}"]\\
	    I\otimes  A\ar[r,"\lambda"]&  A
	\end{tikzcd}
      \]

    \item If $K[]=\elimone(u,K'[])$.
      Then $\Gamma=\Gamma_1,\Gamma_2$.
      Let 
      \begin{align*}
	f_1&=\sem{\Gamma_2\vdash K'[r_1]:A}  & f_2&=\sem{\Gamma_2\vdash K'[r_2]:A}\\
	g&=\sem{\Gamma_1\vdash u:\one}
      \end{align*}
      Then,
      \begin{align*}
	\sem{\Gamma\vdash K[t]:A} 
	&=\sem{\Gamma\vdash\elimone(u,K'[t]):A}\\
	&=\sem{\Gamma_1,\Gamma_2\vdash\elimone(u,K'[t]):A}\\
	&=\lambda\circ(g\otimes\sem{\Gamma_2\vdash K'[t]:A})\\
	\textrm{(by IH)} &= \lambda\circ(g\otimes(\sumwe{\semS{\escalar p}}{\semS{\escalar q}}\circ (f_1\oplus f_2)\circ\Delta))\\
	\textrm{\color{red}(*)}&=\sumwe{\semS{\escalar p}}{\semS{\escalar q}}\circ(
	(\lambda\circ(g\otimes f_1))
	\oplus
	(\lambda\circ(g\otimes f_2))
	)\circ\Delta\\
	&=
	\begin{aligned}[t]
	  \sumwe{\semS{\escalar p}}{\semS{\escalar q}}\circ(&\sem{\Gamma\vdash \elimone(u,K'[r_1]):A} \\
	  \oplus&\sem{\Gamma\vdash \elimone(u,K'[r_2]):A})\circ\Delta
	\end{aligned}
	\\
	&=\sumwe{\semS{\escalar p}}{\semS{\escalar q}}\circ(\sem{\Gamma\vdash K[r_1]:A}\oplus\sem{\Gamma\vdash K[r_2]:A})\circ\Delta
      \end{align*}
      Where the equality {\color{red}(*)} is justified by the following commuting diagram.
      \[
	\begin{tikzcd}[labels=description,column sep=4cm,row sep=1cm,
	    execute at end picture={
	      \path (\tikzcdmatrixname-2-1) -- (\tikzcdmatrixname-1-2) node[midway,red,yshift=4mm,xshift=-6mm,sloped]{\small (Corollary~\ref{cor:PropsFinv})};
	      \path 
	      (\tikzcdmatrixname-2-1) -- (\tikzcdmatrixname-3-1) coordinate[pos=0.5](aux1)
	      (\tikzcdmatrixname-1-2) -- (\tikzcdmatrixname-2-2) coordinate[pos=0.5](aux2)
	      (aux1) -- (aux2) node[midway,red,sloped]{\small (Lemma~\ref{lem:distrib})};
	      \path (\tikzcdmatrixname-4-1) -- (\tikzcdmatrixname-2-2) node[midway,red,yshift=4mm,xshift=-4mm,sloped]{\small ({Lemma \ref{lem:PropsF}})};
	      \path (\tikzcdmatrixname-4-1) -- (\tikzcdmatrixname-2-2) node[midway,red,yshift=-4mm,xshift=-4mm,sloped]{\small {(Lemma~\ref{lem:sumwe-nat})}};
	    }
	    ]
	    {\Gamma_1}\otimes{\Gamma_2}\ar[d,"\Id\otimes\Delta"]\ar[r,"\Delta"] & ({\Gamma_1}\otimes{\Gamma_2})\oplus({\Gamma_1}\otimes{\Gamma_2})\ar[d,"(g\otimes f_1)\oplus(g\otimes f_2)"]\\
	    {\Gamma_1}\otimes({\Gamma_2}\oplus{\Gamma_2})\ar[ru,"d_r",dashed]\ar[d,"g\otimes(f_1\oplus f_2)"]& ( I\otimes  A)\oplus( I\otimes  A)\ar[d,"\lambda\oplus\lambda"]\ar[ddl,"\sumwe{\semS{\escalar p}}{\semS{\escalar q}}",dashed,sloped]\\
	    I\otimes( A\oplus  A)\ar[d,"\Id\otimes\sumwe{\semS{\escalar p}}{\semS{\escalar q}}"]\ar[ru,"d_r",dashed]&  A\oplus  A\ar[d,"\sumwe{\semS{\escalar p}}{\semS{\escalar q}}"]\\
	    I\otimes  A\ar[r,"\lambda"]&  A
	\end{tikzcd}
      \]

    \item
      If $K[]=K'[]\otimes u$. Then $A=B\otimes C$ and $\Gamma=\Gamma_1,\Gamma_2$.
      Let 
      \begin{align*}
	f_1&=\sem{\Gamma_1\vdash K'[r_1]:B}   &  f_2&=\sem{\Gamma_2\vdash K'[r_2]:B}\\
	g&=\sem{\Gamma_2\vdash u:C}
      \end{align*}
      Then,
      \begin{align*}
	\sem{\Gamma\vdash K[t]:A} 
	&=\sem{\Gamma_1,\Gamma_2\vdash K'[t]\otimes u:B\otimes C}\\
	&=\sem{\Gamma_1\vdash K'[t]:B}\otimes g\\
	\textrm{(by IH)}&=(\sumwe{\semS{\escalar p}}{\semS{\escalar q}}\circ(f_1\oplus f_2)\circ\Delta)\otimes g\\
	\textrm{\color{red}(*)}&=\sumwe{\semS{\escalar p}}{\semS{\escalar q}}\circ(
	(f_1\otimes g)
	\oplus
	(f_2\otimes g)
	)\circ\Delta\\
	&=
	\begin{aligned}[t]
	  \sumwe{\semS{\escalar p}}{\semS{\escalar q}}\circ(
	  &\sem{\Gamma_1,\Gamma_2\vdash K'[r_1]\otimes u:B\otimes C}\\
	  \oplus 
	  &\sem{\Gamma_1,\Gamma_2\vdash K'[r_2]\otimes u:B\otimes C}
	  )\circ\Delta
	\end{aligned}
	\\
	&=\sumwe{\semS{\escalar p}}{\semS{\escalar q}}\circ(\sem{\Gamma\vdash K[r_1]:A}\oplus\sem{\Gamma\vdash K[r_2]:A})\circ\Delta
      \end{align*}

      Where the equality {\color{red}(*)} is justified by the following commuting diagram.
      \[
	\begin{tikzcd}[labels=description,column sep=3.2cm,row sep=2cm,
	    execute at end picture={
	      \path (\tikzcdmatrixname-2-1) -- (\tikzcdmatrixname-1-2) coordinate[pos=0.5](aux)
	      (\tikzcdmatrixname-1-1) -- (aux) node[midway,red]{\small (Corollary~\ref{cor:PropsFinv})};
	      \path (\tikzcdmatrixname-3-1) -- (\tikzcdmatrixname-2-2) coordinate[pos=0.5](aux)
	      (\tikzcdmatrixname-3-2) -- (aux) node[midway,red]{\small ({Lemma \ref{lem:PropsF}})};
	      \path (\tikzcdmatrixname-2-1) -- (\tikzcdmatrixname-2-2) node[midway,red]{\small (Lemma~\ref{lem:distrib})};
	    }]
	    {{\Gamma_1}\otimes{\Gamma_2}} & {({\Gamma_1}\otimes{\Gamma_2})\oplus({\Gamma_1}\otimes{\Gamma_2})} \\
	    {({\Gamma_1}\oplus{\Gamma_1})\otimes {\Gamma_2}} & {( B\otimes  C)\oplus( B\otimes  C)} \\
	    {( B\oplus  B)\otimes  C} & { B\otimes  C}
	    \arrow["\Delta\otimes\Id", from=1-1, to=2-1]
	    \arrow["{(f_1\oplus f_2)\otimes g}", from=2-1, to=3-1]
	    \arrow["{\sumwe{\semS{\escalar p}}{\semS{\escalar q}}\otimes\Id}", from=3-1, to=3-2]
	    \arrow["\Delta", from=1-1, to=1-2]
	    \arrow["{(f_1\otimes g)\oplus(f_2\otimes g)}", from=1-2, to=2-2]
	    \arrow["{\sumwe{\semS{\escalar p}}{\semS{\escalar q}}}", from=2-2, to=3-2]
	    \arrow["d",sloped, dashed, from=3-1, to=2-2]
	    \arrow["d",sloped, dashed, from=2-1, to=1-2]
	\end{tikzcd}
      \]

    \item If $K[]=u\otimes K'[]$. This case is analogous to the case $K'[]\otimes u$.

    \item 
      If $K[]=\elimtens(K'[],xy.u)$. Then $\Gamma=\Gamma_1,\Gamma_2$.
      Let 
      \begin{align*}
	f_1&=\sem{\Gamma_2\vdash K'[r_1]:B\otimes C} & f_2&=\sem{\Gamma_2\vdash K'[r_2]:B\otimes C}\\
	g&=\sem{\Gamma_1,x:B,y:C\vdash u:A}
      \end{align*}
      Then,
      \begin{align*}
	\sem{\Gamma\vdash K[t]:A} 
	&= \sem{\Gamma_1,\Gamma_2\vdash \elimtens(K'[t],xy.u):A} \\
	&= g\circ(\Id\otimes\sem{\Gamma_2\vdash K'[t]:B\otimes C}) \\
	\textrm{(by IH)}&= g\circ(\Id\otimes(\sumwe{\semS{\escalar p}}{\semS{\escalar q}}\circ(f_1\oplus f_2)\circ\Delta)) \\
	\textrm{\color{red}(*)}&=\sumwe{\semS{\escalar p}}{\semS{\escalar q}}\circ(
	(g\circ(\Id\otimes f_1))
	\oplus
	(g\circ(\Id\otimes f_2))
	)\circ\Delta\\
	&=
	\begin{aligned}[t]
	  \sumwe{\semS{\escalar p}}{\semS{\escalar q}}\circ(&\sem{\Gamma_1,\Gamma_2\vdash \elimtens(K'[r_1],xy.u):A}\\
	  \oplus&\sem{\Gamma_1,\Gamma_2\vdash \elimtens(K'[r_2],xy.u):A})\circ\Delta
	\end{aligned}
	\\
	&=\sumwe{\semS{\escalar p}}{\semS{\escalar q}}\circ(\sem{\Gamma\vdash K[r_1]:A}\oplus\sem{\Gamma\vdash K[r_2]:A})\circ\Delta
      \end{align*}
      Where the equality {\color{red}(*)} is justified by the following commuting diagram.

      \[
	\begin{tikzcd}[labels=description,column sep=2.5cm,row sep=huge,
	    execute at end picture={
	      \path (\tikzcdmatrixname-2-1) -- (\tikzcdmatrixname-1-2) coordinate[pos=0.5](aux)
	      (\tikzcdmatrixname-1-1) -- (aux) node[midway,red]{\small (Corollary~\ref{cor:PropsFinv})};
	      \path (\tikzcdmatrixname-3-1) -- (\tikzcdmatrixname-2-2) coordinate[pos=0.5](aux)
	      (\tikzcdmatrixname-3-2) -- (aux) node[midway,red]{\small ({Lemma \ref{lem:PropsF}})};
	      \path (\tikzcdmatrixname-2-1) -- (\tikzcdmatrixname-2-2) node[midway,red]{\small (Lemma~\ref{lem:distrib})};
	      \path (\tikzcdmatrixname-3-1) -- (\tikzcdmatrixname-4-1) coordinate[pos=0.5](aux1)
	      (\tikzcdmatrixname-3-2) -- (\tikzcdmatrixname-4-2) coordinate[pos=0.5](aux2)
	      (aux1) -- (aux2) node[midway,red]{\small (Lemma~\ref{lem:sumwe-nat})};
	    }]
	    {{\Gamma_1}\otimes{\Gamma_2}} & {{\Gamma_1}\otimes({\Gamma_2}\oplus{\Gamma_2})} \\
	    {({\Gamma_1}\otimes{\Gamma_2})\oplus({\Gamma_1}\otimes{\Gamma_2})} & {{\Gamma_1}\otimes(( B\otimes  C)\oplus( B\otimes  C))} \\
	    {({\Gamma_1}\otimes  B\otimes  C)\oplus({\Gamma_2}\otimes  B\otimes  C)} & {{\Gamma_1}\otimes  B\otimes  C} \\
	    { C\oplus  C} &  C
	    \arrow["\Id\otimes\Delta", from=1-1, to=1-2]
	    \arrow["{\Id\otimes(f_1\oplus f_2)}", from=1-2, to=2-2]
	    \arrow["{\Id\otimes\sumwe{\semS{\escalar p}}{\semS{\escalar q}}}", from=2-2, to=3-2]
	    \arrow["g", from=3-2, to=4-2]
	    \arrow["\Delta"', from=1-1, to=2-1]
	    \arrow["{(\Id\otimes f_1)\oplus(\Id\otimes f_2)}"', from=2-1, to=3-1]
	    \arrow["{g\oplus g}"', from=3-1, to=4-1]
	    \arrow["{\sumwe{\semS{\escalar p}}{\semS{\escalar q}}}"', from=4-1, to=4-2]
	    \arrow["d"{description}, dashed, from=1-2, to=2-1]
	    \arrow["d"{description}, dashed, from=2-2, to=3-1]
	    \arrow["{\sumwe{\semS{\escalar p}}{\semS{\escalar q}}}"{description}, dashed, from=3-1, to=3-2]
	\end{tikzcd}
      \]

    \item 
      If $K[]=\elimtens(u,xy.K'[])$.
      Then $\Gamma=\Gamma_1,\Gamma_2$.
      Let 
      \begin{align*}
	f_1&=\sem{\Gamma_1,x:B,y:C\vdash K'[r_1]:A} & f_2&=\sem{\Gamma_1,x:B,y:C\vdash K'[r_2]:A}\\
	g&=\sem{\Gamma_2\vdash u:B\otimes C}
      \end{align*}
      \begin{align*}
	\sem{\Gamma\vdash K[t]:A} 
	&= \sem{\Gamma_1,\Gamma_2\vdash \elimtens(u,xy.K'[t]):A} \\
	&= \sem{\Gamma_1,x:B,y:C\vdash K'[t]:A}\circ(\Id\otimes g) \\
	\textrm{(by IH)}&= \sumwe{\semS{\escalar p}}{\semS{\escalar q}}\circ(f_1\oplus f_2)\circ\Delta\circ (\Id\otimes g) \\
	\textrm{\color{red}(*)}
	&=\sumwe{\semS{\escalar p}}{\semS{\escalar q}}\circ(f_1\oplus f_2)\circ(
	(\Id\otimes g)
	\oplus
	(\Id\otimes g)
	)\circ\Delta\\
	&=\sumwe{\semS{\escalar p}}{\semS{\escalar q}}\circ(
	(f_1\circ(\Id\otimes g))
	\oplus
	(f_2\circ(\Id\otimes g))
	)\circ\Delta\\
	&=
	\begin{aligned}[t]
	  \sumwe{\semS{\escalar p}}{\semS{\escalar q}}\circ(&\sem{\Gamma_1,\Gamma_2\vdash \elimtens(u,xy.K'[r_1]):A}\\
	  \oplus&\sem{\Gamma_1,\Gamma_2\vdash \elimtens(u,xy.K'[r_2]):A})\circ\Delta
	\end{aligned}
	\\
	&=\sumwe{\semS{\escalar p}}{\semS{\escalar q}}\circ(\sem{\Gamma\vdash K[r_1]:A}\oplus\sem{\Gamma\vdash K[r_2]:A})\circ\Delta
      \end{align*}
      Where the equality {\color{red}(*)} is justified by the following commuting diagram.
      \[
	\begin{tikzcd}[labels=description,column sep=3.5cm,row sep=1cm,
	    execute at end picture={
	      \path (\tikzcdmatrixname-1-1) -- (\tikzcdmatrixname-2-1) coordinate[pos=0.5](aux1)
	      (\tikzcdmatrixname-1-2) -- (\tikzcdmatrixname-2-2) coordinate[pos=0.5](aux2)
	      (aux1) -- (aux2) node[midway,red]{\small (Naturality of $\Delta$)};
	    }]
	    {{\Gamma_1}\otimes{\Gamma_2}} & {{\Gamma_1}\otimes  B\otimes  C} \\
	    {({\Gamma_1}\otimes{\Gamma_2})\oplus({\Gamma_1}\otimes{\Gamma_2})} & {({\Gamma_1}\otimes  B\otimes  C)\oplus({\Gamma_1}\otimes  B\otimes  C)}\ar[d,"f_1\oplus f_2"] \\
	    A & A\oplus A\ar[l,"\sumwe{\semS{\escalar p}}{\semS{\escalar q}}"]
	    \arrow["{\Id\otimes g}", from=1-1, to=1-2]
	    \arrow["\Delta", from=1-2, to=2-2]
	    \arrow["\Delta"', from=1-1, to=2-1]
	    \arrow["{(\Id\otimes g)\oplus(\Id\otimes g)}"{description}, from=2-1, to=2-2]
	\end{tikzcd}
      \]

    \item 
      Let $K[]=\lambda x.K'[]$.
      Then $A=B\multimap C$.
      Let 
      \begin{align*}
	f_1&=\sem{\Gamma,x:B\vdash K'[r_1]:C} & f_2 & =\sem{\Gamma,x:B\vdash K'[r_2]:C}
      \end{align*}
      Then,
      \begin{align*}
	\sem{\Gamma\vdash K[t]:A}
	&=\sem{\Gamma\vdash \lambda x.K'[t]:B\multimap C}\\
	&=\home{ B}{\sem{\Gamma,x:B\vdash K'[t]:C}}\circ\eta^{ B} \\
	\textrm{(by IH)} &=\home{ B}{\sumwe{\semS{\escalar p}}{\semS{\escalar q}}\circ(f_1\oplus f_2)\circ\Delta}\circ\eta^{ B}
	\\
	\textrm{\color{red}(*)}  & =
	\sumwe{\semS{\escalar p}}{\semS{\escalar q}}\circ(
	\home{ B}{f_1}\circ\eta^{ B}
	\oplus
	\home{ B}{f_2}\circ\eta^{ B}
	)\circ\Delta\\
	& = 
	\begin{aligned}[t]
	  \sumwe{\semS{\escalar p}}{\semS{\escalar q}}\circ(&\sem{\Gamma\vdash \lambda x.K'[r_1]:B\multimap C}\\
	  \oplus&\sem{\Gamma\vdash\lambda x.K'[r_2]:B\multimap C})\circ\Delta
	\end{aligned}
	\\
	& = \sumwe{\semS{\escalar p}}{\semS{\escalar q}}\circ(\sem{\Gamma\vdash K[r_1]:A}\oplus\sem{\Gamma\vdash K[r_2]:A})\circ\Delta
      \end{align*}
      Where the equality {\color{red}(*)} is justified by the following commuting diagram.
      \[
	\begin{tikzcd}[labels=description,column sep=2cm,row sep=1.5cm,
	    execute at end picture={
	      \path (\tikzcdmatrixname-1-1) -- (\tikzcdmatrixname-3-1) coordinate[pos=0.5](aux)
	      (aux) -- (\tikzcdmatrixname-1-2) node[midway,red,xshift=-5mm,sloped]{\small (Naturality of $\Delta$)};
	      \path (\tikzcdmatrixname-1-2) -- (\tikzcdmatrixname-2-2) coordinate[pos=0.5](aux)
	      (\tikzcdmatrixname-3-1) -- (aux) node[midway,red,sloped,xshift=-1cm]{\small (Corollary~\ref{cor:PropsFinv})};
	      \path (\tikzcdmatrixname-3-1) -- (\tikzcdmatrixname-4-1) coordinate[pos=0.5](aux1)
	      (\tikzcdmatrixname-2-2) -- (\tikzcdmatrixname-3-2) coordinate[pos=0.5](aux2)
	      (aux1) -- (aux2) node[midway,red,sloped,yshift=-2mm]{\small (Lemma~\ref{lem:distrib})};
	      \path (\tikzcdmatrixname-3-2) -- (\tikzcdmatrixname-4-2) coordinate[pos=0.5](aux2)
	      (\tikzcdmatrixname-4-1) -- (aux2) node[midway,red,sloped]{\small (Corollary~\ref{cor:PropsF})};
	    }
	    ]
	    \Gamma\ar[d,"\Delta"]\ar[r,"\eta^{ B}"] & \home{ B}{\Gamma\otimes  B}\ar[d,"\home{ B}{\Delta}"]\ar[ddl,"\Delta",dashed,sloped] \\
	    \Gamma\oplus\Gamma\ar[d,"\eta^{ B}\oplus\eta^{ B}"]& \home{ B}{(\Gamma\otimes  B)\oplus(\Gamma\otimes  B)}\ar[d,"{\home{ B}{f_1\oplus f_2}}"]\ar[dl,"\gamma",dashed,sloped]  \\
	    \home{ B}{\Gamma\otimes  B}\oplus\home{ B}{\Gamma\otimes  B}\ar[d,"\home{ B}{f_1}\oplus\home{ B}{f_2}"]& \home{ B}{ C\oplus  C}\ar[d,"{\home{ B}{\sumwe{\semS{\escalar p}}{\semS{\escalar q}}}}"]\ar[dl,"\gamma",dashed] \\
	    \home{ B}{ C}\oplus\home{ B}{ C}\ar[r,"\sumwe{\semS{\escalar p}}{\semS{\escalar q}}",sloped]&  \home{ B}{ C} 
	\end{tikzcd}
      \]

    \item 
      If $K[]=K'[]u$. Then $\Gamma=\Gamma_1,\Gamma_2$. Let 
      \begin{align*}
	f_1&=\sem{\Gamma_1\vdash K'[r_1]:B\multimap A} & f_2 & =\sem{\Gamma_1\vdash K'[r_2]:B\multimap A}\\
	g&=\sem{\Gamma_2\vdash u:B}
      \end{align*}
      Then,
      \begin{align*}
	\sem{\Gamma\vdash K[t]:A} 
	&=\sem{\Gamma\vdash K'[t]u:A}\\
	&=\varepsilon\circ(\sem{\Gamma_1\vdash K'[t]:B\multimap A}\otimes g)\\
	\textrm{(by IH)} &= \varepsilon\circ((\sumwe{\semS{\escalar p}}{\semS{\escalar q}}\circ(f_1\oplus f_2)\circ\Delta)\otimes g)\\
	\textrm{\color{red}(*)}&=\sumwe{\semS{\escalar p}}{\semS{\escalar q}}\circ(
	(\varepsilon\circ(f_1\otimes g))
	\oplus
	(\varepsilon\circ(f_2\otimes g))
	)\circ\Delta\\
	&=\sumwe{\semS{\escalar p}}{\semS{\escalar q}}\circ(\sem{\Gamma\vdash K'[r_1]u:A}\oplus\sem{\Gamma\vdash K'[r_2]u:A})\circ\Delta\\
	&=\sumwe{\semS{\escalar p}}{\semS{\escalar q}}\circ(\sem{\Gamma\vdash K[r_1]:A}\oplus\sem{\Gamma\vdash K[r_2]:A})\circ\Delta
      \end{align*}
      Where the equality {\color{red}(*)} is justified by the following commuting diagram.
      \[
	\begin{tikzcd}[column sep=2cm,row sep=1cm,labels=description,
	    execute at end picture={
	      \path (\tikzcdmatrixname-2-1) -- (\tikzcdmatrixname-1-2) node[midway,red,yshift=4mm,xshift=-4mm,sloped]{\small (Corollary~\ref{cor:PropsFinv})};
	      \path 
	      (\tikzcdmatrixname-2-1) -- (\tikzcdmatrixname-3-1) coordinate[pos=0.5](aux1)
	      (\tikzcdmatrixname-1-2) -- (\tikzcdmatrixname-2-2) coordinate[pos=0.5](aux2)
	      (aux1) -- (aux2) node[pos=0.4,red,sloped]{\small (Lemma~\ref{lem:distrib})};
	      \path (\tikzcdmatrixname-4-1) -- (\tikzcdmatrixname-2-2) node[midway,red,yshift=4mm,xshift=-3mm,sloped]{\small ({Lemma \ref{lem:PropsF}})};
	      \path (\tikzcdmatrixname-4-1) -- (\tikzcdmatrixname-2-2) node[midway,red,yshift=-4mm,xshift=-4mm,sloped]{\small {(Lemma~\ref{lem:sumwe-nat})}};
	    }
	    ]
	    {\Gamma_1}\otimes{\Gamma_2}\ar[r,"\Delta"]\ar[d,"\Delta\otimes\Id"] & ({\Gamma_1}\otimes{\Gamma_2})\oplus({\Gamma_1}\otimes{\Gamma_2})\ar[d,"(f_1\otimes g)\oplus(f_2\otimes g)"]\\
	    ({\Gamma_1}\oplus{\Gamma_1})\otimes{\Gamma_2}\ar[d,"(f_1\oplus f_2)\otimes g"]\ar[ur,"d",dashed,sloped]& (\home{ B}{ A}\otimes  B)\oplus(\home{ B}{ A}\otimes  B)\ar[d,"\varepsilon\oplus\varepsilon"]\ar[ddl,"\sumwe{\semS{\escalar p}}{\semS{\escalar q}}",sloped,dashed]\\
	    (\home{ B}{ A}\oplus\home{ B}{ A})\otimes  B\ar[d,"\sumwe{\semS{\escalar p}}{\semS{\escalar q}}\otimes\Id"]\ar[ur,"d",dashed,sloped] &  A\oplus  A\ar[d,"\sumwe{\semS{\escalar p}}{\semS{\escalar q}}"]\\
	    \home{ B}{ A}\otimes  B\ar[r,"\varepsilon"]&  A
	\end{tikzcd}
      \]

    \item 
      If $K[]=uK'[]$. Then $\Gamma=\Gamma_1,\Gamma_2$. Let 
      \begin{align*}
	f_1&=\sem{\Gamma_2\vdash K'[r_1]:B} & f_2=\sem{\Gamma_2\vdash K'[r_2]:B}\\
	g&=\sem{\Gamma_1\vdash u:B\multimap A}
      \end{align*}
      Then,
      \begin{align*}
	\sem{\Gamma\vdash K[t]:A} 
	&=\sem{\Gamma\vdash uK'[t]:A}\\
	&=\varepsilon\circ(g\otimes\sem{\Gamma_2\vdash K'[t]:B})\\
	\textrm{(by IH)} &= \varepsilon\circ(g\otimes(\sumwe{\semS{\escalar p}}{\semS{\escalar q}}\circ(f_1\oplus f_2)\circ\Delta))\\
	\textrm{\color{red}(*)}&=\sumwe{\semS{\escalar p}}{\semS{\escalar q}}\circ(
	(\varepsilon\circ(g\otimes f_1))
	\oplus
	(\varepsilon\circ(g\otimes f_2))
	)\circ\Delta\\
	&=\sumwe{\semS{\escalar p}}{\semS{\escalar q}}\circ(\sem{\Gamma\vdash uK'[r_1]:A}\oplus\sem{\Gamma\vdash uK'[r_2]:A})\circ\Delta\\
	&=\sumwe{\semS{\escalar p}}{\semS{\escalar q}}\circ(\sem{\Gamma\vdash K[r_1]:A}\oplus\sem{\Gamma\vdash K[r_2]:A})\circ\Delta
      \end{align*}
      Where the equality {\color{red}(*)} is justified by the following commuting diagram.
      \[
	\begin{tikzcd}[column sep=3cm,row sep=1cm,labels=description,
	    execute at end picture={
	      \path (\tikzcdmatrixname-2-1) -- (\tikzcdmatrixname-1-2) node[midway,red,yshift=4mm,xshift=-4mm,sloped]{\small (Corollary~\ref{cor:PropsFinv})};
	      \path 
	      (\tikzcdmatrixname-2-1) -- (\tikzcdmatrixname-3-1) coordinate[pos=0.5](aux1)
	      (\tikzcdmatrixname-1-2) -- (\tikzcdmatrixname-2-2) coordinate[pos=0.5](aux2)
	      (aux1) -- (aux2) node[midway,red,sloped]{\small (Lemma~\ref{lem:distrib})};
	      \path (\tikzcdmatrixname-4-1) -- (\tikzcdmatrixname-2-2) node[midway,red,yshift=4mm,xshift=-4mm,sloped]{\small ({Lemma \ref{lem:PropsF}})};
	      \path (\tikzcdmatrixname-4-1) -- (\tikzcdmatrixname-2-2) node[midway,red,yshift=-4mm,xshift=-4mm,sloped]{\small ({Lemma~\ref{lem:sumwe-nat}})};
	    }
	    ]
	    {\Gamma_1}\otimes{\Gamma_2}\ar[r,"\Delta"]\ar[d,"\Id\otimes\Delta"] & ({\Gamma_1}\otimes{\Gamma_2})\oplus({\Gamma_1}\otimes{\Gamma_2})\ar[d,"(g\otimes f_1)\oplus(g\otimes f_2)"]\\
	    {\Gamma_1}\otimes({\Gamma_2}\oplus{\Gamma_2})\ar[d,"g\otimes (f_1\oplus f_2)"]\ar[ur,"(\sigma\oplus\sigma)\circ d\circ\sigma",dashed,sloped]& (\home{ B}{ A}\otimes B)\oplus(\home{ B}{ A}\otimes B)\ar[d,"\varepsilon\oplus\varepsilon"]\ar[ddl,"\sumwe{\semS{\escalar p}}{\semS{\escalar q}}",sloped,dashed]\\
	    \home{ B}{ A}\otimes( B\oplus B)\ar[d,"\Id\otimes\sumwe{\semS{\escalar p}}{\semS{\escalar q}}"]\ar[ur,"(\sigma\oplus\sigma)\circ d\circ\sigma",dashed,sloped] &  A\oplus A\ar[d,"\sumwe{\semS{\escalar p}}{\semS{\escalar q}}"]\\
	    \home{ B}{ A}\otimes B\ar[r,"\varepsilon"]&  A
	\end{tikzcd}
      \]

    \item 
      If $K[]=\elimzero(K'[])$. Then $\Gamma=\Gamma_1,\Gamma_2$.
      Let
      \begin{align*}
	f_1&=\sem{\Gamma_1\vdash K'[r_2]:\zero} & f_2&=\sem{\Gamma_1\vdash K'[r_2]:\zero}
      \end{align*}
      Then,
      \begin{align*}
	\sem{\Gamma\vdash K[t]:A} &= \sem{\Gamma_1,\Gamma_2\vdash\elimzero(K'[t]):A}\\
	&=0\circ(\sem{\Gamma_1\vdash K'[t]:\zero}\otimes\Id)\\
	\textrm{(by IH)}&= 0\circ(\sumwe{\semS{\escalar p}}{\semS{\escalar q}}\circ (f_1\oplus f_2)\circ\Delta)\otimes\Id)\\
	\textrm{\color{red}(*)}&=\sumwe{\semS{\escalar p}}{\semS{\escalar q}}\circ((0\circ(f_1\otimes\Id))\oplus(0\circ(f_2\otimes\Id)))\circ\Delta\\
	&=\sumwe{\semS{\escalar p}}{\semS{\escalar q}}\circ(\sem{\Gamma\vdash\elimzero(K'[r_1]):A}\oplus\sem{\Gamma\vdash\elimzero(K'[r_2]):A})\circ\Delta\\
	&=\sumwe{\semS{\escalar p}}{\semS{\escalar q}}\circ(\sem{\Gamma\vdash K[r_1]:A}\oplus\sem{\Gamma\vdash K[r_2]:A})\circ\Delta
      \end{align*}
      Where the equality {\color{red}(*)} is justified by the following commuting diagram.
      \[
	\begin{tikzcd}[labels=description,column sep=2cm,row sep=1cm,
	    execute at end picture={
	      \path (\tikzcdmatrixname-1-1) -- (\tikzcdmatrixname-4-2) node[midway,red,yshift=4mm,xshift=-4mm,sloped]{\small ($0\circ f=0$)};
	      \path (\tikzcdmatrixname-1-1) -- (\tikzcdmatrixname-4-2) node[midway,red,yshift=-4mm,xshift=-4mm,sloped]{\small ($0\circ f=0$)};
	    }
	    ]
	    {\Gamma_1}\otimes{\Gamma_2}\ar[dddr,dashed,"0",sloped]\ar[d,"\Delta\otimes\Id"]\ar[r,"\Delta"] & ({\Gamma_1}\otimes{\Gamma_2})\oplus({\Gamma_1}\otimes{\Gamma_2})\ar[d,"(f_1\otimes\Id)\oplus(f_2\otimes\Id)"]\\
	    ({\Gamma_1}\oplus{\Gamma_1})\otimes{\Gamma_2}\ar[d,"(f_1\oplus f_2)\otimes\Id"]& ( 0 \otimes{\Gamma_2})\oplus( 0 \otimes{\Gamma_2})\ar[d,"0\oplus 0"]\ar[dd,bend right=50,dashed,"0"]\\
	    ( 0 \oplus 0 )\otimes{\Gamma_2}\ar[d,"\sumwe{\semS{\escalar p}}{\semS{\escalar q}}\otimes\Id"]&  A\oplus  A\ar[d,"\sumwe{\semS{\escalar p}}{\semS{\escalar q}}"]\\
	    0 \otimes  A\ar[r,"0"]&  A
	\end{tikzcd}
      \]

    \item 
      If $K[]=\pair{K'[]}u$, then $A=B\with C$. 
      Let 
      \begin{align*}
	f_1&=\sem{\Gamma\vdash K'[r_1]:B} &
	f_2&=\sem{\Gamma\vdash K'[r_2]:B} \\
	g&=\sem{\Gamma\vdash u:C}
      \end{align*}
      Then,
      \begin{align*}
	\sem{\Gamma\vdash K[t]:A} 
	& = \sem{\Gamma\vdash\super{K'[t]}u:B\with C} \\
	& = (\sem{\Gamma\vdash K'[t]:B}\oplus g)\circ\Delta\\
	\textrm{(by IH)}&= ((\sumwe{\semS{\escalar p}}{\semS{\escalar q}}\circ(f_1\oplus f_2)\circ\Delta)\oplus g)\circ\Delta\\
	\textrm{\color{red}(*)}&=\sumwe{\semS{\escalar p}}{\semS{\escalar q}}\circ(
	(
	(f_1\oplus g)\circ\Delta
	)
	\oplus
	(
	(f_2\oplus g)\circ\Delta
	)
	)\circ\Delta\\
	&=
	\begin{aligned}
	  \sumwe{\semS{\escalar p}}{\semS{\escalar q}}\circ(&\sem{\Gamma\vdash\super{K'[r_1]}u:B\with C}\\
	  \oplus&\sem{\Gamma\vdash\super{K'[r_2]}u:B\with C})\circ\Delta
	\end{aligned}
	\\
	&=\sumwe{\semS{\escalar p}}{\semS{\escalar q}}\circ(\sem{\Gamma\vdash K[r_1]:A}\oplus\sem{\Gamma\vdash K[r_2]:A})\circ\Delta
      \end{align*}
      Where the equality {\color{red}(*)} is justified by the following commuting diagram.
      \[
	\begin{tikzcd}[labels=description,column sep=4cm,row sep=1cm,
	    execute at end picture={
	      \path
	      (\tikzcdmatrixname-1-2) -- (\tikzcdmatrixname-2-2) coordinate[pos=0.5](aux)
	      (\tikzcdmatrixname-1-1) -- (\tikzcdmatrixname-2-1) coordinate[pos=0.5](aux1)
	      (aux) -- (\tikzcdmatrixname-1-3) node[midway,red,sloped]{\small (Lemma~\ref{lem:deltaDelta})}
	      (aux1) -- (aux) node[midway,red]{\small (Naturality of $\Delta$)};
	      \path 
	      (\tikzcdmatrixname-2-2) -- (\tikzcdmatrixname-3-2) coordinate[pos=0.5](aux1)
	      (\tikzcdmatrixname-1-3) -- (\tikzcdmatrixname-2-3) coordinate[pos=0.5](aux2)
	      (aux1) -- (aux2) node[midway,red,sloped]{\small (Lemma~\ref{lem:delta-nat})};
	      \path 
	      (\tikzcdmatrixname-2-3) -- (\tikzcdmatrixname-3-3) coordinate[pos=0.5](aux1)
	      (\tikzcdmatrixname-3-2) -- (aux1) node[midway,red,sloped]{\small (Lemma~\ref{lem:prop-times-sumwe})};
	    }
	    ]
	    \Gamma\ar[d,"\Delta"] &[-2.5cm]\Gamma\oplus\Gamma & {(\Gamma\oplus\Gamma)\oplus\Gamma} \\
	    \Gamma\oplus\Gamma\ar[r,"\Delta\oplus\Delta"]&{(\Gamma\oplus\Gamma)\oplus(\Gamma\oplus\Gamma)} & {(B\oplus B)\oplus C} \\
	    &{(B\oplus C)\oplus(B\oplus C)} & {(B\oplus C)}
	    \arrow["\Delta", from=1-1, to=1-2]
	    \arrow["\Delta", from=1-2, to=2-2,dashed]
	    \arrow["{(f_1\oplus g)\oplus(f_2\oplus g)}", from=2-2, to=3-2]
	    \arrow["\Delta\oplus\Id", from=1-2, to=1-3]
	    \arrow["{(f_1\oplus f_2)\oplus g}", from=1-3, to=2-3]
	    \arrow["{\sumwe{\semS{\escalar p}}{\semS{\escalar q}}\oplus\Id}", from=2-3, to=3-3]
	    \arrow["{\sumwe{\semS{\escalar p}}{\semS{\escalar q}}}", from=3-2, to=3-3]
	    \arrow["\delta", dashed, from=2-3, to=3-2]
	    \arrow["\delta", dashed, from=1-3, to=2-2]
	\end{tikzcd}
      \]

    \item 
      If $K[]=\pair s{K'[]}$. This case is analogous to the case $\pair{K'[]}s$.

    \item
      If $K[]=\fst(K'[])$.
      Let 
      \begin{align*}
	f_1&=\sem{\Gamma\vdash K'[r_1]:A\with B} & f_2 &=\sem{\Gamma\vdash K'[r_2]:A\with B}
      \end{align*}
      Then,
      \begin{align*}
	\sem{\Gamma\vdash K[t]:A}
	&=\sem{\Gamma\vdash\fst(K'[t]):A}\\
	&=\pi_1\circ\sem{\Gamma\vdash K'[t]:A\with B}\\
	\textrm{(by IH)}&=\pi_1\circ\sumwe{\semS{\escalar p}}{\semS{\escalar q}}\circ(f_1\oplus f_2)\circ\Delta\\
	\textrm{\color{red}(*)}&
	=\sumwe{\semS{\escalar p}}{\semS{\escalar q}}\circ( \pi_1\oplus \pi_1)\circ(f_1\oplus f_2)\circ\Delta\\
	& =\sumwe{\semS{\escalar p}}{\semS{\escalar q}}\circ(
	(\pi_1\circ f_1)
	\oplus
	(\pi_1\circ f_2)
	)\circ\Delta\\
	&=\sumwe{\semS{\escalar p}}{\semS{\escalar q}}\circ(\sem{\Gamma\vdash \fst(K'[r_1]):A}\oplus\sem{\Gamma\vdash\fst(K'[r_2]):A})\circ\Delta\\
	&=\sumwe{\semS{\escalar p}}{\semS{\escalar q}}\circ(\sem{\Gamma\vdash K[r_1]:A}\oplus\sem{\Gamma\vdash K[r_2]:A})\circ\Delta
      \end{align*}
      Where the equality {\color{red}(*)} is justified by the following commuting diagram.
      \[
	\begin{tikzcd}[labels=description,column sep=1cm,row sep=8mm,
	    execute at end picture={
	      \path
	      (\tikzcdmatrixname-1-3) -- (\tikzcdmatrixname-2-3) coordinate[pos=0.5](aux1)
	      (\tikzcdmatrixname-1-4) -- (\tikzcdmatrixname-2-4) coordinate[pos=0.5](aux2)
	      (aux1) -- (aux2) node[midway,red]{\small ({Lemma~\ref{lem:sumwe-nat}})};
	    }
	    ]
	    \Gamma\ar[r,"\Delta"] & \Gamma\oplus\Gamma\ar[r,"f_1\oplus f_2"] &[5mm]    ( A\oplus  B)\oplus( A\oplus  B)\ar[r,"\pi_1\oplus\pi_1"]\ar[d,"\sumwe{\semS{\escalar p}}{\semS{\escalar q}}"]&[1cm]  A\oplus  A\ar[d,"\sumwe{\semS{\escalar p}}{\semS{\escalar q}}"]\\
	    & &     A\oplus  B\ar[r,"\pi_1"]&  A
	\end{tikzcd}
      \]

    \item
      If $K[]=\snd(K'[])$. This case is analogous to the case $\fst(K'[])$.

    \item 
      If $K[]=\inl(K'[])$, then $A=B\oplus C$.
      Let 
      \begin{align*}
	f_1&=\sem{\Gamma\vdash K'[r_1]:B} &^f_2&=\sem{\Gamma\vdash K'[r_2]:B}
      \end{align*}
      Then,
      \begin{align*}
	\sem{\Gamma\vdash K[t]:A}
	&=\sem{\Gamma\vdash\inl(K'[t]):B\oplus C}\\
	&=i_1\circ\sem{\Gamma\vdash K'[t]:B}\\
	\textrm{(by IH)}&=i_1\circ\sumwe{\semS{\escalar p}}{\semS{\escalar q}}\circ(f_1\oplus f_2)\circ\Delta\\
	\textrm{\color{red}(*)}
	&=\sumwe{\semS{\escalar p}}{\semS{\escalar q}}\circ( i_1 \oplus i_1)\circ(f_1\oplus f_2)\circ\Delta\\
	&=\sumwe{\semS{\escalar p}}{\semS{\escalar q}}\circ(
	(i_1\circ f_1)
	\oplus
	(i_1\circ f_2)
	)\circ\Delta\\
	&=\sumwe{\semS{\escalar p}}{\semS{\escalar q}}\circ(\sem{\Gamma\vdash i_1(K'[r_1]):B}\oplus\sem{\Gamma\vdash i_1(K'[r_2]):B})\circ\Delta\\
	&=\sumwe{\semS{\escalar p}}{\semS{\escalar q}}\circ(\sem{\Gamma\vdash K[r_1]:A}\oplus\sem{\Gamma\vdash K[r_2]:A})\circ\Delta
      \end{align*}
      Where the equality {\color{red}(*)} is justified by the following commuting diagram.
      \[
	\begin{tikzcd}[labels=description,column sep=3cm,
	    execute at end picture={
	      \path
	      (\tikzcdmatrixname-2-1) -- (\tikzcdmatrixname-3-1) coordinate[pos=0.5](aux1)
	      (\tikzcdmatrixname-2-2) -- (\tikzcdmatrixname-3-2) coordinate[pos=0.5](aux2)
	      (aux1) -- (aux2) node[midway,red]{\small ({Lemma~\ref{lem:sumwe-nat}})};
	    }
	    ]
	    \Gamma\oplus\Gamma\ar[d,"f_1\oplus f_2"] & \Gamma\ar[l,"\Delta"]\\
	    B\oplus  B\ar[r,"i_1\oplus i_1"]\ar[d,"\sumwe{\semS{\escalar p}}{\semS{\escalar q}}"]& ( B\oplus  C)\oplus ( B\oplus  C)\ar[d,"\sumwe{\semS{\escalar p}}{\semS{\escalar q}}"]\\
	    B\oplus  B\ar[r,"i_1"]&  B\oplus  C
	\end{tikzcd}
      \]

    \item If $K[]=\inr(K'[])$. This case is analogous to the case $\inl(K'[])$.

    \item 
      If $K[]=\elimoplus(K'[],x.u_1,y.u_2)$.
      Then $\Gamma=\Gamma_1,\Gamma_2$.
      Let 
      \begin{align*}
	f_1&=\sem{\Gamma_1\vdash K'[r_1]:B\oplus C} &
	f_2&=\sem{\Gamma_1\vdash K'[r_2]:B\oplus C}\\
	g_1&=\sem{x:B,\Gamma_2\vdash u_1:A} &
	g_2&=\sem{y:C,\Gamma_2\vdash u_2:A} 
      \end{align*}
      Then,
      \begin{align*}
	\sem{\Gamma\vdash K[t]:A}
	&=\sem{\Gamma_1,\Gamma_2\vdash\elimoplus(K'[t],x.u_1,y.u_2):A}\\
	&=\coproducto{g_1}{g_2}\circ d\circ(\sem{\Gamma\vdash K'[t]:B\oplus C}\otimes\Id)\\
	\textrm{(by IH)}
	&=\coproducto{g_1}{g_2}\circ d\circ( (\sumwe{\semS{\escalar p}}{\semS{\escalar q}}\circ(f_1\oplus f_2)\circ\Delta) \otimes\Id)\\
	\textrm{\color{red}(*)}&=\sumwe{\semS{\escalar p}}{\semS{\escalar q}}\circ(
	(\coproducto{g_1}{g_2}\circ d\circ (f_1\otimes\Id))
	\oplus
	(\coproducto{g_1}{g_2}\circ d\circ (f_2\otimes\Id))
	)\circ\Delta\\
	&=
	\begin{aligned}[t]
	  \sumwe{\semS{\escalar p}}{\semS{\escalar q}}\circ(&\sem{\Gamma_1,\Gamma_2\vdash\elimoplus(K'[r_1],x.u_1,y.u_2):A}\\
	  \oplus&\sem{\Gamma_1,\Gamma_2\vdash\elimoplus(K'[r_2],x.u_1,y.u_2):A})\circ\Delta
	\end{aligned}
	\\
	&=\sumwe{\semS{\escalar p}}{\semS{\escalar q}}\circ(\sem{\Gamma\vdash K[r_1]:A}\oplus\sem{\Gamma\vdash K[r_2]:A})\circ\Delta
      \end{align*}
      Where the equality {\color{red}(*)} is justified by the following commuting diagram.
      \[
	\begin{tikzcd}[labels=description,column sep=2.5cm,row sep=1cm,
	    execute at end picture={
	      \path (\tikzcdmatrixname-2-1) -- (\tikzcdmatrixname-1-2) node[midway,red,yshift=4mm,xshift=-4mm,sloped]{\small (Corollary~\ref{cor:PropsFinv})};
	      \path 
	      (\tikzcdmatrixname-2-1) -- (\tikzcdmatrixname-3-1) coordinate[pos=0.5](aux1)
	      (\tikzcdmatrixname-1-2) -- (\tikzcdmatrixname-2-2) coordinate[pos=0.5](aux2)
	      (aux1) -- (aux2) node[midway,red,sloped]{\small (Lemma~\ref{lem:distrib})};
	      \path (\tikzcdmatrixname-4-1) -- (\tikzcdmatrixname-2-2) node[midway,red,yshift=4mm,xshift=-4mm,sloped]{\small ({Lemma \ref{lem:PropsF}})};
	      \path (\tikzcdmatrixname-4-1) -- (\tikzcdmatrixname-4-2) node[midway,red]{\small ({Lemma~\ref{lem:sumwe-nat}})};
	    }
	    ]
	    {\Gamma_1}\otimes{\Gamma_2}\ar[d,"\Delta\otimes\Id"]\ar[r,"\Delta"] & ({\Gamma_1}\otimes{\Gamma_2})\oplus({\Gamma_1}\otimes{\Gamma_2})\ar[d,"(f_1\otimes\Id)\oplus(f_2\otimes\Id)"]\\
	    ({\Gamma_1}\oplus{\Gamma_1})\otimes{\Gamma_2}\ar[d,"(f_1\oplus f_2)\otimes\Id"]\ar[ur,"d",dashed,sloped]& (( B\oplus  C)\otimes{\Gamma_2})\oplus(( B\oplus  C)\otimes{\Gamma_2})\ar[d,"d\oplus d"]\ar[ddl,"\sumwe{\semS{\escalar p}}{\semS{\escalar q}}",sloped,dashed]\\
	    (( B\oplus  C)\oplus( B\oplus  C))\otimes{\Gamma_2}\ar[d,"\sumwe{\semS{\escalar p}}{\semS{\escalar q}}\otimes\Id"]\ar[ur,"d",dashed,sloped]& ( B\otimes{\Gamma_2}\oplus  C\otimes{\Gamma_2})\oplus( B\otimes{\Gamma_2}\oplus  C\otimes{\Gamma_2})\ar[d,"\coproducto{g_1}{g_2}\oplus\coproducto{g_1}{g_2}"]\\
	    ( B\oplus  C)\otimes{\Gamma_2}\ar[d,"d"] &  A\oplus  A\ar[d,"\sumwe{\semS{\escalar p}}{\semS{\escalar q}}"]\\
	    B\otimes{\Gamma_2}\oplus  C\otimes{\Gamma_2}\ar[r,"\coproducto{g_1}{g_2}"]&  A
	\end{tikzcd}
      \]

    \item 
      If $K[]=\elimoplus(v,x.K'[],y.u_2)$.
      Then $\Gamma=\Gamma_1,\Gamma_2$.
      Let 
      \begin{align*}
	f_1&=\sem{x:B,\Gamma_2\vdash K'[r_1]:A} &
	f_2&=\sem{x:B,\Gamma_2\vdash K'[r_2]:A} \\
	g&=\sem{\Gamma_1\vdash v:B\oplus C} &
	h&=\sem{y:C,\Gamma_2\vdash u_2:A}
      \end{align*}
      Then,
      \begin{align*}
	\sem{\Gamma\vdash K[t]:A}
	&=\sem{\Gamma_1,\Gamma_2\vdash\elimoplus(v,x.K'[t],y.u_2):A}\\
	&=\coproducto{\sem{x:B,\Gamma_2\vdash K'[t]:A}}{h}\circ d\circ(g\otimes\Id)\\
	\textrm{(by IH)}
	&=\coproducto{\sumwe{\semS{\escalar p}}{\semS{\escalar q}}\circ(f_1\oplus f_2)\circ\Delta}{h}\circ d\circ( g \otimes\Id)\\
	\textrm{\color{red}(*)}&=\sumwe{\semS{\escalar p}}{\semS{\escalar q}}\circ(
	(\coproducto{f_1}{h}\circ d\circ (g\otimes\Id))
	\oplus
	(\coproducto{f_2}{h}\circ d\circ (g\otimes\Id))
	)\circ\Delta\\
	&=
	\begin{aligned}[t]
	  \sumwe{\semS{\escalar p}}{\semS{\escalar q}}\circ(&\sem{\Gamma_1,\Gamma_2\vdash\elimoplus(v,x.K'[r_1],y.u_2):A}\\
	  \oplus&\sem{\Gamma_1,\Gamma_2\vdash\elimoplus(v,x.K'[r_2],y.u_2):A})\circ\Delta
	\end{aligned}
	\\
	&=\sumwe{\semS{\escalar p}}{\semS{\escalar q}}\circ(\sem{\Gamma\vdash K[r_1]:A}\oplus\sem{\Gamma\vdash K[r_2]:A})\circ\Delta
      \end{align*}
      Where the equality {\color{red}(*)} is justified by the following commuting diagram.
      \[
	\begin{tikzcd}[labels=description,column sep=2cm,row sep=1.5cm,
	    execute at end picture={
	      \path 
	      (\tikzcdmatrixname-1-1) -- (\tikzcdmatrixname-1-2) coordinate[pos=0.5](aux1)
	      (\tikzcdmatrixname-2-1) -- (\tikzcdmatrixname-2-2) coordinate[pos=0.5](aux2)
	      (aux1) -- (aux2) node[midway,red]{\small (Naturality of $\Delta$)};
	      \path 
	      (\tikzcdmatrixname-3-1) -- (\tikzcdmatrixname-3-2) coordinate[pos=0.5](aux1)
	      (aux2) -- (aux1) node[midway,red]{\small (Naturality of $\Delta$)};
	      \path 
	      (\tikzcdmatrixname-4-1) -- (\tikzcdmatrixname-4-2) coordinate[pos=0.5](aux2)
	      (aux1) -- (aux2) node[midway,red]{\small (Naturality of $\Delta$)};
	      \path 
	      (\tikzcdmatrixname-5-1) -- (\tikzcdmatrixname-5-2) coordinate[pos=0.5](aux1)
	      (aux2) -- (aux1) node[midway,red]{\small (**)};
	    }
	    ]
	    {\Gamma_1}\otimes {\Gamma_2}\ar[d,"g\otimes\Id"]\ar[r,"\Delta"] & ( {\Gamma_1}\otimes {\Gamma_2})\oplus( {\Gamma_1}\otimes {\Gamma_2})\ar[d,"(g\otimes\Id)\oplus(g\otimes\Id)"]\\
	    ( I\oplus  I)\otimes {\Gamma_2}\ar[d,"d"]\ar[r,"\Delta",dashed]& (( I\oplus  I)\otimes {\Gamma_2})\oplus(( I\oplus  I)\otimes {\Gamma_2})\ar[d,"d\oplus d"]\\
	    I\otimes {\Gamma_2}\oplus  I\otimes {\Gamma_2}\ar[r,"\Delta",dashed]\ar[d,"\lambda\oplus \lambda"]& ( I\otimes {\Gamma_2}\oplus  I\otimes {\Gamma_2})\oplus( I\otimes {\Gamma_2}\oplus  I\otimes {\Gamma_2})\ar[d,"\lambda\oplus \lambda"]\\
	    \Gamma_2\oplus \Gamma_2\ar[r,"\Delta",dashed]\ar[d,"\coproducto{\sumwe{\semS{\escalar p}}{\semS{\escalar q}}\circ(f_1\oplus f_2)\circ\Delta}{h}"] &(\Gamma_2\oplus \Gamma_2)\oplus(\Gamma_2\oplus \Gamma_2)\ar[d,"\coproducto{f_1}{h}\oplus\coproducto{f_2}{h}"]\\
	    A &   A\oplus   A\ar[l,"\sumwe{\semS{\escalar p}}{\semS{\escalar q}}"] 
	\end{tikzcd}
      \]
      The commutation of the diagram {\color{red}(**)} is justified as follows.
      \begin{align*}
	\coproducto{\sumwe{\semS{\escalar p}}{\semS{\escalar q}}\circ(f_1\oplus f_2)\circ\Delta}h
	&=\coproducto{\sumwe{\semS{\escalar p}}{\semS{\escalar q}}\circ(f_1\oplus f_2)\circ\Delta}{\Id\circ h}\\
	\textrm{(Lemma~\ref{lem:nablaPQInvDelta})}&=\coproducto{\sumwe{\semS{\escalar p}}{\semS{\escalar q}}\circ(f_1\oplus f_2)\circ\Delta}{(\sumwe{\semS{\escalar p}}{\semS{\escalar q}}\circ\Delta)\circ h}\\
	&=\sumwe{\semS{\escalar p}}{\semS{\escalar q}}\circ\coproducto{(f_1\oplus f_2)\circ\Delta}{\Delta\circ h}\\
	&=\sumwe{\semS{\escalar p}}{\semS{\escalar q}}\circ\nabla\circ((f_1\oplus f_2)\circ\Delta)\oplus(\Delta\circ h)\\
	\textrm{(Naturality of $\Delta$)}&=\sumwe{\semS{\escalar p}}{\semS{\escalar q}}\circ\nabla\circ((f_1\oplus f_2)\circ\Delta)\oplus((h\oplus h)\circ\Delta)\\
	&=\sumwe{\semS{\escalar p}}{\semS{\escalar q}}\circ\nabla\circ(f_1\oplus f_2\oplus h\oplus h)\circ(\Delta\oplus\Delta)\\
	\textrm{\color{red}(***)}&=(\nabla\oplus\nabla)\circ(f_1\oplus h\oplus f_2\oplus h)\circ\Delta\\
	&=((\nabla\circ(f_1\oplus h))\oplus(\nabla\circ(f_2\oplus h)))\circ\Delta\\
	&=\sumwe{\semS{\escalar p}}{\semS{\escalar q}}\circ(\coproducto {f_1}h\oplus\coproducto{f_2}h)\circ\Delta
      \end{align*}
      Where the equality {\color{red}(***)} is justified by the following commuting diagram, using the fact that $\sumwe{1}{1}=\nabla$.
      \[
	\begin{tikzcd}[labels=description,column sep=2cm,row sep=1.5cm,
	    execute at end picture={
	      \path 
	      (\tikzcdmatrixname-1-1) -- (\tikzcdmatrixname-2-1) coordinate[pos=0.5](aux1)
	      (\tikzcdmatrixname-1-2) -- (\tikzcdmatrixname-2-2) coordinate[pos=0.5](aux2)
	      (\tikzcdmatrixname-2-1) -- (\tikzcdmatrixname-3-1) coordinate[pos=0.5](aux3)
	      (\tikzcdmatrixname-2-2) -- (\tikzcdmatrixname-3-2) coordinate[pos=0.5](aux4)
	      (aux1) -- (\tikzcdmatrixname-1-2) node[midway,red,sloped]{\small (Lemma~\ref{lem:DeltaSigma})}
	      (aux3) -- (aux2) node[midway,red,sloped]{\small (Naturality of $\sigma$)}
	      (\tikzcdmatrixname-3-1) -- (aux4) node[midway,red,sloped]{\small (Lemma~\ref{lem:PropsSumwe})};
	    }
	    ]
	    \Gamma\oplus\Gamma\ar[r,"\Delta\oplus\Delta"]\ar[d,"\Delta"] & \Gamma\oplus\Gamma\oplus\Gamma\oplus\Gamma\ar[d,"f_1\oplus f_2\oplus h\oplus h"]\\
	    \Gamma\oplus\Gamma\oplus\Gamma\oplus\Gamma\ar[ur,"\Id\oplus\sigma\oplus\Id",dashed,sloped]\ar[d,"f_1\oplus h\oplus f_2\oplus h"] & A\oplus A\oplus A\oplus A\ar[d,"\nabla"]\\
	    A\oplus A\oplus A\oplus A\ar[ur,"\Id\oplus\sigma\oplus\Id",dashed,sloped]\ar[r,"\nabla\oplus\nabla"] & A\oplus A
	\end{tikzcd}
      \]

    \item 
      If $K[]=\elimoplus(v,x.u_1,y.K'[])$
      This case is analogous to the case $\elimoplus(v,x.K'[],y.u_2)$.

    \item
      If $K[]=\super{K'[]}s$, then $A=B\odot C$. This case is identical to the case $\pair{K'[]}s$.

    \item
      If $K[]=\super s{K'[]}$. This case is identical to the case $\pair s{K'[]}$.

    \item
      If $K[]=\fstsup(K'[])$. This case is identical to the case $\fst(K'[])$.

    \item
      If $K[]=\sndsup(K'[])$. This case is identical to the case $\snd(K'[])$.

    \item 
      If $K[]=\elimsup[\escalar p'\escalar q'](K'[],x.u_1,y.u_2)$.

      Then $\Gamma=\Gamma_1,\Gamma_2$.
      Let 
      \begin{align*}
	f_1&=\sem{\Gamma_1\vdash K'[r_1]:B\odot C} & f_2&=\sem{\Gamma_1\vdash K'[r_2]:B\odot C}\\
	g_1&=\sem{x:B,\Gamma_2\vdash u_1:A} & g_2&=\sem{y:C,\Gamma_2\vdash u_2:A}
      \end{align*}
      Then,
      \begin{align*}
	\sem{\Gamma\vdash K[t]:A}
	&=\sem{\Gamma_1,\Gamma_2\vdash\elimsup[\escalar p'\escalar q'](K'[t],x.u_1,y.u_2):A}\\
	&=\sumwe{\semS{\escalar p'}}{\semS{\escalar q'}}\circ({g_1}\oplus{g_2})\circ d\circ(\sem{\Gamma\vdash K'[t]:B\oplus C}\otimes\Id)\\
	\textrm{(by IH)}
	&=\sumwe{\semS{\escalar p'}}{\semS{\escalar q'}}\circ({g_1}\oplus{g_2})\circ d\circ( (\sumwe{\semS{\escalar p}}{\semS{\escalar q}}\circ(f_1\oplus f_2)\circ\Delta) \otimes\Id)\\
	\textrm{\color{red}(*)}&
	=
	\begin{aligned}[t]
	  \sumwe{\semS{\escalar p}}{\semS{\escalar q}}\circ(
	  &(\sumwe{\semS{\escalar p'}}{\semS{\escalar q'}}\circ({g_1}\oplus{g_2})\circ d\circ (f_1\otimes\Id))\\
	  \oplus&
	  (\sumwe{\semS{\escalar p'}}{\semS{\escalar q'}}\circ({g_1}\oplus{g_2})\circ d\circ (f_2\otimes\Id))
	  )\circ\Delta
	\end{aligned}
	\\
	&=
	\begin{aligned}[t]
	  \sumwe{\semS{\escalar p}}{\semS{\escalar q}}\circ(&\sem{\Gamma_1,\Gamma_2\vdash\elimsup[\escalar p'\escalar q'](K'[r_1],x.u_1,y.u_2):A}\\
	  \oplus&\sem{\Gamma_1,\Gamma_2\vdash\elimsup[\escalar p'\escalar q'](K'[r_2],x.u_1,y.u_2):A})\circ\Delta
	\end{aligned}
	\\
	&=\sumwe{\semS{\escalar p}}{\semS{\escalar q}}\circ(\sem{\Gamma\vdash K[r_1]:A}\oplus\sem{\Gamma\vdash K[r_2]:A})\circ\Delta
      \end{align*}
      Where the equality {\color{red}(*)} is justified by the following commuting diagram.
      \[
	\hspace{-8mm}\begin{tikzcd}[labels=description,column sep=2cm,row sep=1cm,
	    execute at end picture={
	      \path (\tikzcdmatrixname-2-1) -- (\tikzcdmatrixname-1-2) node[midway,red,yshift=4mm,xshift=-4mm,sloped]{\small (Corollary~\ref{cor:PropsFinv})};
	      \path 
	      (\tikzcdmatrixname-2-1) -- (\tikzcdmatrixname-3-1) coordinate[pos=0.5](aux1)
	      (\tikzcdmatrixname-1-2) -- (\tikzcdmatrixname-2-2) coordinate[pos=0.5](aux2)
	      (aux1) -- (aux2) node[midway,red,sloped]{\small (Lemma~\ref{lem:distrib})};
	      \path (\tikzcdmatrixname-4-1) -- (\tikzcdmatrixname-2-2) node[midway,red,yshift=4mm,xshift=-4mm,sloped]{\small ({Lemma \ref{lem:PropsF}})};
	      \path (\tikzcdmatrixname-4-2) -- (\tikzcdmatrixname-6-2) coordinate[pos=0.5](aux)
	      (aux) -- (\tikzcdmatrixname-6-1) node[midway,red,sloped]{\small ({Lemma~\ref{lem:sumwe-nat}})};
	      \path 
	      (\tikzcdmatrixname-4-1) -- (\tikzcdmatrixname-5-1) coordinate[pos=0.5](aux1)
	      (\tikzcdmatrixname-2-2) -- (\tikzcdmatrixname-3-2) coordinate[pos=0.5](aux2)
	      (aux1) -- (aux2) node[midway,red,sloped]{\small ({Lemma~\ref{lem:sumwe-nat}})};
	      \path 
	      (\tikzcdmatrixname-5-1) -- (\tikzcdmatrixname-6-1) coordinate[pos=0.5](aux1)
	      (\tikzcdmatrixname-3-2) -- (\tikzcdmatrixname-4-2) coordinate[pos=0.5](aux2)
	      (aux1) -- (aux2) node[midway,red,sloped]{\small ({Lemma~\ref{lem:sumwe-nat}})};
	    }
	    ]
	    {\Gamma_1}\otimes{\Gamma_2}\ar[d,"\Delta\otimes\Id"]\ar[r,"\Delta"] & ({\Gamma_1}\otimes{\Gamma_2})\oplus({\Gamma_1}\otimes{\Gamma_2})\ar[d,"(f_1\otimes\Id)\oplus(f_2\otimes\Id)"]\\
	    ({\Gamma_1}\oplus{\Gamma_1})\otimes{\Gamma_2}\ar[d,"(f_1\oplus f_2)\otimes\Id"]\ar[ur,"d",dashed,sloped]& (( B\oplus C)\otimes{\Gamma_2})\oplus(( B\oplus C)\otimes{\Gamma_2})\ar[d,"d\oplus d"]\ar[ddl,"\sumwe{\semS{\escalar p}}{\semS{\escalar q}}",sloped,dashed]\\
	    (( B\oplus C)\oplus( B\oplus C))\otimes{\Gamma_2}\ar[d,"\sumwe{\semS{\escalar p}}{\semS{\escalar q}}\otimes\Id"]\ar[ur,"d",dashed,sloped]& ( B\otimes{\Gamma_2}\oplus C\otimes{\Gamma_2})\oplus( B\otimes{\Gamma_2}\oplus C\otimes{\Gamma_2})\ar[d,"(g_1\oplus g_2)\oplus(g_1\oplus g_2)"]\ar[ddl,"\sumwe{\semS{\escalar p}}{\semS{\escalar q}}",dashed]\\
	    ( B\oplus C)\otimes{\Gamma_2}\ar[d,"d"] &  (A\oplus A)\oplus(A\oplus A)\ar[d,"\sumwe{\semS{\escalar p'}}{\semS{\escalar q'}}\oplus\sumwe{\semS{\escalar p'}}{\semS{\escalar q'}}"]\ar[ddl,"\sumwe{\semS{\escalar p}}{\semS{\escalar q}}",dashed]\\
	    B\otimes{\Gamma_2}\oplus C\otimes{\Gamma_2}\ar[d,"g_1\oplus g_2"] & A\oplus  A\ar[d,"\sumwe{\semS{\escalar p}}{\semS{\escalar q}}"]\\
	    A\oplus A\ar[r,"\sumwe{\semS{\escalar p'}}{\semS{\escalar q'}}"] &  A
	\end{tikzcd}
      \]

    \item
      If $K[]=\elimsup[\escalar p'\escalar q'](s,x.K'[],y.u_2)$.

      Then $\Gamma=\Gamma_1,\Gamma_2$.
      Let 
      \begin{align*}
	f_1&=\sem{x:B,\Gamma_2\vdash K'[r_1]:A} & f_2&=\sem{x:B,\Gamma_2\vdash K'[r_2]:A}\\
	g&=\sem{\Gamma_1\vdash v:B\odot C} & h&=\sem{y:C,\Gamma_2\vdash u_2:A}
      \end{align*}
      \begin{align*}
	\sem{\Gamma\vdash K[t]:A}
	&=\sem{\Gamma_1,\Gamma_2\vdash\elimsup[\escalar p'\escalar q'](v,x.K'[t],y.u_2):A}\\
	&=\sumwe{\semS{\escalar p'}}{\semS{\escalar q'}}\circ((\sem{x:B,\Gamma_2\vdash K'[t]:A})\oplus h)\circ d\circ(g\otimes\Id)\\
	\textrm{(by IH)}
	&=\sumwe{\semS{\escalar p'}}{\semS{\escalar q'}}\circ(({\sumwe{\semS{\escalar p}}{\semS{\escalar q}}\circ(f_1\oplus f_2)\circ\Delta})\oplus{h})\circ d\circ( g \otimes\Id)\\
	\textrm{\color{red}(*)}&=
	\begin{aligned}[t]
	  \sumwe{\semS{\escalar p}}{\semS{\escalar q}}\circ(
	  & (\sumwe{\semS{\escalar p'}}{\semS{\escalar q'}}\circ({f_1}\oplus{h})\circ d\circ (g\otimes\Id))\\
	  \oplus &
	  (\sumwe{\semS{\escalar p'}}{\semS{\escalar q'}}\circ({f_2}\oplus{h})\circ d\circ (g\otimes\Id))
	  )\circ\Delta
	\end{aligned}
	\\
	&=
	\begin{aligned}[t]
	  \sumwe{\semS{\escalar p}}{\semS{\escalar q}}\circ(&\sem{\Gamma_1,\Gamma_2\vdash\elimsup[\escalar p'\escalar q'](v,x.K'[r_1],y.u_2):A}\\
	  \oplus&\sem{\Gamma_1,\Gamma_2\vdash\elimsup[\escalar p'\escalar q'](v,x.K'[r_2],y.u_2):A})\circ\Delta
	\end{aligned}
	\\
	&=\sumwe{\semS{\escalar p}}{\semS{\escalar q}}\circ(\sem{\Gamma\vdash K[r_1]:A}\oplus\sem{\Gamma\vdash K[r_2]:A})\circ\Delta
      \end{align*}
      Where the equality {\color{red}(*)} is justified by the following commuting diagram.
      \[
	\hspace{-7mm}\begin{tikzcd}[labels=description,column sep=8mm,row sep=1cm,
	    execute at end picture={
	      \path (\tikzcdmatrixname-1-1) -- (\tikzcdmatrixname-1-2) coordinate[pos=0.5](aux1)
	      (\tikzcdmatrixname-2-1) -- (\tikzcdmatrixname-2-2) coordinate[pos=0.5](aux2)
	      (aux1) -- (aux2) node[midway,red]{\small (Naturality of $\Delta$)};
	      \path (\tikzcdmatrixname-3-1) -- (\tikzcdmatrixname-3-2) coordinate[pos=0.5](aux1)
	      (aux2) -- (aux1) node[midway,red]{\small (Naturality of $\Delta$)};
	      \path (\tikzcdmatrixname-4-1) -- (\tikzcdmatrixname-3-2) coordinate[pos=0.5](aux)
	      (aux) -- (\tikzcdmatrixname-3-1) node[midway,red] {\small (Lemma~\ref{lem:deltaDelta})};
	      \path (\tikzcdmatrixname-4-1) -- (\tikzcdmatrixname-5-1) coordinate[pos=0.5](aux1)
	      (\tikzcdmatrixname-3-2) -- (\tikzcdmatrixname-4-2) coordinate[pos=0.5](aux)
	      (aux) -- (aux1) node[midway,red,sloped,yshift=-3mm]{\small (Lemma~\ref{lem:delta-nat})};
	      \path (\tikzcdmatrixname-5-1) -- (\tikzcdmatrixname-5-2) node[midway,red]{\small Lemma~\ref{lem:iteration-sumwe}};
	    }
	    ]
	    {\Gamma_1}\otimes{\Gamma_2}\ar[d,"g\otimes\Id"]\ar[r,"\Delta"] & ({\Gamma_1}\otimes{\Gamma_2})\oplus({\Gamma_1}\otimes{\Gamma_2})\ar[d,"(g\otimes\Id)\oplus(g\otimes\Id)"]\\
	    ( B\oplus  C)\otimes{\Gamma_2}\ar[d,"d"]\ar[r,"\Delta",dashed]& (( B\oplus  C)\otimes{\Gamma_2})\oplus(( B\oplus  C)\otimes{\Gamma_2})\ar[d,"d\oplus d"]\\
	    ( B\otimes{\Gamma_2})\oplus ( C\otimes{\Gamma_2})\ar[r,"\Delta",dashed]\ar[d,"\Delta\oplus\Id"]& ( B\otimes{\Gamma_2}\oplus  C\otimes{\Gamma_2})\oplus( B\otimes{\Gamma_2}\oplus  C\otimes{\Gamma_2})\ar[d,"({f_1}\oplus{h})\oplus({f_2}\oplus{h})"]\\
	    ( B\otimes{\Gamma_2}\oplus  B\otimes{\Gamma_2})\oplus ( C\otimes{\Gamma_2})\ar[d,"({f_1\oplus f_2})\oplus{h}"]\ar[ur,"\delta",sloped,dashed] & ( A\oplus  A)\oplus( A\oplus  A)\ar[d,"\sumwe{\semS{\escalar p'}}{\semS{\escalar q'}}\oplus\sumwe{\semS{\escalar p'}}{\semS{\escalar q'}}"] \\
	    ( A\oplus  A)\oplus  A\ar[ur,"\delta",dashed,sloped]\ar[d,"\sumwe{\semS{\escalar p}}{\semS{\escalar q}}\oplus\Id"] &  A\oplus  A\ar[d,"\sumwe{\semS{\escalar p}}{\semS{\escalar q}}"]\\
	    A\oplus  A\ar[r,"\sumwe{\semS{\escalar p'}}{\semS{\escalar q'}}"]&  A
	\end{tikzcd}
      \]

    \item 
      If $K[]=\elimsup[\escalar p'\escalar q'](v,x.u_1,y.K'[])$.
      This is analogous to the case $\elimsup[\escalar p'\escalar q'](v,x.K'[],y.u_2)$.
      \qed
  \end{itemize}
\end{proof}

\end{document}